\documentclass[letterpaper,11pt,onecolumn,draftcls]{IEEEtran}  % Comment this line out
\usepackage{epsfig} % for postscript graphics files
\usepackage{amsmath} % assumes amsmath package installed
\usepackage{amssymb}  % assumes amsmath package installed

\usepackage{cite}      % Written by Donald Arseneau
\usepackage{subfigure} % Written by Steven Douglas Cochran
\usepackage{url}       % Written by Donald Arseneau
\usepackage{verbatim}   % multi-line comment: \begin{comment},\end{comment}

\usepackage{algorithm}
\usepackage{algpseudocode}

\usepackage[dvips]{color}
\newcommand{\high}[1]{{\color{black}{#1}}}

\date{\today}

\begin{document}

\title{%\LARGE \bf
Throughput-optimal Scheduling in Multi-hop Wireless Networks
without Per-flow Information}

%\author{ \parbox{3 in}{\centering Huibert Kwakernaak*
%         \thanks{*Use the $\backslash$thanks command to put information here}\\
%         Faculty of Electrical Engineering, Mathematics and Computer Science\\
%         University of Twente\\
%         7500 AE Enschede, The Netherlands\\
%         {\tt\small h.kwakernaak@autsubmit.com}}
%         \hspace*{ 0.5 in}
%         \parbox{3 in}{ \centering Pradeep Misra**
%         \thanks{**The footnote marks may be inserted manually}\\
%        Department of Electrical Engineering \\
%         Wright State University\\
%         Dayton, OH 45435, USA\\
%         {\tt\small pmisra@cs.wright.edu}}
%}
\author{Bo~Ji, Changhee~Joo, and Ness~B.~Shroff%
\thanks{B. Ji is with Department of ECE at the 
Ohio State University. 
C. Joo is with School of ECE at UNIST, Korea.
N. B. Shroff is with Departments of ECE 
and CSE at the Ohio State University.}%
\thanks{Emails: ji@ece.osu.edu, cjoo@unist.ac.kr, shroff@ece.osu.edu.}
% \thanks{This work was supported in part by ARO MURI Award W911NF-08-1-0238, 
% and NSF Awards 1012700-CNS, 0721236-CNS, and 0721434-CNS.}
\thanks{A preliminary version of this work was presented at IEEE WiOpt 2011.}
}%

\newcounter{theorem}
\newtheorem{theorem}{Theorem}
\newtheorem{proposition}[theorem]{\it Proposition}
\newtheorem{lemma}[theorem]{\it Lemma}
\newtheorem{corollary}[theorem]{\it Corollary}
\newcounter{definition}
\newtheorem{definition}{\it Definition}

\newcommand{\Matching}{\mathcal{M}}
\newcommand{\mH}{\mathcal{H}}
\newcommand{\mS}{\mathcal{S}}
\newcommand{\mP}{\mathcal{P}}
\newcommand{\mF}{\mathcal{F}}
\newcommand{\mB}{\mathcal{B}}
\newcommand{\mC}{\mathcal{C}}
\newcommand{\Graph}{\mathcal{G}}
\newcommand{\Vertex}{\mathcal{V}}
\newcommand{\Edge}{\mathcal{E}}
\newcommand{\Flow}{\mathcal{S}}
\newcommand{\Route}{\mathcal{L}}
\newcommand{\Path}{P}
\newcommand{\Pathset}{\mathcal{P}}
\newcommand{\Tree}{\mathcal{T}}
\newcommand{\Cycle}{\mathcal{C}}
\newcommand{\Component}{\mathcal{Z}}
\newcommand{\Rank}{\Gamma}
\newcommand{\Class}{\Xi}
\newcommand{\Queue}{\mathcal{Q}}
\newcommand{\Int}{\mathbb{Z}}
\newcommand{\Expect}{\operatorname{E}}
\newcommand{\Var}{\operatorname{Var}}
\newcommand{\System}{\mathcal{Y}}
\newcommand{\Markov}{\mathcal{X}}
\newcommand{\vM}{\vec{M}}
\newcommand{\bM}{\mathbf{M}}
\newcommand{\s}{{(s)}}
\newcommand{\sk}{{s,k}}
\newcommand{\ls}{{l,s}}
\newcommand{\ik}{{i,k}}
\newcommand{\lk}{{l,k}}
\newcommand{\rp}{{r^{\prime}}}
\newcommand{\Hslk}{{H^s_{l,k}}}
\newcommand{\hsk}{{\hat{s},\hat{k}}}
\newcommand{\rj}{{r,j}}
\newcommand{\xn}{{x_n}}
\newcommand{\xnj}{{x_{n_j}}}
\newcommand{\Cong}{{W}}
\newcommand{\cs}{{w}}
\newcommand{\UGraph}{{U}}
\newcommand{\UVertex}{{X}}
\newcommand{\UEdge}{{Y}}
\newcommand{\vlambda}{\vec{\lambda}}
\newcommand{\vpi}{\vec{\pi}}
\newcommand{\vphi}{\vec{\phi}}
\newcommand{\vpsi}{\vec{\psi}}
\newcommand{\blambda}{\mathbf{\lambda}}
\newcommand{\tM}{\tilde{M}}
\newcommand{\tpi}{\tilde{\pi}}
\newcommand{\tphi}{\tilde{\phi}}
\newcommand{\tpsi}{\tilde{\psi}}
\newcommand{\vmu}{\vec{\mu}}
\newcommand{\vnu}{\vec{\nu}}
\newcommand{\valpha}{\vec{\alpha}}
\newcommand{\vbeta}{\vec{\beta}}
\newcommand{\ve}{\vec{e}}
\newcommand{\vxi}{\vec{\xi}}
\newcommand{\hgamma}{\gamma}
\newcommand{\LINK}{\mathcal{E}}
\newcommand{\GMSLambda}{\Lambda_{\text{\it GMS}}}
\newcommand{\argmax}{\operatornamewithlimits{argmax}}
\newcommand{\card}{\aleph}
\newcommand{\V}{\mathcal{V}}
\newcommand{\X}{\mathcal{X}}
\newcommand{\Y}{\mathcal{Y}}

\maketitle
%\thispagestyle{empty}
%\pagestyle{empty}

%%%%%%%%%%%%%%%%
\begin{abstract}
%%%%%%%%%%%%%%%%
\high{In this paper, we consider the problem of link scheduling in 
multi-hop wireless networks under general interference constraints. 
Our goal is to design scheduling schemes that do not use per-flow 
or per-destination information, maintain a single data queue for 
each link, and exploit only local information, while guaranteeing
throughput optimality.} 
Although the celebrated back-pressure algorithm maximizes throughput, 
it requires per-flow or per-destination information. It is usually
difficult to obtain and maintain this type of information, especially 
in large networks, where there are numerous flows. Also, the back-pressure
algorithm maintains a complex data structure at each node, keeps exchanging
queue length information among neighboring nodes, and commonly results 
in poor delay performance. 
\high{In this paper, we propose scheduling schemes that can circumvent 
these drawbacks and guarantee throughput optimality. These schemes use 
either the readily available hop-count information or only the local 
information for each link.} 
We rigorously analyze 
the performance of the proposed schemes using fluid limit techniques via 
an inductive argument and show that they are throughput-optimal. We also 
conduct simulations to validate our theoretical results in various 
settings, and show that the proposed schemes can substantially improve 
the delay performance in most scenarios. 
\end{abstract}

%%%%%%%%%%%%%%%%%%%%%%%%%%%%%%%%%%%%%%
\section{Introduction} \label{sec:int}
%%%%%%%%%%%%%%%%%%%%%%%%%%%%%%%%%%%%%%
Link scheduling is a critical resource allocation functionality 
in multi-hop wireless networks, and also perhaps the most 
challenging. The seminal work of \cite{tassiulas92} introduces 
a joint adaptive routing and scheduling algorithm, called 
back-pressure, that has been shown to be throughput-optimal, 
i.e., it can stabilize the network under any feasible load. 
This paper focuses on the settings with fixed routes, where the
back-pressure algorithm becomes a scheduling algorithm 
consisting of two components: flow scheduling and link scheduling. 
The back-pressure algorithm calculates the weight of a link as 
the product of the link capacity and the maximum ``back-pressure" 
(i.e., the queue length difference between the queues at the
transmitting nodes of this link and the next hop link for each flow) 
among all the flows passing through the link, and solves a MaxWeight 
problem to activate a set of non-interfering links that have the 
largest weight sum. The flow with the maximum queue length difference 
at a link is chosen to transmit packets when the link is activated.

\high{
The back-pressure algorithm, although throughput-optimal, 
needs to solve a MaxWeight problem, which requires centralized 
operations and is NP-hard in general \cite{sharma06}. To this end,
simple scheduling algorithms based on Carrier Sensing Multiple 
Access (CSMA) \cite{jiang10, ni09, rajagopalan09} are developed
to achieve the optimal throughput in a distributed manner for
single-hop traffic, and are later extended to the case of multi-hop 
traffic \cite{jiang10} leveraging the basic idea of back-pressure.
}

% Since the development of the back-pressure algorithm, there have 
% been numerous variations that have integrated them into an overall 
% optimal cross-layer solution. However, the MaxWeight problem that 
% the back-pressure-type of algorithms need to solve requires centralized 
% operations, and is NP-hard under general interference constraints 
% \cite{sharma06}.

% Recently, exciting advances have been made in developing simple, 
% distributed and throughput-optimal algorithms based on Carrier 
% Sensing Multiple Access (CSMA) \cite{jiang10, ni09, rajagopalan09}. 
% The key idea of these CSMA-based algorithms is that each link 
% adaptively adjusts its own channel attempt probability using 
% local queue length information (or locally measured arrival and 
% service rates). The solutions were \emph{originally designed for 
% single-hop traffic \cite{jiang10, ni09, rajagopalan09}, and 
% extended to the case of multi-hop traffic \cite{jiang10} using 
% the idea of back-pressure.}

However, the back-pressure-type of scheduling algorithms (including
CSMA for multi-hop traffic) have the following shortcomings: 1) 
require per-flow or per-destination information, which is usually 
difficult to obtain and maintain, especially in large networks 
where there are numerous flows, 2) need to maintain 
separate queues for each flow or destination at each node, 3) rely 
on extensive exchange of queue length information among neighboring 
nodes to calculate link weights, which becomes the major obstacle to 
their distributed implementation, and 4) may result in poor overall 
delay performance, as the queue length needs to build up (creating the 
back-pressure) from a flow destination to its source, which leads to 
large queues along the route a flow takes \cite{bui11,stolyar11}.
An important question is whether one can circumvent the above drawbacks 
of the back-pressure-type of algorithms and design throughput-optimal
scheduling algorithms that do not require per-flow or per-destination 
information, maintain a small number of data queues (ideally, a single 
data queue for each link), exploit only local information when making 
scheduling decisions, and potentially have good delay performance. 

There have been some recent studies (e.g., \cite{bui11, liu10c, ying08, 
ying09}) in this direction. A cluster-based back-pressure algorithm 
that can reduce the number of queues is proposed in \cite{ying08}, where 
nodes (or routers) are grouped into clusters and each node needs only to 
maintain separate queues for destinations within its cluster. In 
\cite{bui11}, the authors propose a back-pressure policy making scheduling 
decisions in a shadow layer (where counters are used as per-flow shadow 
queues). Their scheme only needs to maintain a single \emph{First-In 
First-Out (FIFO)} queue instead of per-flow queues for each link and 
shows dramatic improvement in the delay performance. However, their 
shadow algorithm still requires per-flow information and constant 
exchange of shadow queue length information among neighboring nodes. 
The work in \cite{liu10c} proposes to exploit the local queue length 
information to design throughput-optimal scheduling algorithms. Their 
approach combined with CSMA algorithms can achieve fully distributed 
scheduling without any information exchange. Their scheme is based on 
a two-stage queue structure, where each node maintains two types of 
data queues: per-flow queues and per-link queues. The two-stage queue 
structure imposes additional complexity, and is similar to queues with 
regulators \cite{wu07}, which have been empirically noted to have very 
large delays. In \cite{ying09}, the authors propose a back-pressure 
algorithm that integrates the shortest path routing to minimize the average 
number of hops between each source and destination pair. However, their 
scheme further increases the number of queues by maintaining a separate 
queue $\{i,d,k\}$ at each node $i$ for the packets that will be delivered 
to destination node $d$ within $k$ hops. 

Although these algorithms partly alleviate the effect of the 
aforementioned disadvantages of the traditional back-pressure 
algorithms, to the best of our knowledge, no work has addressed
all the aforementioned four issues. In particular, a critical 
drawback of the earlier mentioned works is that they require 
\emph{per-flow or per-destination information} to guarantee 
throughput optimality. In this paper, we propose a class of 
throughput-optimal schemes that can remove this per-flow or 
per-destination information requirement, maintain a single 
data queue for each link, and remove information exchange. 
As a by-product, these proposed schemes also improve the delay 
performance in a variety of scenarios. 

The main contributions of our paper are as follows.

First, we propose a scheduling scheme with \emph{per-hop} queues
to address the four key issues mentioned earlier. The proposed 
scheme maintains multiple FIFO queues $Q_\lk$ at the transmitting 
node of each link $l$. Specifically, any packet whose transmission 
over link $l$ is the $k$-th hop forwarding from its source node
is stored at queue $Q_\lk$.
This hop-count information is much easier to obtain and maintain 
compared to per-flow or per-destination information. For example, 
hop-count information can be obtained using \emph{Time-To-Live} or 
\emph{TTL} information in packet headers. Moreover, as mentioned 
earlier, while the number of flows in a large network is very large, 
the number of hops is often much smaller. 
In the Internet, the longest 
route a flow takes typically has tens of hops\footnote{In the Routing 
Information Protocol (RIP) \cite{rfc2453}, the longest route is limited 
to 15 hops. In general, an upper bound on the length of a route is 255 
hops in the Internet, as specified by TTL in the Internet Protocol (IP) 
\cite{rfc791}.}, while there are billions of users or nodes 
\cite{internet11} and thus the number of flows could be extremely large. 
A shadow algorithm similar to \cite{bui11} is adopted in our framework, 
where a shadow queue is associated with each data queue. We consider 
the MaxWeight algorithm based on shadow queue lengths, and show that 
this per-Hop-Queue-based MaxWeight Scheduler (HQ-MWS) is throughput-optimal 
using fluid limit techniques via a hop-by-hop inductive argument. For 
illustration, in this paper, we focus on the centralized MaxWeight-type 
of policies. However, one can readily extend our approach to a large 
class of scheduling policies (where fluid limit techniques can be used). 
For example, combining our approach with the CSMA-based algorithms of 
\cite{jiang10,ni09,rajagopalan09}, one can completely remove the 
requirement of queue length information exchange, and develop fully 
distributed scheduling schemes, under which no information exchange 
is required. \emph{To the best of our knowledge, this is the first 
work that develops throughput-optimal scheduling schemes without 
per-flow or per-destination information in wireless networks with 
multi-hop traffic.} In addition, we believe that using this type of 
per-hop queue structure to study the problem of link scheduling is of 
independent interest.

Second, we have also developed schemes with per-link queues (i.e., a 
single data queue per link) instead of per-hop queues, extending the 
idea to per-Link-Queue-based MaxWeight Scheduler (LQ-MWS). We propose 
two schemes based on LQ-MWS using different queueing disciplines. We 
first combine it with the \emph{priority} queueing discipline (called 
PLQ-MWS), where a higher priority is given to the packet that traverses 
a smaller number of hops, and then prove throughput optimality of PLQ-MWS. 
It is of independent interest that this type of hop-count-based priority 
discipline enforces stability. This, however, requires that nodes sort 
packets according to their hop-count information. We then remove this 
restriction by combining LQ-MWS with the FIFO queueing discipline (called 
FLQ-MWS), and prove throughput optimality of FLQ-MWS in networks where 
flows do not form loops. 
% Tech Report
We further propose fully distributed heuristic algorithms by combining 
our approach with the CSMA algorithms, and show that the fully distributed 
CSMA-based algorithms are throughput-optimal under the time-scale separation 
assumption.

Finally, we show through simulations that the proposed schemes can 
significantly improve the delay performance in most scenarios.
In addition, the schemes with per-link queues (PLQ-MWS and FLQ-MWS)
perform well in a wider variety of scenarios, which implies that 
maintaining per-link queues not only simplifies the data structure, 
but also can contribute to scheduling efficiency and delay performance.

The remainder of the paper is organized as follows. In Section~\ref{sec:mod}, 
we present a detailed description of our system model. In Section~\ref{sec:hq}, 
we prove throughput optimality of HQ-MWS using fluid limit techniques via a 
hop-by-hop inductive argument. We extend our ideas to show throughput-optimality
of PLQ-MWS and FLQ-MWS in Section~\ref{sec:lq}. Further, we show that our approach
combined with the CSMA-based algorithms leads to fully distributed scheduling 
schemes in Section~\ref{sec:csma}. We evaluate different scheduling schemes 
through simulations in Section~\ref{sec:sim}. Finally, we conclude our paper 
in Section~\ref{sec:con}.

%%%%%%%%%%%%%%%%%%%%%%%%%%%%%%%%%%%%%%
\section{System Model} \label{sec:mod}
%%%%%%%%%%%%%%%%%%%%%%%%%%%%%%%%%%%%%%
We consider a multi-hop wireless network described by a directed graph 
$\Graph=(\Vertex,\Edge)$, where $\Vertex$ denotes the set of nodes and 
$\Edge$ denotes the set of links. Nodes are wireless transmitters/receivers 
and links are wireless channels between two nodes if they can directly 
communicate with each other. Let $b(l)$ and $e(l)$ denote the transmitting 
node and receiving node of link $l=(b(l),e(l)) \in \Edge$, respectively. 
Note that we distinguish links $(i,j)$ and $(j,i)$. We assume a time-slotted 
system with a single frequency channel. Let $c_l$ denote the link capacity 
of link $l$, i.e., link $l$ can transmit at most $c_l$ packets during a 
time slot if none of the links that interfere with $l$ is transmitting
at the same time. We assume unit capacity links, i.e., $c_l=1$ for all 
$l \in \Edge$. A flow is a stream of packets from a source node to a 
destination node. Packets are injected at the source, and traverse multiple 
links to the destination via multi-hop communications. Let $\Flow$ denote 
the set of flows in the network. We assume that each flow $s$ has a single, 
fixed, and loop-free route that is denoted by $\Route(s)=(l^s_1, \cdots, 
l^s_{|\Route(s)|})$, where the route of flow $s$ has $|\Route(s)|$ 
hop-length from the source to the destination, $l^s_k$ denotes the $k$-th 
hop link on the route of flow $s$, and $|\cdot|$ denotes the cardinality 
of a set. Let $L^{\max} \triangleq \max_{s} |\Route(s)| < \infty$ denote 
the length of the longest route over all flows. Let $\Hslk \in \{0,1\}$ be 
1, if link $l$ is the $k$-th hop link on the route of flow $s$, and 0, otherwise. 
Note that the assumption of single route and unit capacity is only for ease of 
exposition, and one can readily extend the results to more general scenarios 
with \emph{multiple fixed routes and heterogeneous link capacities}, applying
the techniques used in this paper. We also restrict our attention to those 
links that have flows passing through them. Hence, without loss of generality, 
we assume that $\sum_s \sum^{|\Route(s)|}_{k=1} \Hslk \ge 1$, for all $l \in \Edge$. 

The interference set of link $l$ is defined as $I(l) \triangleq \{ j \in \Edge 
~|~ \text{link}~j ~\text{interferes with link}~ l \}$. We consider a general 
interference model, where the interference is symmetric, i.e., for any $l, j 
\in \Edge$, if $l \in I(j)$, then $j \in I(l)$. A \emph{schedule} is a set of
(active or inactive) links, and can be represented by a vector $M \in 
\{0,1\}^{|\Edge|}$, where component $M_l$ is set to 1 if link $l$ is active, 
and 0 if it is inactive. A schedule $M$ is said to be \emph{feasible} if no two 
links of $M$ interfere with each other, i.e., $l \notin I(j)$ for all $l$, $j$ 
with $M_l = 1$ and $M_j = 1$. Let $\Matching$ denote the set of all feasible 
schedules over $\Edge$, and let $Co(\Matching)$ denote its convex hull.

Let $F_s(t)$ denote the cumulative number of packet arrivals at the source 
node of flow $s$ up to time slot $t$. We assume that packets are of unit 
length. \high{We assume that each arrival process $F_s(t)-F_s(t-1)$ is an 
irreducible positive recurrent Markov chain with countable state space, and 
satisfies the Strong Law of Large Numbers (SLLN): That is, with probability one,
\begin{equation}
\label{eq:slln}
\textstyle \lim_{t \rightarrow \infty} \frac {F_s(t)} {t} = \lambda_s,
\end{equation}
for each flow $s \in \Flow$, where $\lambda_s$ is the mean arrival rate of
flow $s$. We let $\lambda \triangleq [\lambda_s]$ denote the arrival rate 
vector. Also, we assume that the arrival processes are mutually independent 
across flows. (This assumption can be relaxed as in \cite{andrews04}.)} 

As in \cite{andrews04}, a stochastic queueing network is said to be 
\emph{stable}, if it can be described as a discrete-time countable
Markov chain and the Markov chain is \emph{stable} in the following 
sense: The set of positive recurrent states is non-empty, and it 
contains a finite subset such that with probability one, this subset 
is reached within finite time from any initial state. When all the 
states communicate, stability is equivalent to the Markov chain being 
positive recurrent \cite{bramson08}. We define the \emph{throughput 
region} of a scheduling policy as the set of arrival rate vectors for 
which the network is stable under this policy. Further, we define the 
\emph{optimal throughput region} (or \emph{stability region}) as the 
union of the throughput regions of all possible scheduling policies, 
including the offline policies \cite{tassiulas92}. We denote by $\Lambda^{*}$ 
the optimal throughput region, whereby $\Lambda^{*}$ can be represented as
\begin{equation}
\label{eq:otr}
\begin{split}
\textstyle  \Lambda^* \triangleq \{ \lambda ~|~ & \text{for some}~ \phi \in Co(\Matching), \\
& \textstyle \sum_s \sum_k \Hslk \lambda_s \le \phi_l ~\text{for all links}~ l \in \Edge \}.
% \textstyle  \Lambda^* \triangleq \{ \lambda ~|~ \exists \phi \in Co(\Matching),
% \sum_s \sum_k \Hslk \lambda_s \le \phi_l, \forall l \in \Edge \}.
\end{split}
\end{equation}
An arrival rate vector is strictly inside $\Lambda^{*}$, if the inequalities 
above are all strict.

Throughout the paper, we let $(z)^+ \triangleq \max(z,0)$ denote the larger
value between $z$ and 0.

%%%%%%%%%%%%%%%%%%%%%%%%%%%%%%%%%%%%%%%%%%%%%%%%%%%%%%%%%%%%%%%%%%%%%%%%%%%%%
\section{Scheduling with Per-hop Queues} \label{sec:hq}
%%%%%%%%%%%%%%%%%%%%%%%%%%%%%%%%%%%%%%%%%%%%%%%%%%%%%%%%%%%%%%%%%%%%%%%%%%%%%
In this section, we propose scheduling policies with per-hop queues and shadow 
algorithm. We will later extend our ideas to developing schemes with per-link 
queues in Section~\ref{sec:lq}. We describe our scheduling schemes using the 
centralized MaxWeight algorithm for ease of presentation. 
% Tech Report
Our approach combined 
with the CSMA algorithms can be extended to develop fully distributed scheduling
algorithms in Section~\ref{sec:csma}.

%%%%%%%%%%%%%%%%%%%%%%%%%%%%%%%%%%%%%%%%%%%%%%%%%%%%%%%%%%%%%%%%%%%%%%%%%
\subsection{Queue Structure and Scheduling Algorithm} \label{subsec:desc}
%%%%%%%%%%%%%%%%%%%%%%%%%%%%%%%%%%%%%%%%%%%%%%%%%%%%%%%%%%%%%%%%%%%%%%%%%
We start with the description of queue structure, and then specify our scheduling
scheme based on per-hop queues and a shadow algorithm. We assume that, at the 
transmitting node of each link $l$, a single FIFO data queue $Q_\lk$ is maintained 
for packets whose $k$-th hop is link $l$, where $1 \le k \le L^{max}$. Such queues 
are called \emph{per-hop} queues. For notational convenience, we also use $Q_\lk(t)$ 
to denote the queue length of $Q_\lk$ at time slot $t$. Let $\Pi_\lk(t)$ denote the 
service of $Q_\lk$ at time slot $t$, which takes a value of $c_l$ (i.e., 1 in our setting), 
if queue $Q_\lk$ is active, or $0$, otherwise. Let $D_\lk(t)$ denote the cumulative 
number of packet departures from queue $Q_\lk$ up to time slot $t$, and let $\Psi_\lk(t) 
\triangleq D_\lk(t) - D_\lk(t-1)$ be the number of packet departures from queue $Q_\lk$ 
at time slot $t$. Since a queue may be empty when it is scheduled, we have $\Psi_\lk(t) 
\le \Pi_\lk(t)$ for all time slots $t \ge 0$. Let $U_\sk(t)$ denote the cumulative 
number of packets transmitted from the $(k-1)$-st hop to the $k$-th hop for flow $s$ 
up to time slot $t$ for $1 \le k \le \Route(s)$, where we set $U_{s,1}(t) = F_s(t)$. 
And let $A_\lk(t)$ be the cumulative number of aggregate packet arrivals (including 
both exogenous arrivals and arrivals from the previous hops) at queue $Q_\lk$ up to 
time slot $t$. Then, we have $A_\lk(t) = \sum_{s} \Hslk U_\sk(t)$, and in particular, 
$A_{l,1}(t) = \sum_{s} H^s_{l,1} F_s(t)$. Let $P_\lk(t) \triangleq A_\lk(t) - A_\lk(t-1)$ 
denote the number of arrivals for queue $Q_\lk$ at time slot $t$. We adopt the convention 
that $A_\lk(0) = 0$ and $D_\lk(0)=0$ for all $l \in \Edge$ and $1 \le k \le L^{\max}$. 
The queue length evolves as
\begin{equation} 
\label{eq:q_evolution}
\textstyle  Q_\lk(t) = Q_\lk(0) + A_\lk(t) - D_\lk(t).
\end{equation}

For each data queue $Q_\lk$, we maintain a shadow queue $\hat{Q}_\lk$, and let $\hat{Q}_\lk(t)$ 
denote its queue length at time slot $t$. The arrival and departure processes of the shadow 
queues are controlled as follows. We denote by $\hat{A}_\lk(t)$ and $\hat{D}_\lk(t)$ its 
cumulative amount of arrivals and departures up to time slot $t$, respectively. Also, let 
$\hat{\Pi}_\lk(t)$, $\hat{P}_\lk(t) \triangleq \hat{A}_\lk(t) - \hat{A}_\lk(t-1)$ and 
$\hat{\Psi}_\lk(t) \triangleq \hat{D}_\lk(t) - \hat{D}_\lk(t-1)$ denote the amount of service, 
arrivals and departures of queue $\hat{Q}_\lk$ at time slot $t$, respectively. Likewise, we 
have $\hat{\Psi}_\lk(t) \le \hat{\Pi}_\lk(t)$ for $t \ge 0$. We set by convention that, 
$\hat{A}_\lk(0)=0$ and $\hat{D}_\lk(0) = 0$ for all queues $\hat{Q}_\lk$. The arrivals 
for shadow queue $\hat{Q}_\lk$ are set to $(1+\epsilon)$ times the average amount of 
packet arrivals at data queue $Q_\lk$ up to time slot $t$, i.e.,
\begin{equation}
\label{eq:shadowarr}
\textstyle \hat{P}_\lk(t) = (1+\epsilon) \frac {A_\lk(t)} {t}, 
\end{equation}
where $\epsilon > 0$ is a sufficiently small positive number such that 
$(1+\epsilon) \lambda$ is also strictly inside $\Lambda^*$ given that $\lambda$ is strictly 
inside $\Lambda^*$. Then, the shadow queue length evolves as
\begin{equation} 
\label{eq:shadowB_evolution}
\textstyle  \hat{Q}_\lk(t) = \hat{Q}_\lk(0) + \hat{A}_\lk(t) - \hat{D}_\lk(t).
\end{equation}

Using these shadow queues, we determine the service of both data queues and shadow 
queues using the following MaxWeight algorithm. 

\noindent {\bf Per-Hop-Queue-based MaxWeight Scheduler (HQ-MWS):} At each 
time slot $t$, the scheduler serves data queues $Q_{l,k^*(l)}$ for $l \in M^*$,
where 
\begin{eqnarray}
&& \textstyle k^*(l) \in \argmax_{k} \hat{Q}_\lk(t), ~\text{for each link}~ l \in \Edge, \label{eq:ks} \\
&& \textstyle M^* \in \argmax_{M \in \Matching} \sum_{l\in\Edge} \hat{Q}_{l,k^*(l)}(t) \cdot M_l. \label{eq:mw}
\end{eqnarray}
In other words, we set the service of data queue as $\Pi_\lk(t)=1$ if $l\in M^*$ 
and $k=k^*(l)$, and $\Pi_\lk(t)=0$ otherwise. We also set the service of shadow 
queues as $\hat{\Pi}_\lk(t) = \Pi_\lk(t)$ for all $l$ and $k$. 

\emph{Remark:} The algorithm needs to solve a MaxWeight problem based on the shadow 
queue lengths, and ties can be broken arbitrarily if there is more than one queue 
having the largest shadow queue length at a link or there is more than one schedule 
having the largest weight sum. Note that we have $\Pi_\lk(t) = \hat{\Pi}_\lk(t)$ under 
this scheduling scheme, for all links $l \in \Edge$ and $1 \le k \le L^{max}$ and for 
all time slots $t \ge 0$. Once a schedule $M^*$ is selected, data queues $Q_{l,k^*(l)}$ 
for links $l$ with $M^*_l = 1$ are activated to transmit packets if they are non-empty, 
and shadow queues $\hat{Q}_{l,k^*(l)}$ ``transmit" shadow packets as well. Note that 
shadow queues are just counters. The arrival and departure process of a shadow queue 
is simply an operation of addition and subtraction, respectively.

%%%%%%%%%%%%%%%%%%%%%%%%%%%%%%%%%%
\subsection{Throughput Optimality}
%%%%%%%%%%%%%%%%%%%%%%%%%%%%%%%%%%
We present the main result of this section as follows.

\begin{proposition}
\label{pro:hq}
HQ-MWS is throughput-optimal, i.e., the network is stable under HQ-MWS 
for any arrival rate vector $\lambda$ strictly inside $\Lambda^{*}$.
\end{proposition}

We prove the stability of the network in the sense that the underlying 
Markov chain (whose state accounts for both data queues and shadow 
queues; see Appendix~\ref{app:pro:hq} for the detailed state description) 
is stable under HQ-MWS, using fluid limit techniques \cite{andrews04,dai95}. 
We provide the proof of Proposition~\ref{pro:hq} in Appendix~\ref{app:pro:hq}, 
and discuss the outline of the proof as follows. 

Note that the shadow queues serve only single-hop traffic, i.e., after
packets in the shadow queues are served, they leave the system without 
being transmitted to another shadow queue. We also emphasize that the 
single-hop shadow traffic gets smoothed under the arrival process of 
(\ref{eq:shadowarr}), and in the fluid limits (which will be formally 
established in Appendix~\ref{app:pro:hq}), after a finite time, the 
instantaneous shadow arrival rate is strictly inside the optimal 
throughput region $\Lambda^*$ with small enough $\epsilon>0$. Then, 
using the standard Lyapunov approach, we can show the stability for 
the sub-system consisting of shadow queues. 

Now, we consider the data queues in the fluid limits starting from the first 
hop data queue for each link $l\in\Edge$. Since the arrival process of data 
queue $Q_{l,1}$ satisfies the SLLN, the instantaneous arrival of shadow queue 
$\hat{Q}_{l,1}$ will be equal to $(1+\epsilon) \sum_s H^s_{l,1} \lambda_s$. 
This implies that the service rate of shadow queue $\hat{Q}_{l,1}$ is no smaller 
than $(1+\epsilon) \sum_s H^s_{l,1} \lambda_s$ due to the stability of shadow 
queues. Then, the service rate of data queue $Q_{l,1}$ is also no smaller than 
$(1+\epsilon) \sum_s H^s_{l,1} \lambda_s$ because $\Pi_\lk(t)=\hat{\Pi}_\lk(t)$ 
under HQ-MWS. Since the arrival rate of data queue $Q_{l,1}$ is $\sum_s H^s_{l,1} 
\lambda_s$, the service rate is strictly greater than the arrival rate for 
$Q_{l,1}$, establishing its stability. Using this as an induction base, we can 
show the stability of data queues via a hop-by-hop inductive argument. This 
immediately implies that the fluid limit model of the joint system is stable 
under HQ-MWS.

Although our proposed scheme shares similarities with \cite{bui11,liu10c}, it has important
differences. First, in \cite{bui11}, per-flow information is still required by their 
shadow algorithm. The shadow packets are injected into the network at the sources, 
and are then ``transmitted" to the destinations via multi-hop communications. Their 
scheme strongly relies on the information exchange of shadow queue lengths to calculate 
the link weights. In contrast, we take a different approach for constructing the instantaneous 
arrivals at each shadow queue according to (\ref{eq:shadowarr}) that is based on the 
average amount of packet arrivals at the corresponding data queue. This method of 
injecting shadow packets allows us to decompose multi-hop traffic into single-hop 
traffic for shadow queues and exploit only local information when making scheduling 
decisions. Second, although the basic idea behind the shadow arrival process of 
(\ref{eq:shadowarr}) is similar to the service process of the per-flow queues in 
\cite{liu10c}, the scheme in \cite{liu10c} requires per-flow information and relies 
on a two-stage queue architecture that consists of both per-flow and per-link data 
queues. In contrast, our scheme needs only per-hop (and not per-flow) information, 
i.e., the number of hops each packet has traversed, completely removing per-flow 
information and per-flow queues. This simplification of required information and data 
structure is critical, due to the fact that the maximum number of hops in a network is 
usually much smaller than the number of flows in a large network. For example, in the 
Internet, the longest route a flow takes typically has tens of hops, while there are 
billions of nodes and thus the number of flows could be extremely large. 

Note that the hop-count in our approach is counted from the source. Such per-hop 
information is easy to obtain (e.g., from \emph{Time-to-Live} or \emph{TTL} information 
in the Internet and ad hoc networks). At each link, packets with the same hop-count 
(from the source of each packet to the link) are kept at the same queue, regardless of 
sources, destinations, and flows, which significantly reduces the number of queues. In 
Section~\ref{sec:lq}, we extend our approach to the schemes with per-link queues, and 
further remove even the requirement of per-hop information.

%%%%%%%%%%%%%%%%%%%%%%%%%%%%%%%%%%%%%%%%%%%%%%%%%%%%%%%%
\section{Scheduling with Per-link Queues} \label{sec:lq}
%%%%%%%%%%%%%%%%%%%%%%%%%%%%%%%%%%%%%%%%%%%%%%%%%%%%%%%%
In the previous section, we show that per-hop-queue-based MaxWeight 
scheduler (HQ-MWS) achieves optimal throughput performance. In this 
section, we extend our ideas to developing schemes with \emph{per-link} 
queues. To elaborate, we show that per-link-queue-based MaxWeight
scheduler, when associated with \emph{priority} or \emph{FIFO} 
queueing discipline, can also achieve throughput optimality.

%%%%%%%%%%%%%%%%%%%%%%%%%%%%%%%%%%%%%%%%%%%%%%%%%%%%%%%%%%%%%%%%%%%%%%%%%%%
\subsection{MaxWeight Algorithm with Per-link Queues} \label{subsec:lq-mws}
%%%%%%%%%%%%%%%%%%%%%%%%%%%%%%%%%%%%%%%%%%%%%%%%%%%%%%%%%%%%%%%%%%%%%%%%%%%
We consider a network where each link $l$ has a single data queue $Q_l$. Let 
$Q_l(t)$, $A_l(t)$, $D_l(t)$, $\Pi_l(t)$, $\Psi_l(t)$ and $P_l(t)$ denote the 
queue length, cumulative arrival, cumulative departure, service, departure and 
arrival at the data queue $Q_l$, respectively. Also, we maintain a shadow queue 
$\hat{Q}_l$ associated with each $Q_l$, and let $\hat{Q}_l(t)$, $\hat{A}_l(t)$, 
$\hat{D}_l(t)$, $\hat{\Pi}_l(t)$, $\hat{\Psi}_l(t)$ and $\hat{P}_l(t)$ denote 
the queue length, cumulative arrival, accumulative departure, service, departure 
and arrival at the shadow queue $\hat{Q}_l$, respectively. Similar to 
(\ref{eq:shadowarr}) for per-hop shadow queues, we control the arrivals to the 
shadow queue $\hat{Q}_l$ as
\begin{equation}
\label{eq:lq-shadowarr}
\textstyle \hat{P}_l(t) = (1+\epsilon) \frac {A_l(t)} {t}, 
\end{equation}
where $\epsilon > 0$ is a sufficiently small positive number.

Next, we specify the MaxWeight algorithm with per-link queues as follows.

\noindent {\bf Per-Link-Queue-based MaxWeight Scheduler (LQ-MWS):} At each 
time slot $t$, the scheduler serves links in $M^*$ (i.e., $\Pi_l(t)=1$ for 
$l\in M^*$, and $\Pi_l(t)=0$ otherwise), where 
%\begin{equation}
\[
\textstyle M^* \in \argmax_{M \in \Matching} \sum_{l \in \Edge} \hat{Q}_l(t) \cdot M_l.
\]
%\end{equation}
Also, we set the service of shadow queues as $\hat{\Pi}_l(t) = \Pi_l(t)$ for 
all $l$.

Similar as in HQ-MWS, the shadow traffic under LQ-MWS gets smoothed due to 
the shadow arrival assignment of (\ref{eq:lq-shadowarr}), and the instantaneous 
arrival rate of shadow queues can be shown to be strictly inside the optimal 
throughput region $\Lambda^*$. Hence, we show in Lemma~\ref{lem:sub-stable-lq}
(see Appendix~\ref{app:shadow-stab-lq}) that the fluid limit model for the 
sub-system consisting of shadow queues is stable under LQ-MWS, using the standard 
Lyapunov approach and following the same line of analysis for HQ-MWS.

%%%%%%%%%%%%%%%%%%%%%%%%%%%%%%%%%%%%%%%%%%%%%%%%%%%%%%%%%%%%%%%
\subsection{LQ-MWS with Priority Discipline} \label{subsec:plq}
%%%%%%%%%%%%%%%%%%%%%%%%%%%%%%%%%%%%%%%%%%%%%%%%%%%%%%%%%%%%%%%
We develop a scheduling scheme by combining LQ-MWS with priority 
queueing discipline, called \textbf{PLQ-MWS}. Regarding priority 
of packets at each per-link queue, we define \emph{hop-class} as 
follows: A packet has hop-class-$k$, if the link where the packet 
is located is the $k$-th hop from the source of the packet. When 
a link is activated to transmit packets, packets with a smaller 
hop-class will be transmitted earlier; and packets with the same 
hop-class will be transmitted in a FIFO fashion. 

\begin{proposition}
\label{pro:plq}
PLQ-MWS is throughput-optimal. 
\end{proposition}

We provide the outline of the proof and refer to Appendix~\ref{app:pro:plq} for 
the detailed proof. Basically, we follow the line of analysis 
for HQ-MWS using fluid limit techniques and induction method. Since a link 
transmits packets according to their priorities (i.e., hop-classes or hop-count 
from their respective source nodes), we can view packets with hop-class-$k$ at 
link $l$ as in a sub-queue $Q_{l,k}$ (similar to the per-hop queues under HQ-MWS).
Now, we consider the data queues in the fluid limits. Since the exogenous arrival 
process satisfies the SLLN, the instantaneous arrival to shadow queue $\hat{Q}_l$ 
will be at least $(1+\epsilon) \sum_s H^s_{l,1} \lambda_s$ for each link $l \in \Edge$. 
This implies that the service rate of shadow queue $\hat{Q}_l$ is no smaller than 
$(1+\epsilon) \sum_s H^s_{l,1} \lambda_s$ due to the stability of the shadow queues
(see Lemma~\ref{lem:sub-stable-lq} in Appendix~\ref{app:shadow-stab-lq}). 
Then, the service rate of sub-queue $Q_{l,1}$ is also no smaller than $(1+\epsilon) 
\sum_s H^s_{l,1} \lambda_s$, because: 1) $\Pi_l(t)=\hat{\Pi}_l(t)$ under PLQ-MWS; and 
2) the highest priority is given to sub-queue $Q_{l,1}$ when link $l$ is activated 
to transmit. Since the arrival rate of sub-queue $Q_{l,1}$ is $\sum_s H^s_{l,1} 
\lambda_s$, the service rate is strictly greater than the arrival rate for sub-queue
$Q_{l,1}$, establishing its stability. Similarly, we can show that the hop-class-$j$ 
sub-queues are stable for all $j \le k+1$, given the stability of the hop-class-$j^{\prime}$ 
sub-queues for all $j^{\prime} \le k$. Therefore, we can show the stability of the 
data queues via a hop-by-hop inductive argument. This immediately implies that the
fluid limit model of the joint system is stable under PLQ-MWS.

\high{
We emphasize that a ``bad" priority discipline may cause instability (even in 
wireline networks). See \cite{rybko92,lu91} for two simple counterexamples 
showing that in a two-station network, a static priority discipline that gives 
a higher priority to customers with a larger hop-count, may result in instability.  
(Interested readers are also referred to Chapter~3 of \cite{bramson08} for a good 
summary of the instability results.) The key intuition of these counterexamples 
is that, by giving a higher priority to packets with a larger hop-count in one 
station, the priority discipline may impede forwarding packets with a smaller 
hop-count to the next-hop station, which in turn starves the next-hop station. 
On the other hand, PLQ-MWS successfully eliminates this type of inefficiency by 
giving a higher priority to the packets with a smaller hop-count, and continues 
to push the packets to the following hops. 
}

Note that PLQ-MWS is different from HQ-MWS, although they appear 
to be similar. HQ-MWS makes scheduling decisions based on the queue 
length of each per-hop shadow queue. This may result in a waste of 
service if a per-hop queue is activated but does not have enough 
packets to transmit, even though the other per-hop queues of the 
same link have packets. In contrast, PLQ-MWS makes decisions based 
on the queue length of each per-link shadow queue and allows a link 
to transmit packets of multiple hop-classes, avoiding such an 
inefficiency. The performance difference due to this phenomenon will 
be illustrated through simulations in Section~\ref{sec:sim}. 
Furthermore, the implementation of PLQ-MWS is easier than HQ-MWS, 
since PLQ-MWS needs to maintain only one single shadow queue per link.

Another aspect of PLQ-MWS we would like to discuss
is about the \emph{hop-count-based priority discipline} in the 
context of multi-class queueing networks (or wireline networks). 
In operations research, stability of multi-class queueing networks 
has been extensively studied in the literature (e.g., see 
\cite{bramson08} and the references therein). To the best of our 
knowledge, however, there is very limited work on the topic of 
``priority enforces stability" \cite{chen00,chen02,chen96}. In 
\cite{chen00,chen02}, the authors obtained sufficient conditions
(based on linear or piecewise linear Lyapunov functions) for the 
stability of a multiclass fluid network and/or queueing network 
under priority disciplines. However, to verify these sufficient 
conditions relies on verifying the feasibility of a set of 
inequalities, which in general can be very difficult. The most 
related work to ours is \cite{chen96}. There, the authors showed
that under the condition of ``Acyclic Class Transfer", where 
customers can switch classes unless there is a loop in class
transfers, a simple priority discipline stabilizes the network 
under the usual traffic condition (i.e., the normalized load is 
less than one). 
Their priority discipline gives a higher priority to customers 
that are closer to their respective sources.

Interestingly, our hop-count-based priority discipline (for wireline networks) 
is similar to the discipline proposed in \cite{chen96}. However, there is a
major difference in that while \cite{chen96} studies stability of wireline 
networks (without link interferences) under the usual traffic condition, we 
consider stability of wireless networks with interference constraints that 
impose the (link) scheduling problem, which is much more challenging. In wireless 
networks, the service rate of each link depends on the underlying scheduling 
scheme, rather than being fixed as in wireline networks. Hence, the difficulty
is to establish the usual traffic condition by designing appropriate wireless
scheduling schemes. In this paper, we develop PLQ-MQS scheme and show that the 
usual traffic condition and then stability can be established via a hop-by-hop 
inductive argument under the PLQ-MWS scheme.

%%%%%%%%%%%%%%%%%%%%%%%%%%%%%%%%%%%%%%%%%%%%%%%%%%%%%%%%%%%
\subsection{LQ-MWS with FIFO Discipline} \label{subsec:flq}
%%%%%%%%%%%%%%%%%%%%%%%%%%%%%%%%%%%%%%%%%%%%%%%%%%%%%%%%%%%
In this section, we develop a scheduling scheme, called \textbf{FLQ-MWS},
by combining the LQ-MWS algorithm developed in Section~\ref{subsec:lq-mws} 
with \emph{FIFO} queueing discipline (instead of priority queueing discipline), 
and show that this scheme is throughput-optimal if flows do not form loops. 
We emphasize that FLQ-MWS requires neither per-flow information nor 
hop-count information. 

To begin with, we define a positive integer $r(l)$ as the rank of link 
$l \in \Edge$, and call $R(\Edge)=(r(l), l \in \Edge)$ a ranking of $\Edge$. 
Recall that $\Route(s)$ denotes the loop-free route of flow $s$. In the 
following, we prove a key property of the network where flows do not form 
loops, which will be used to prove the main results in this section.

\begin{lemma}
\label{lem:mon}
Consider a network $\Graph=(\Vertex,\Edge)$ with a set of flows $\Flow$, 
where the flows do not form loops. There exists a ranking $R(\Edge)$ such 
that the following two statements hold:
\begin{enumerate}
\item For any flow $s \in \Flow$, the ranks are monotonically increasing 
when one traverses the links of flow $s$ from $l^s_1$ to $l^s_{|\Route(s)|}$, 
i.e., $r(l^s_i) < r(l^s_{i+1})$ for all $1 \le i < |\Route(s)|$. 
\item The packet arrivals at a link are either exogenous, or forwarded 
from links with a smaller rank.
\end{enumerate}
\end{lemma}

We provide the proof of Lemma~\ref{lem:mon} in Appendix~\ref{app:lem:mon}.
Note that such a ranking with the monotone property exists because the flows 
do not form a loop. In contrast, it is clear that if flows form a loop, then 
such a ranking does not exist. Two examples of the networks where flows do 
not form loops are provided in Figs.~\ref{fig:ft8} and \ref{fig:ft9}, and an 
example of the network where flows do form a loop is provided in Fig.~\ref{fig:fl}.
Note that the ranking is only for the purpose of analysis and plays a key role 
in proving the system stability under FLQ-MWS, while it will not be used in the 
actual link scheduling algorithm.

Now, we give the main results of this section in the following proposition.

\begin{proposition}
\label{pro:flq}
FLQ-MWS is throughput-optimal in networks where \emph{flows do not form loops}.
\end{proposition}

We omit the detailed proof and refer to Appendix~\ref{app:lem:mon}.
In the following, we provide the outline of the proof. Motivated by Lemma~\ref{lem:mon}, 
we extend our analysis for HQ-MWS (or PLQ-MWS). Compared to the PLQ-MWS algorithm, 
there are differences only in the operations with data queues, and the underlying 
LQ-MWS algorithm remains the same. Thus, the shadow queues will exhibit similar 
behaviors, and the fluid limit model for the sub-system of shadow queues is stable 
under FLQ-MWS (see Lemma~\ref{lem:sub-stable-lq} in Appendix~\ref{app:shadow-stab-lq}). 
Also, note that Lemma~\ref{lem:mon} implies that given the qualified ranking 
(without loss of generality, assuming that the smallest rank is 1), the packet 
arrivals at links with rank 1 are all exogenous, then following a similar argument 
in the proof of Proposition~\ref{pro:hq}, we can prove the stability of the 
corresponding data queues by showing that the instantaneous arrival rate is less
than the instantaneous service rate. Since Lemma~\ref{lem:mon} also implies that 
the packet arrivals at links with rank~2 are either exogenous or from links with 
rank~1, we can similarly show the stability of links with rank~2. Repeating the 
above argument, we can prove the stability of all data queues by induction, which 
completes the proof of Proposition~\ref{pro:flq}. 

\begin{corollary}
\label{cor:flq}
FLQ-MWS is throughput-optimal in tree networks.
\end{corollary}

The above corollary follows immediately from Proposition~\ref{pro:flq}, 
because a tree network itself does not contain a cycle of links and flows
are all loop-free.

%%%%%%%%%%%%%%%%%%%%%%%%%%%%%%%%%%%%%%%%%%%%%%%%%%%%%%%%%%%%%%%%%%%%%%%
\section{Extension to CSMA-based Distributed Algorithms} \label{sec:csma}
%%%%%%%%%%%%%%%%%%%%%%%%%%%%%%%%%%%%%%%%%%%%%%%%%%%%%%%%%%%%%%%%%%%%%%%
In this section, we employ CSMA techniques to develop fully distributed 
throughput-optimal scheduling schemes for multi-hop traffic. We consider 
per-link-queue-based schemes combined with the CSMA-based scheduling of 
\cite{ni09}. 

%%%%%%%%%%%%%%%%%%%%%%%%%%%%%%%%%%%%%%%%%%%%%%%%%%%%%%%%%%%%%%
\subsection{Basic Scheduling Algorithm} \label{subsec:basic}
%%%%%%%%%%%%%%%%%%%%%%%%%%%%%%%%%%%%%%%%%%%%%%%%%%%%%%%%%%%%%%
We start with description of basic scheduling algorithm based on CSMA. 
As in \cite{ni09}, we divide each time slot $t$ into a \emph{control} 
slot and a \emph{data} slot, where the control slot is further divided 
into $W$ mini-slots. The purpose of the control slot is to generate a 
collision-free transmission schedule $M(t) \in \Matching$. To this end, 
the distributed CSMA scheduling selects at each time slot a set of links 
that form a feasible schedule. Such a schedule is called a \emph{decision} 
schedule and used to change links' state (between active and inactive). 
Let $\sigma(t)$ denote a decision schedule at time slot $t$.

Let $\Matching_0 \subseteq \Matching$ denote the set of possible
decision schedules under our CSMA-based algorithm. A decision 
schedule is selected through a randomized procedure, e.g., a 
decision schedule $\sigma(t) \in \Matching_0$ is selected with a
positive probability $\alpha(\sigma(t))$ satisfying that 
$\sum_{\sigma(t) \in \Matching_0} \alpha(\sigma(t)) = 1$. Based 
on the decision schedule, the schedule for actual data transmission 
is determined as follows. For each link $l \in \sigma(t)$, if no 
link in its interfering neighbors $I(l)$ was active at time slot
$t-1$, then the state of link $l$ becomes active with probability 
$p_l$ (which will be specified later) and inactive with probability 
$\bar{p}_l=1-p_l$ during time slot $t$. If at least one link in 
$I(l)$ was active in the previous time slot, then link $l$ remains 
inactive\footnote{In the previous data slot, link $l$ must be inactive 
since the schedule must be feasible.} in the current data slot. Any 
link $l^{\prime} \notin \sigma(t)$ will have its state unchanged from 
the previous time slot. Since the current state $M(t)$ depends only on 
the previous state $M(t-1)$ and the randomly selected decision schedule 
$\sigma(t)$, the transmission schedule $M(t)$ evolves as a discrete-time 
Markov chain (DTMC). Our basic scheduling algorithm is very similar to 
that of \cite{ni09}. The key difference is that the link activation 
probability is based on the shadow queue lengths instead of the data 
queue lengths. We refer the readers to \cite{ni09} for the detailed 
operations of the CSMA-based algorithms.

%%%%%%%%%%%%%%%%%%%%%%%%%%%%%%%%%%%%%%%%%%%%%%%%%%%%%%%%%%%%%%%%%%%%%%%%%%%%%%
\subsection{Distributed Implementation with Per-link Queues} \label{subsec:dist}
%%%%%%%%%%%%%%%%%%%%%%%%%%%%%%%%%%%%%%%%%%%%%%%%%%%%%%%%%%%%%%%%%%%%%%%%%%%%%%
In this section, we describe our distributed CSMA-based scheduling scheme 
with per-link queues, called \textbf{LQ-CSMA}. The LQ-CSMA algorithm can
be combined with priority or FIFO queueing discipline to develop fully 
distributed scheduling schemes.

We use the system settings and notations of per-link-queue structure as in 
Section~\ref{sec:lq}. We also control the shadow arrivals as (\ref{eq:lq-shadowarr}). 
As in \cite{ni09}, we set link activation probability $p_l = \frac{ e^{w_l(t)}} 
{e^{w_l(t)} + 1}$, where $w_l(t)$ is the weight of link $l$. We begin with 
defining a class of functions that will be used for weight calculation. As 
in \cite{eryilmaz05,ni09}, let $\mB$ denote the set of functions $g(\cdot): 
[0,\infty] \rightarrow [0,\infty]$ that satisfy the following conditions:
\begin{enumerate}
\item $g(x)$ is a non-decreasing and continuous function with $\lim_{x 
\rightarrow \infty} g(x) = \infty$.
\item Given any $M_1>0, M_2>0$ and $0 < \epsilon < 1$, there exists a 
$B < \infty$, such that for all $x>B$, we have $(1-\epsilon) g(x) \le g(x-M_1) 
\le g(x+M_2) \le (1+\epsilon) g(x)$.
\end{enumerate}
For example, functions $g(x)=\log(x+1)$, $g(x)=x^\alpha$ 
with $\alpha>0$, and $g(x)=e^{\sqrt{x}}$ belong to $\mB$, while $g(x)=e^x$ does
not. Similar to Chapter 4 of \cite{shah04}, to guarantee the existence of the 
fluid limit, we further define $\mC$ as a subset of $\mB$ such that $g(0)=0$, 
and for any $(x_1,\dots,x_n)$ and $(y_1,\dots,y_n)$ in $[0,\infty]^n$ and 
for any $\eta \in [0,1]$,
\begin{equation}
\label{eq:gcond}
\sum_i g(x_i) \ge \eta \sum_i g(y_i) \Rightarrow \sum_i g(r x_i) \ge \eta \sum_i g(r y_i),
~\text{for all}~ r>0.
\end{equation}
For example, $g(x)=x^\alpha$ with $\alpha>0$ is in $\mC$.

We set the weight of link $l \in \Edge$ at time slot $t$ as $w_l(t) = 
g_l(\hat{Q}_l(t))$, where $g_l \in \mC$. We highlight the differences 
from the original CSMA-based scheduling schemes as follows: i) the link 
weight is calculated by a function in set $\mC$ instead of $\mB$. This 
restriction is necessary to apply the fluid limit techniques; ii) the 
shadow queue length $\hat{Q}_l(t)$ is used for the weight calculation 
instead of the data queue length $Q_l(t)$. The following scheduling 
scheme is an extension of per-link-queue-based scheduling schemes to
CSMA-based algorithm. 

\noindent {\bf Per-Link-Queues-and-CSMA-based Scheduling Algorithm (LQ-CSMA):}
 
Let $p_l = \frac {e^{w_l(t)}} {e^{w_l(t)}+1}$, where $w_l(t) = g_l(\hat{Q}_l(t))$ 
is an appropriate function of the shadow queue length of link $l$ as shown above. 
At the beginning of each time slot, each link $l$ randomly selects a backoff time 
among $\{0,1,2,\cdots,W-1\}$, where $W$ denotes the contention window size. Link 
$l$ will send an INTENT message to announce its decision of attempting channel 
when this backoff time expires, unless an interfering link in $I(l)$ sent an INTENT 
message in an earlier mini-slot. The details are shown in Algorithm~\ref{alg:lq-csma}, 
which is similar to the Q-CSMA algorithm of \cite{ni09}, except that the activation 
probability $p_l$ is now determined based on the shadow queue lengths.

\begin{algorithm}
\caption{LQ-CSMA (at time slot $t$)} \label{alg:lq-csma}
\begin{algorithmic}
\State 1) Link $l$ selects a random (integer) backoff time $B_l$ 
uniformly in $[0,W-1]$ and waits for $B_l$ control mini-slots.
\State 2) IF link $l$ hears an INTENT message from a link in
$I(l)$ before the $(B_l + 1)$-st control mini-slot, $l$ will 
not be included in $\sigma(t)$ and will not transmit an INTENT
message anymore. Link $l$ will set $M_l(t) = M_l(t-1)$.
\State 3) IF link $l$ does not hear an INTENT message from any
link in $I(l)$ before the $(B_l + 1)$-st control mini-slot, 
it will send (broadcast) an INTENT message to all links in
$I(l)$ at the beginning of the $(B_l+1)$-st control mini-slot.

{\addtolength{\leftskip}{2pc} 
    \State - If there is a collision (i.e., if there is another 
	link in $I(l)$ transmitting an INTENT message in the same
	mini-slot), link $l$ will not be included in $\sigma(t)$ and
	will set $M_l(t) = M_l(t - 1)$.
    \State - If there is no collision, link $l$ will be included 
	in $\sigma(t)$ and decide its state as follows:
	
}

{\addtolength{\leftskip}{4pc}
\If{no links in $I(l)$ were active in the previous data slot}
	\State $M_l(t) = 1$ with probability $p_l$, $0 < p_l < 1$;
	\State $M_l(t) = 0$ with probability $\bar{p}_l = 1 - p_l$.
\Else
	\State $M_l(t) = 0$.
\EndIf
\State 4) IF $M_l(t)=1$, link $l$ will transmit a packet in the 
data slot, and will set $\hat{Q}_l(t)=(\hat{Q}_l(t)-1)^+$.}
\end{algorithmic}
\end{algorithm} 

\textbf{Remark:} The weight function $g_l(\hat{Q}_l(t))$ needs
to be appropriately chosen such that the DTMC of the transmission 
schedules converge faster compared to the dynamics of the link 
weights. For example\footnote{In \cite{jiang10, ni09, rajagopalan09}, 
the weight function $g_l$ is a function of the queue length $Q_l(t)$ 
rather than $\hat{Q}_l(t)$.}, $g_l(Q_l(t)) = \alpha Q_l(t)$ with a 
small $\alpha$ is suggested as a heuristic to satisfy the time-scale 
separation assumption in \cite{jiang10} and $g_l(Q_l(t)) = \log \log 
(Q_l(t)+e)$ is used in the proof of throughput optimality in 
\cite{rajagopalan09} to essentially separate the time scales. In addition, 
it has been reported in \cite{ni09} that the weight function $g_l(Q_l(t)) 
= \log (\alpha Q_l(t))$ with a small $\alpha$ gives the best empirical 
delay performance. In this paper, we make the time-scale separation 
assumption as in \cite{jiang10,ni09} and assume that the DTMC is in the 
steady state at every time slot.

Applying Lemma~3 of \cite{ni09}, we can show that the transmission schedule
$M(t)$ produced by LQ-CSMA is feasible and the decision schedule $\sigma$ 
satisfies $\bigcup_{\sigma \in \Matching_0} \sigma = \Edge$ when $W \ge 2$. 
Applying Proposition~1 of \cite{ni09}, we can obtain that the DTMC of the 
transmission schedules is irreducible and aperiodic (and reversible in this 
case), and has the following product-form stationary distribution:
\begin{eqnarray}
&& \textstyle \mu(M) = \frac {1} {\kappa} \prod_{l \in M} \frac {p_l} {\bar{p}_l}, \label{eq:dist} \\
&& \textstyle \kappa = \sum_{M \in \Matching} \prod_{l \in M} \frac {p_l} {\bar{p}_l}.
\end{eqnarray}
Then from Proposition~2 of \cite{ni09}, we can obtain the following lemma.

\begin{lemma}
\label{lem:weight}
If the window size $W \ge 2$, LQ-CSMA has the product-form distribution 
given by (\ref{eq:dist}). Further, given any $\zeta$ and $\gamma$, 
$0 < \zeta, \gamma < 1$, there exists a $Q_B>0$ such that: at any time 
slot $t$, with probability greater than $1-\zeta$, LQ-CSMA chooses a 
schedule $M(t) \in \Matching$ that satisfies
\begin{equation}
\label{eq:weight}
\sum_{l \in \Edge} w_l(t) \cdot M_l(t) \ge (1-\gamma) \max_{M \in \Matching}
\sum_{l \in \Edge} w_l(t) \cdot M_l 
\end{equation}
whenever $\|\hat{Q}(t)\| > Q_B$.
\end{lemma}

We omit the proof and refer interested readers to \cite{ni09} (Lemma 3,
Propositions~1 and 2) for details. 

Note that we have (\ref{eq:lple}) since Lemma~\ref{lem:lple} also holds 
under LQ-CSMA. Applying Lemma~\ref{lem:weight} and following the same 
line of analysis for the proof of Lemma~\ref{lem:sub-stable}, we can 
easily show that the sub-system of shadow queues $\hat{q}$ is stable 
under LQ-CSMA in the fluid limit model. 

\begin{lemma}
\label{lem:lq-csma}
Given any $\zeta$ and $\gamma$, $0 < \theta, \gamma < 1$, with probability 
greater than $1-\theta$, the sub-system of shadow queues $\hat{q}$ operating 
under LQ-CSMA satisfies that: For any $\zeta > 0$, there exists a finite 
$T_4 > 0$ such that, for any fluid model solution with $\|\hat{q}(0)\| \le 1$, 
we have
\begin{equation}
\|\hat{q}(t)\| \le \zeta, ~\text{for all}~ t \ge T_4,
\end{equation}
for any arrival rate vector strictly inside $(1-\gamma)\Lambda^{*}$. 
\end{lemma} 

The proof is provided in Appendix~\ref{app:lem:lq-csma}.

%%%%%%%%%%%%%%%%%%%%%%%%%%%%%%%%%%%%%%%%%%%%%%%%%%%%%%
%\subsection{LQ-CSMA with Different Queueing Disciplines} 
%%%%%%%%%%%%%%%%%%%%%%%%%%%%%%%%%%%%%%%%%%%%%%%%%%%%%%
The LQ-CSMA algorithm combined with priority queueing discipline 
and FIFO queueing discipline is called \textbf{PLQ-CSMA} and 
\textbf{FLQ-CSMA}, respectively. We present the main results 
of this section as follows.

\begin{proposition}
\label{pro:plq-csma} 
PLQ-CSMA is throughput-optimal.
\end{proposition}

\begin{proposition}
\label{pro:flq-csma}
FLQ-CSMA is throughput-optimal in networks where flows do not form loops.
\end{proposition}

Since the fluid limit model for the sub-system of shadow queues $\hat{q}$ 
is stable from Lemma~\ref{lem:lq-csma}, the results of Propositions~\ref{pro:plq-csma} 
and \ref{pro:flq-csma} follow the same line of analysis for the proof of
Propositions~\ref{pro:plq} and \ref{pro:flq}, respectively. We omit the proofs.

%%%%%%%%%%%%%%%%%%%%%%%%%%%%%%%%%%%%%%%%%%%
\section{Numerical Results} \label{sec:sim}
%%%%%%%%%%%%%%%%%%%%%%%%%%%%%%%%%%%%%%%%%%%
In this section, we evaluate different scheduling schemes through simulations. 
We compare scheduling performance of HQ-MWS, PLQ-MWS, FLQ-MWS with the original 
back-pressure (BP) algorithm under the \emph{node-exclusive}\footnote{It is also 
called the \emph{primary} or \emph{1-hop} interference model, where two links 
sharing a common node cannot be activated simultaneously. It has been known as 
a good representation for Bluetooth or FH-CDMA networks \cite{sharma06}.} 
interference model. Note that we focus on the node-exclusive interference model 
only for the purpose of illustration. Our scheduling schemes can be applied to 
general interference constraints as specified in Section~\ref{sec:mod}.
\high{We will first focus on a simple linear network topology to illustrate the 
advantages of the proposed schemes, and further validate our theoretical results 
in a larger and more realistic grid network topology. The impact of the parameter 
$\epsilon$ on the scheduling performance will also be explored and discussed.}

First, we evaluate and compare the scheduling performance of HQ-MWS, PLQ-MWS, 
FLQ-MWS and the back-pressure algorithm in a simple linear network that consists 
of 11 nodes and 10 links as shown in Fig.~\ref{fig:linear_1_a}, where nodes are 
represented by circles and links are represented by dashed lines with link capacity, 
respectively. We establish 10 flows that are represented by arrows, where each 
flow $i$ is from node $1$ to node $i+1$ via all the nodes in-between. We consider 
uniform traffic where all flows have packet arrivals at each time slot following 
Poisson distribution with the same mean rate $\lambda > 0$. We run our simulations 
with changing traffic load $\lambda$. Clearly, in this scenario, any traffic load 
with $\lambda < 0.5$ is feasible. We use $\epsilon=0.005$ for HQ-MWS, PLQ-MWS and FLQ-MWS. 
\high{We evaluate the performance by measuring average packet delays 
(in unit of time slot) over all the delivered packets (that reach their respective 
destination nodes) in the network.}

\begin{figure}[t]
\centering
\subfigure[Linear network topology with ten links]{\epsfig{file=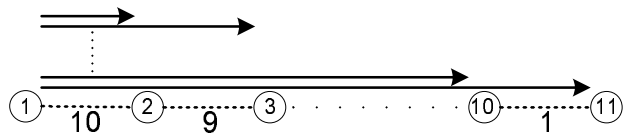,width=0.5\linewidth}\label{fig:linear_1_a}}\\
\subfigure[Average delay]{\epsfig{file=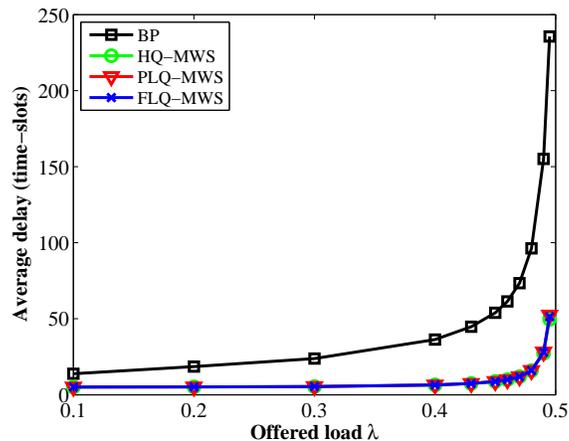,width=0.5\linewidth}\label{fig:linear_1_b}}\\
% \subfigure[Linear network topology with ten links]{\epsfig{file=linear1.eps,width=0.7\linewidth}\label{fig:linear_1_a}}\\
% \subfigure[Average delay]{\epsfig{file=delay_linear1.eps,width=0.8\linewidth}\label{fig:linear_1_b}}\\
\caption{Performance of BP, HQ-MWS, PLQ-MWS and FLQ-MWS in a linear network topology ($\epsilon=0.005$).}
\label{fig:linear_1}
\end{figure}

\begin{figure}[t]
\centering
\subfigure[Linear network topology with ten links]{\epsfig{file=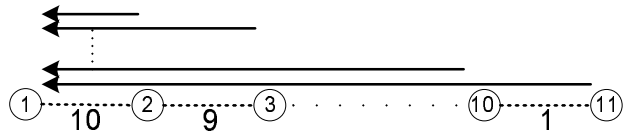,width=0.5\linewidth}\label{fig:linear_2_a}}\\
\subfigure[Average delay]{\epsfig{file=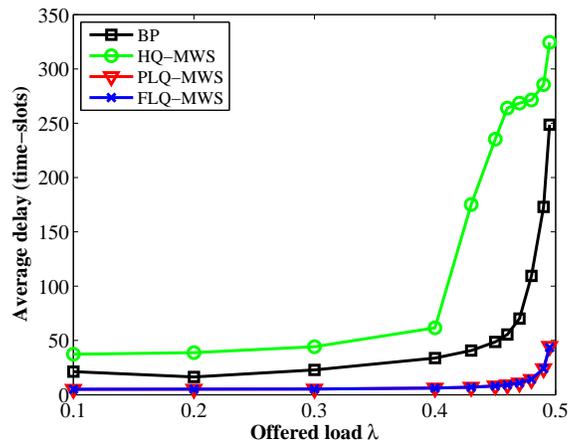,width=0.5\linewidth}\label{fig:linear_2_b}}\\
% \subfigure[Linear network topology with ten links]{\epsfig{file=linear2.eps,width=0.7\linewidth}\label{fig:linear_2_a}}\\
% \subfigure[Average delay]{\epsfig{file=delay_linear2.eps,width=0.8\linewidth}\label{fig:linear_2_b}}\\
\caption{Performance of BP, HQ-MWS, PLQ-MWS and FLQ-MWS in a linear network topology ($\epsilon=0.005$).}
\label{fig:linear_2}
\end{figure}

Fig.~\ref{fig:linear_1_b} plots the average delays under different offered loads 
to examine the performance limits of different scheduling schemes. Each result 
represents a simulation run that lasts for $10^7$ time slots. Since the optimal 
throughput region $\Lambda^*$ is defined as the set of arrival rate vectors under 
which queue lengths and thus delays remain finite, we can consider the traffic 
load, under which the average delay increases rapidly, as the boundary of the 
optimal throughput region. Fig.~\ref{fig:linear_1_b} shows that all schemes achieve 
the same boundary (i.e., $\lambda<0.5$), which supports our theoretical results on 
throughput optimality. Moreover, all the three proposed schemes achieve substantially 
better delay performance than the back-pressure algorithm. This is because under 
the back-pressure algorithm, the queue lengths have to build up along the route 
a flow takes from the destination to the source, and in general, earlier hop link 
has a larger queue length. This leads to poor delay performance especially when 
the route of a flow is lengthy, which is the case in Fig.~\ref{fig:linear_1_a}. 
Note that in this specific scenario, there is only one per-hop queue at each link 
under HQ-MWS. Hence, HQ-MWS is equivalent to PLQ-MWS and FLQ-MWS in this scenario, 
which explains why the three proposed schemes perform the same as in Fig.~\ref{fig:linear_1_b}. 

\begin{figure}
\centering
% \subfigure[A grid network topology]{\epsfig{file=grid.eps,width=0.5\linewidth}\label{fig:grid_a}}\\
% \subfigure[Average delay for MWS schemes with $\epsilon=0.05$]{\epsfig{file=delay_grid.eps,width=0.8\linewidth}\label{fig:grid_b}}
%\subfigure[Average delay for CSMA schemes with $\epsilon=0.005$]{\epsfig{file=delay_csma.eps,width=0.8\linewidth}\label{fig:grid_c}}
\subfigure[A grid network topology]{\epsfig{file=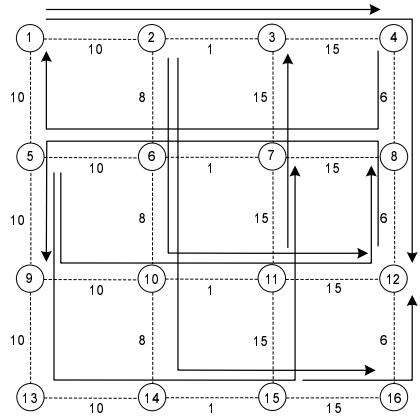,width=0.33\linewidth}\label{fig:grid_a}}\\
\subfigure[Average delay for MWS schemes with $\epsilon=0.05$]{\epsfig{file=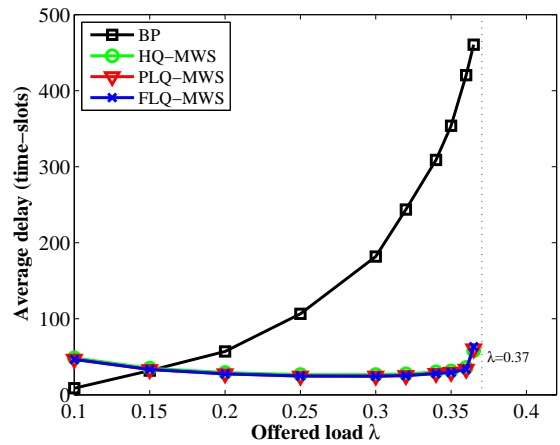,width=0.5\linewidth}\label{fig:grid_b}}
\subfigure[Average delay for CSMA schemes with $\epsilon=0.005$]{\epsfig{file=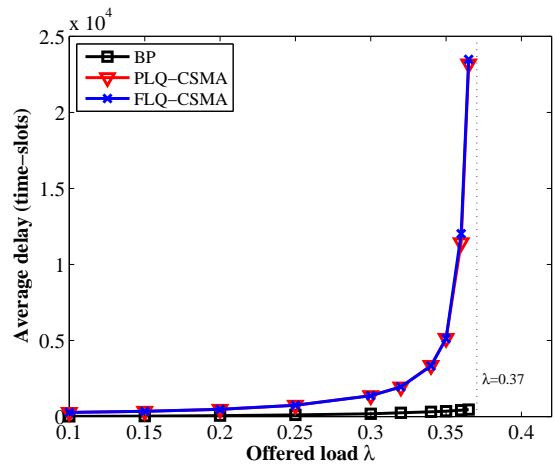,width=0.5\linewidth}\label{fig:grid_c}}
\caption{Performance of all the proposed scheduling schemes in a grid network 
with 16 nodes and 24 links. In Fig.~\ref{fig:grid_b}, the vertical dotted line 
$\lambda=0.37$ denotes an upper bound for the feasible values of $\lambda$.}
\label{fig:grid}
\end{figure}

Second, we evaluate the performance of the proposed schemes in the same linear network 
as in the previous case while reversing the direction of each flow. The new topology is 
illustrated in Fig.~\ref{fig:linear_2_a}. In this scenario, the number of per-hop queues 
HQ-MWS maintains for each link is the same as the number of flows passing through that 
link. Hence, HQ-MWS is expected to operate differently from PLQ-MWS and FLQ-MWS, and 
achieves different (and potentially poorer) delay performance. All the other simulation 
settings are kept the same as in the previous case. Fig.~\ref{fig:linear_2_b} shows that 
all schemes achieve the same boundary (i.e., $\lambda<0.5$) in this scenario, which again 
supports our theoretical results on throughput performance. However, we observe that 
HQ-MWS has the worst delay performance, while PLQ-MWS and FLQ-MWS achieve substantially 
better performance. This is because PLQ-MWS and FLQ-MWS transmit packets 
more efficiently and do not waste service as long as there are enough packets at the 
activated link, while the back-pressure algorithm and HQ-MWS maintain multiple queues 
for each link, and may waste service if the activated queue has less packets than the 
link capacity. HQ-MWS has larger delays than the back-pressure algorithm because the 
scheduling decisions of HQ-MWS are based on the shadow queue lengths rather than the 
actual queue lengths: a queue with very small (or even zero) queue length could be 
activated. This introduces another type of inefficiency in HQ-MWS. Note that PLQ-MWS 
and FLQ-MWS also make scheduling decisions based on the shadow queue lengths. However, 
their performance improvement from a single queue per link dominates delay increases 
from the inefficiency. These observations imply that maintaining per-link queues not 
only simplifies the data structure, but also improves scheduling efficiency and reduces 
delays.

\high{
Next, we evaluate the performance of all the proposed schemes in a larger grid network 
with 16 nodes and 24 links as shown in Fig.~\ref{fig:grid_a}, where the capacity of each 
link has been shown beside the link and carefully assigned to avoid traffic symmetry. 
Similar type of grid networks have been adopted in the literature (e.g., \cite{ni09,bui11,
ji11}) to numerically evaluate scheduling performance. We establish 10 multi-hop flows 
that are represented by arrows in Fig.~\ref{fig:grid_a}. Again, we consider uniform traffic 
where each flow has independent packet arrivals at each time slot following Poisson 
distribution with the same mean rate $\lambda > 0$. In this scenario, we can calculate 
an upper bound of $1/(4/8+2/10+2) = 10/27 \approx 0.37$ for the feasible value of $\lambda$, 
by looking at the flows passing through node 6, which is the bottleneck in the network. 

We choose $\epsilon=0.05$ for HQ-MWS, PLQ-MWS and FLQ-MWS. Under each scheduling scheme 
along with the back-pressure algorithm, we measure average packet delays under different 
offered loads to examine their performance limits. Fig.~\ref{fig:grid_b} shows that the 
proposed schemes have higher packet delays than the back-pressure algorithm when traffic 
load is light (e.g., $\lambda < 0.15$). This is due to the aforementioned inefficiency 
under the proposed schemes: since the scheduling decisions are based on the shadow 
queue lengths rather than the actual queue lengths, queues with very small (or even zero) 
queue length can be activated. However, the effect tends to decrease with heavier traffic 
load, since the queue lengths are likely to be large. The results also show that the proposed 
schemes consistently outperform the back-pressure algorithm when $\lambda > 0.15$. Note that 
with $\epsilon = 0.05$, the shadow traffic rate vector is outside the optimal throughput 
region when $\lambda > 0.37/(1+0.05) \approx 0.35$, however, interestingly, the schedules 
chosen based on the shadow queue lengths can still stabilize the data queues even if $0.35 
< \lambda < 0.37$ (which is still feasible). 
Nevertheless, we later will show that this is not always the case. 
For PLQ-CSMA and FLQ-SMA, similar as in \cite{ni09}, we choose contention window size 
$W=48$, weight function $w_l(t)=\log(0.1 \hat{Q}_l(t))$, and link activation probability 
$p_l = \frac {e^{w_l(t)}} {e^{w_l(t)}+1}$. We choose $\epsilon=0.005$ for PLQ-CSMA and 
FLQ-CSMA, and plot their average delays over offered loads in Fig.~\ref{fig:grid_c}, 
along with the back-pressure algorithm. Fig.~\ref{fig:grid_c} shows that although PLQ-CSMA 
and FLQ-CSMA achieve the optimal throughput performance, they suffer from very poor delay 
performance as expected. This is due to the long mixing time of the underlying Markov 
chain formed by the transmission schedules \cite{ni09}. 
Note that in the above scenario, FLQ-MWS does not guarantee throughput optimality, since 
flows $(5 \rightarrow 9 \rightarrow 10 \rightarrow 11 \rightarrow 12 \rightarrow 8)$ and 
$(12 \rightarrow 8 \rightarrow 7 \rightarrow 6 \rightarrow 5 \rightarrow 9)$ form a loop. 
However, the results in Fig.~\ref{fig:grid_b} suggest that all the schemes, including FLQ-MWS, 
empirically achieve the optimal throughput performance. This opens up an interesting question 
about throughput performance of FLQ-MWS in general settings.
}

\begin{figure}[t]
\centering
\epsfig{file=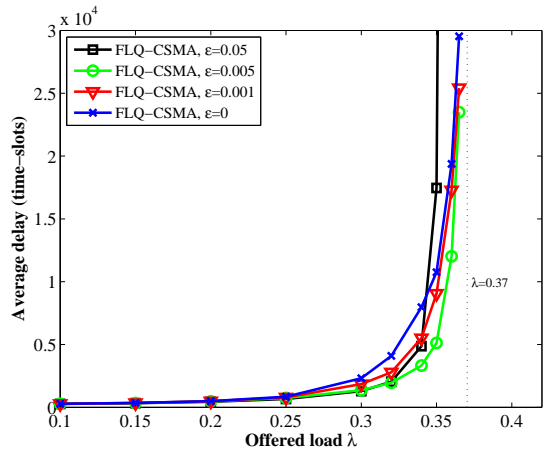,width=0.5\linewidth}
\caption{The impact of the value of $\epsilon$ on the scheduling performance.}
\label{fig:epsilon}
\end{figure}

\high{
Finally, we investigate sensitivity of parameter $\epsilon$ on the scheduling
performance, by runing simulations for PLQ-CSMA and FLQ-CSMA with different values 
of $\epsilon$ in the grid network in Fig.~\ref{fig:grid_a}. Since the performances 
of PLQ-CSMA and FLQ-CSMA are very close, we report only the results for FLQ-CSMA 
in Fig.~\ref{fig:epsilon}, where we plot average packet delays over the offered load
$\lambda$ for FLQ-CSMA with $\epsilon=0, 0.001, 0.005$ and $0.05$, respectively. 
The results show that the delay performance generally improves with a larger value 
of $\epsilon$, in particular under moderate and heavy traffic loads (e.g., $\lambda 
> 0.25$). This is because a larger value of $\epsilon$ leads to more aggressive 
link activations. However, it can be observed that a larger value of $\epsilon$ 
(e.g., $\epsilon = 0.05$) could make the system unstable when the offered load is 
close to the capacity boundary (e.g., $\lambda > 0.35$). On the other hand, the 
impact of $\epsilon$ becomes marginal under light traffic loads (i.e., $\lambda$ 
is small), as the inefficiency of small queue activation dominates the scheduling 
performance. Interestingly, although we require $\epsilon$ be positive in the 
analysis for throughput optimality, the simulation results show that the proposed 
schemes can empirically achieve the optimal throughput performance even when 
$\epsilon=0$, leading to much larger delays though.
}

%%%%%%%%%%%%%%%%%%%%%%%%%%%%%%%%%%%%
\section{Conclusion} \label{sec:con}
%%%%%%%%%%%%%%%%%%%%%%%%%%%%%%%%%%%%
In this paper, we developed scheduling policies with per-hop or per-link queues and 
a shadow algorithm to achieve the overall goal of removing per-flow or per-destination 
information requirement, simplifying queue structure, exploiting only local information, 
and potentially reducing delay. We showed throughput optimality of the proposed schemes 
that use only the readily available hop-count information, using fluid limit techniques 
via an inductive argument. We further simplified the solution using FIFO queueing discipline 
with per-link queues and showed that this is also throughput-optimal in networks without 
flow-loops. The problem of proving throughput optimality in general networks with algorithms 
(like FLQ-MWS) that use only per-link information remains an important open and challenging 
problem. Further, it is also worthwhile to investigate the problem with dynamic routing 
and see if per-flow and per-destination information can be removed even when routes are
not fixed.

\appendices

\high{
%%%%%%%%%%%%%%%%%%%%%%%%%%%%%%%%%%%%%%%%%%%%%%%%%%%%%%%%%%%%%%
\section{Proof of Proposition~\ref{pro:hq}} \label{app:pro:hq}
%%%%%%%%%%%%%%%%%%%%%%%%%%%%%%%%%%%%%%%%%%%%%%%%%%%%%%%%%%%%%%
To begin with, let $Q(t) \triangleq [Q_\lk(t)]$ and $\hat{Q}(t) \triangleq 
[\hat{Q}_\lk(t)]$ denote the queue length vector and the shadow queue length 
vector at time slot $t$, respectively. We use $\|\cdot\|$ to denote the 
$L_1$-norm of a vector, e.g., $\|Q(t)\|=\sum_{l\in\Edge} \sum_{k=1}^{L^{\max}} 
Q_\lk(t)$. We let $m_\lk(i)$ be the index of the flow to which the $i$-th 
packet of queue $Q_\lk$ belongs. In particular, $m_\lk(1)$ indicates the 
index of the flow to which the head-of-line packet of queue $Q_\lk$ belongs. 
We define the state of queue $Q_\lk$ at time slot $t$ as $\Queue_\lk(t) = 
[m_\lk(1), \cdots, m_\lk(Q_\lk(t))]$ in an increasing order of the arriving 
time, or an empty sequence if $Q_\lk(t)=0$. Then we denote its vector by 
$\Queue(t) \triangleq [\Queue_\lk(t)]$. Define $\Int_{\Flow} \triangleq 
\{1,2,\cdots, |\Flow|\}$, and let $\Int^{\infty}_{\Flow}$ be the set of 
finitely terminated sequences taking values in $\Int_{\Flow}$. It is evident 
that $\Queue_\lk(t) \in \Int^{\infty}_{\Flow}$, and hence $\Queue(t) \in 
(\Int^{\infty}_{\Flow})^{|\Edge| \times L^{\max}}$. We define $\Markov(t) 
\triangleq (\Queue(t),\hat{Q}(t), \frac{1}{t+1} A(t))$, and then $\Markov 
= (\Markov(t), t\ge0)$ is the process describing the behavior of the 
underlying system. Note that in the third term of $\Markov(t)$, we use 
$\frac{1}{t+1}A(t)$ instead of $\frac{1}{t}A(t)$ so that it is well-defined 
when $t=0$. Clearly, the evolution of $\Markov$ forms a countable Markov 
chain under HQ-MWS. We abuse the notation of $L_1$-norm by writing the norm 
of $\Markov(t)$ as $\| \Markov(t) \| \triangleq \| Q(t) \| + \lceil \| 
\hat{Q}(t) \| \rceil + \lceil \frac{1}{t+1} \|A(t)\| \rceil$. Let $\Markov^{(x)}$ 
denote a process $\Markov$ with an initial condition such that
\begin{equation}
\label{eq:init_config}
\textstyle \| \Markov^{(x)}(0) \| = x.
\end{equation}

The following Lemma was derived in \cite{rybko92} for continuous-time 
countable Markov chains, and it follows from more general results in 
\cite{malyshev79} for discrete-time countable Markov chains. 

\begin{lemma}[Theorem~4 of \cite{andrews04}]
\label{lem:stab_cri}
Suppose that there exist a $\xi>0$ and a finite integer $T > 0$ such 
that for any sequence of processes $\{\frac {1} {x} \Markov^{(x)}(x T), 
x=1,2,\cdots\}$, we have
\begin{equation}
\label{eq:stab_cri}
\textstyle \limsup_{x \rightarrow \infty} \Expect \left[\frac {1} {x} 
\| \Markov^{(x)} (x T) \| \right] \le 1 - \xi.
\end{equation}
Then, the Markov chain $\Markov$ is stable.
\end{lemma}

Lemma~\ref{lem:stab_cri} implies the stability of the network. 
A stability criterion of type (\ref{eq:stab_cri}) leads to a 
fluid limit approach \cite{dai95} to the stability problem of 
queueing systems. We start our analysis by establishing the 
\emph{fluid limit model} as in \cite{andrews04,dai95}. We 
define another process $\System \triangleq \left(F, U, Q, \Pi, 
\Psi, A, D, P, \hat{Q}, \hat{\Pi}, \hat{\Psi}, \hat{A}, \hat{D}, 
\hat{P} \right)$, where the tuple denotes a list of vector 
processes. Clearly, a sample path of $\System^{(x)}$ uniquely 
defines the sample path of $\Markov^{(x)}$. Then we extend the 
definition of $\System$ to each continuous time $t \ge 0$ as 
$\System^{(x)}(t) \triangleq \System^{(x)}(\lfloor t \rfloor)$. 
%Hence, each component of $\System^{(n)}(t)$ is right continuous with left limits. 

Recall that a sequence of functions $f_n(\cdot)$ is said to converge to
a function $f(\cdot)$ \emph{uniformly over compact (u.o.c.)} intervals 
if for all $t \ge 0$, $\lim_{n \rightarrow \infty} \sup_{0 \le t^{\prime} 
\le t} |f_n(t^{\prime}) - f(t^{\prime})|=0$. Next, we consider a sequence 
of processes $\{\frac {1} {\xn} \System^{(\xn)}(\xn \cdot)\}$ that is 
scaled both in time and space. Then, using the techniques of Theorem~4.1 
of \cite{dai95} or Lemma 1 of \cite{andrews04}, we can show the convergence 
properties of the sequences in the following lemma.
\begin{lemma}
\label{lem:fluid}
With probability one, for any sequence of processes $\{\frac {1} {\xn} 
\System^{(\xn)}(\xn \cdot)\}$, where $\{\xn\}$ is a sequence of positive
integers with $\xn \rightarrow \infty$, there exists a subsequence 
$\{\xnj\}$ with $\xnj \rightarrow \infty$ as $j \rightarrow \infty$ 
such that the following \emph{u.o.c. convergences} hold:
% \begin{eqnarray}
% && \textstyle\frac{1}{\xnj}F^{(\xnj)}_s(\xnj t) \rightarrow  f_s(t), \label{eq:fluid_f} \\
% && \textstyle\frac{1}{\xnj}U^{(\xnj)}_\sk(\xnj t) \rightarrow  u_\sk(t), \\
% && \textstyle\frac{1}{\xnj}A^{(\xnj)}_\lk(\xnj t) \rightarrow  a_\lk(t), \quad \frac{1}{\xnj}\hat{A}^{(\xnj)}_\lk(\xnj t) \rightarrow  \hat{a}_\lk(t), \label{eq:fluid_a} \\
% && \textstyle\frac{1}{\xnj}Q^{(\xnj)}_\lk(\xnj t) \rightarrow  q_\lk(t), \quad \frac{1}{\xnj}\hat{Q}^{(\xnj)}_\lk(\xnj t) \rightarrow  \hat{q}_\lk(t), \\
% && \textstyle\frac{1}{\xnj}D^{(\xnj)}_\lk(\xnj t) \rightarrow  d_\lk(t), \quad \frac{1}{\xnj}\hat{D}^{(\xnj)}_\lk(\xnj t) \rightarrow  \hat{d}_\lk(t), \\
%%\end{eqnarray}
%%\begin{eqnarray}
% && \textstyle\frac{1}{\xnj}\int_0^{\xnj t} \Pi^{(\xnj)}_\lk(\tau)d\tau \rightarrow \int_0^t \pi_\lk(\tau) d\tau, \quad \frac{1}{\xnj}\int_0^{\xnj t} \hat{\Pi}^{(\xnj)}_\lk(\tau)d\tau \rightarrow \int_0^t \hat{\pi}_\lk(\tau) d\tau, \\
% && \textstyle\frac{1}{\xnj}\int_0^{\xnj t} \Psi^{(\xnj)}_\lk(\tau)d\tau \rightarrow \int_0^t \psi_\lk(\tau) d\tau, \quad	\frac{1}{\xnj}\int_0^{\xnj t} \hat{\Psi}^{(\xnj)}_\lk(\tau)d\tau \rightarrow \textstyle \int_0^t \hat{\psi}_\lk(\tau) d\tau, \\
% && \textstyle\frac{1}{\xnj}\int_0^{\xnj t} P^{(\xnj)}_\lk(\tau)d\tau \rightarrow  \int_0^t p_\lk(\tau) d\tau, \quad	\frac{1}{\xnj}\int_0^{\xnj t} \hat{P}^{(\xnj)}_\lk(\tau)d\tau \rightarrow \int_0^t \hat{p}_\lk(\tau) d\tau. \label{eq:fluid_hp}
% \end{eqnarray}
\begin{eqnarray}
&& \textstyle\frac{1}{\xnj}F^{(\xnj)}_s(\xnj t) \rightarrow  f_s(t), \label{eq:fluid_f} \\
&& \textstyle\frac{1}{\xnj}U^{(\xnj)}_\sk(\xnj t) \rightarrow  u_\sk(t), \label{eq:fluid_u}\\
&& \textstyle\frac{1}{\xnj}A^{(\xnj)}_\lk(\xnj t) \rightarrow  a_\lk(t), \label{eq:fluid_a} \\
&& \textstyle\frac{1}{\xnj}\hat{A}^{(\xnj)}_\lk(\xnj t) \rightarrow  \hat{a}_\lk(t), \label{eq:fluid_ha} \\
&& \textstyle\frac{1}{\xnj}Q^{(\xnj)}_\lk(\xnj t) \rightarrow  q_\lk(t), \label{eq:fluid_q}\\ 
&& \textstyle\frac{1}{\xnj}\hat{Q}^{(\xnj)}_\lk(\xnj t) \rightarrow  \hat{q}_\lk(t), \\
&& \textstyle\frac{1}{\xnj}D^{(\xnj)}_\lk(\xnj t) \rightarrow  d_\lk(t), \label{eq:fluid_d} \\ 
&& \textstyle\frac{1}{\xnj}\hat{D}^{(\xnj)}_\lk(\xnj t) \rightarrow  \hat{d}_\lk(t), \label{eq:fluid_hd}\\
%\end{eqnarray}
%\begin{eqnarray}
&& \textstyle\frac{1}{\xnj}\int_0^{\xnj t} \Pi^{(\xnj)}_\lk(\tau)d\tau \rightarrow \int_0^t \pi_\lk(\tau) d\tau, \label{eq:fluid_pi}\\
&& \textstyle\frac{1}{\xnj}\int_0^{\xnj t} \hat{\Pi}^{(\xnj)}_\lk(\tau)d\tau \rightarrow \int_0^t \hat{\pi}_\lk(\tau) d\tau, \\
&& \textstyle\frac{1}{\xnj}\int_0^{\xnj t} \Psi^{(\xnj)}_\lk(\tau)d\tau \rightarrow \int_0^t \psi_\lk(\tau) d\tau, \label{eq:fluid_psi}\\
&& \textstyle\frac{1}{\xnj}\int_0^{\xnj t} \hat{\Psi}^{(\xnj)}_\lk(\tau)d\tau \rightarrow \textstyle \int_0^t \hat{\psi}_\lk(\tau) d\tau, \\
&& \textstyle\frac{1}{\xnj}\int_0^{\xnj t} P^{(\xnj)}_\lk(\tau)d\tau \rightarrow  \int_0^t p_\lk(\tau) d\tau, \label{eq:fluid_p}\\
&& \textstyle\frac{1}{\xnj}\int_0^{\xnj t} \hat{P}^{(\xnj)}_\lk(\tau)d\tau \rightarrow \int_0^t \hat{p}_\lk(\tau) d\tau, \label{eq:fluid_hp}
\end{eqnarray}
where the functions $f_s,u_\sk,a_\lk,d_\lk,q_\lk,\hat{a}_\lk,\hat{d}_\lk,\hat{q}_\lk$ 
are Lipschitz continuous in $[0,\infty)$.
\end{lemma}

Note that the proof of the above lemma is quite standard using the techniques developed 
in \cite{andrews04,dai95,dai00}. We provide the proof in Appendix~\ref{app:lem:fluid} 
for completeness.

Any set of limiting functions $(f, u, q$, $\pi, \psi$, $a, d, p$, $\hat{q}$, $\hat{\pi}$, 
$\hat{\psi}$, $\hat{a}$, $\hat{d}$, $\hat{p})$ is called a \emph{fluid limit}. The family 
of these fluid limits is associated with our original stochastic network. The scaled
sequences $\{\frac {1} {\xn} \System^{(\xn)}(\xn \cdot)\}$ and their limits are referred
to as a \emph{fluid limit model} \cite{bramson08}. Since some of the limiting functions, 
namely $f_s,u_\sk,a_\lk,d_\lk,q_\lk,\hat{a}_\lk,\hat{d}_\lk,\hat{q}_\lk$, are Lipschitz 
continuous in $[0,\infty)$, they are absolutely continuous. Therefore, these limiting 
functions are differentiable at almost all time $t \in [0,\infty)$, which we call 
\emph{regular} time. 

Next, we will present the \emph{fluid model equations} of the system, i.e., 
Eqs. (\ref{eq:flt})-(\ref{eq:pipi}). Fluid model equations can be thought of 
as belonging to a \emph{fluid network} which is the deterministic equivalence
of the original stochastic network. Any set of functions satisfying the fluid 
model equations is called a \emph{fluid model solution} of the system. We show 
in the following lemma that any fluid limit is a fluid model solution.

\begin{lemma}
\label{lem:fluid_eq}
Any fluid limit $(f, u, q$, $\pi, \psi$, $a, d, p$, $\hat{q}$, $\hat{\pi}$, 
$\hat{\psi}$, $\hat{a}$, $\hat{d}$, $\hat{p})$ satisfies the following equations:
\begin{eqnarray} 
&& \textstyle	f_s(t) = \lambda_s t, \label{eq:flt} \\
&& \textstyle	q_\lk(t) = q_\lk(0) + a_\lk(t) - d_\lk(t), \label{eq:qad} \\
&& \textstyle	a_\lk(t) = \sum_s \Hslk u_\sk(t), \label{eq:ar} \\
&& \textstyle	a_\lk(t) = \int_0^t p_\lk(\tau) d\tau, \label{eq:ap} \\
% \end{eqnarray}
% \begin{eqnarray}
&& \textstyle	d_\lk(t) = \int_0^t \psi_\lk(\tau) d\tau, \label{eq:dp} \\
&& \textstyle	\psi_\lk(t) \le \pi_\lk(t), \label{eq:pp} \\
&& \textstyle	\frac{d}{dt} q_\lk(t) = p_\lk(t) - \psi_\lk(t),  \label{eq:dq1} \\
&& \textstyle	\frac{d}{dt} q_\lk(t) = \left\{
\begin{array}{ll}
p_\lk(t) - \pi_\lk(t), & \text{if}~q_\lk(t) > 0,\\
(p_\lk(t) - \pi_\lk(t))^+, & \text{otherwise},
\end{array}
\right. \label{eq:dq2} \\
&& \textstyle	\hat{q}_\lk(t) = \hat{q}_\lk(0) + \hat{a}_\lk(t)  - \hat{d}_\lk(t),  \label{eq:hqad} \\
&& \textstyle	\hat{a}_\lk(t) = \int_0^t \hat{p}_\lk(\tau) d\tau, \\
&& \textstyle	\hat{d}_\lk(t) = \int_0^t \hat{\psi}_\lk(\tau) d\tau, \\
&& \textstyle	\hat{\psi}_\lk(t) \le \hat{\pi}_\lk(t),  \label{eq:hpp} \\
&& \textstyle	\frac{d}{dt} \hat{q}_\lk(t) = \hat{p}_\lk(t) - \hat{\psi}_\lk(t), \label{eq:hdq1}\\
&& \textstyle	\frac{d}{dt} \hat{q}_\lk(t) = \left\{
\begin{array}{ll}
\hat{p}_\lk(t) - \hat{\pi}_\lk(t), & \text{if}~\hat{q}_\lk(t) > 0,\\
(\hat{p}_\lk(t) - \hat{\pi}_\lk(t))^+, & \text{otherwise},
\end{array}
\right. \label{eq:hdq2} \\
&& \textstyle	\|q(0)\| + \| \hat{q}(0)\| \le 1, \label{eq:init} \\
&& \textstyle	\pi_\lk(t) = \hat{\pi}_\lk(t). \label{eq:pipi}
\end{eqnarray}
\end{lemma}

\begin{proof}
\high{
Note that (\ref{eq:flt}) follows from the strong law of large numbers.
Eqs.~(\ref{eq:qad})-(\ref{eq:pp}) and (\ref{eq:hqad})-(\ref{eq:hpp}) are satisfied 
from the definitions. Since each of the limiting functions $q_\lk(t)$ is differentiable 
at any regular time $t\ge0$, (\ref{eq:dq1}) is satisfied from (\ref{eq:ap}) and 
(\ref{eq:dp}), by taking derivative of both sides of (\ref{eq:qad}). Similarly, 
(\ref{eq:hdq1}) is satisfied. Further, (\ref{eq:dq1}) and (\ref{eq:hdq1}) can be 
rewritten as (\ref{eq:dq2}) and (\ref{eq:hdq2}), respectively. Eq.~(\ref{eq:init}) 
is from the initial configuration (\ref{eq:init_config}), and (\ref{eq:pipi}) is 
due to the operations of HQ-MWS algorithm. 
}
\end{proof}

Due to the result of Lemma~\ref{lem:stab_cri}, we want to show 
that the stability criterion of (\ref{eq:stab_cri}) holds. Note 
that from system causality, we have $a_\lk(t) \le t \sum_s H^s_\lk 
\lambda_s + \sum_s \sum_h q_{s,h}(0)$ for all link $l \in \Edge$ 
and all $1 \le k \le L^{\max}$, for all $t \ge 0$. Then, we have 
\[
\begin{split}
\textstyle \lim_{j \rightarrow \infty} & \textstyle \frac {1} {\xnj} \| A^{(\xnj)} (\xnj t) \| \\
&\le \textstyle \sum_l \sum_k (t \sum_s H^s_\lk \lambda_s + \sum_s \sum_h q_{s,h}(0))
\end{split}
\]
almost surely, and thus,
\begin{equation}
\label{eq:A0}
\textstyle \lim_{j \rightarrow \infty} \frac {1} {\xnj} \left \lceil 
\frac {1} {\xnj t + 1} \| A^{(\xnj)} (\xnj t) \| \right \rceil = 0
\end{equation}
almost surely, for all $t \ge 0$. Therefore, it remains to be shown that 
the fluid limit model for the joint system of data queues and shadow
queues is stable (Lemma~\ref{lem:joint-stable}). Then, by uniform 
integrability of the sequence $\{ \frac {1} {x} \| \Markov^{(x)}(x T)\|, 
x=1,2,\cdots \}$ it implies that (\ref{eq:stab_cri}) holds. We divide the 
proof of Lemma~\ref{lem:joint-stable} into two parts: 1) in 
Lemma~\ref{lem:sub-stable}, we show that the sub-system consisting of 
shadow queues is stable; 2) in Lemma~\ref{lem:stable}, the sub-system 
consisting of data queues is stable. Before proving Lemmas~\ref{lem:sub-stable} 
and \ref{lem:stable}, we state and prove Lemmas~\ref{lem:ple} and \ref{lem:p}, 
which are used to prove Lemmas~\ref{lem:sub-stable} and \ref{lem:stable}, 
respectively.

The following lemma shows that the instantaneous shadow arrival 
rate is bounded in the fluid limit, and is used to show that the 
fluid limit model for the sub-system consisting of shadow queues 
is stable under HQ-MWS.

\begin{lemma}
\label{lem:ple}
For all (scaled) time $t>0$, and for all links $l \in \Edge$ and $1 \le 
k \le L^{\max}$, with probability one, the following inequality holds,
\begin{equation}
\label{eq:lk}
\textstyle \hat{p}_\lk(t) \le (1+\epsilon) \left( \sum_s \Hslk \lambda_s + \frac {1} {t} \right),
\end{equation}
and in particular,
\begin{equation}
\label{eq:l1}
\textstyle \hat{p}_{l,1}(t) = (1+\epsilon) \sum_s H^s_{l,1} \lambda_s.
\end{equation}
\end{lemma}

\begin{proof}
We start by stating the following lemma, which will be used to prove
Lemma~\ref{lem:ple}. 

\begin{lemma}
\label{lem:conv}
If a sequence $\{F(n), n=1,2,\cdots\}$ satisfies $\lim_{n \rightarrow
\infty} F(n) = f$, then the following holds,
%\begin{equation}
\[
\textstyle \lim_{n \rightarrow \infty} \frac {\sum^{n}_{\tau=1} F(\tau)} {n} = f.
\]
%\end{equation}
\end{lemma}

\begin{proof}
We want to show that, for any $\epsilon_1 > 0$, there exists an $N < \infty$ 
such that $\left |\frac {\sum^{n}_{\tau=1} F(\tau)} {n} - f \right| < \epsilon_1$, 
for all $n \ge N$. 

Since $\lim_{n \rightarrow \infty} F(n) = f$, then for any $\epsilon_1 > 0$, 
there exists a $N_1 < \infty$ such that $\left|F(n) - f \right| < \frac 
{\epsilon_1} {3}$, for all $n \ge N_1$. Letting $N = \max \left\{N_1, \frac 
{3(N_1-1)f} {\epsilon_1}, \frac {3\sum^{N_1-1}_{\tau=1} F(\tau)} {\epsilon_1} 
\right\}$, then for all $n \ge N$, we have
\begin{equation}
\begin{split}
\left|\frac {\sum^{n}_{\tau=1} F(\tau)} {n} - f \right|
&= \left|\frac {\sum^{N_1-1}_{\tau=1} F(\tau)} {n} + 
\frac {\sum^{n}_{\tau=N_1} F(\tau)} {n} - f \right| \\
&\le \frac {\epsilon_1} {3} + \left|\frac {\sum^{n}_{\tau=N_1} F(\tau)} {n} 
- \frac {n-N_1+1} {n} f \right| + \left|\frac {N_1-1} {n} f \right| \\
&< \frac {\epsilon_1} {3} + \frac {n-N_1+1} {n} \frac {\epsilon_1} {3} 
+ \frac {\epsilon_1} {3} \le \epsilon_1.
\end{split}
\end{equation}
\end{proof}

% The proof of the above lemma is quite standard, and is omitted in
% this paper. Interested readers can find the full proof in our online 
% technical report \cite{ji12tr1}.

Now, we prove Lemma~\ref{lem:ple}. Note that we have
\begin{equation}
\label{eq:causality}
\textstyle A_\lk(t) \le \sum_{s \in \Flow} \Hslk F_s(t) + 
\sum_{i \in \Edge} \sum^{L^{\max}}_{h=1} Q_{i,h}(0),
\end{equation}
for any $t > 0$ and for any link $l \in \Edge$ and $1 \le k 
\le L^{\max}$ due to system causality.

Since the arrival processes satisfy SLLN of type (\ref{eq:slln}),
we obtain from Lemma~\ref{lem:conv} that with probability one,
\begin{equation}
\label{eq:slimit}
\textstyle \lim_{n \rightarrow \infty} \frac {\sum^{n}_{\tau=1} \frac 
{F_s(\tau)} {\tau}} {n} = \lambda_s, ~\text{for all}~ s \in \Flow.
\end{equation}

Note that we will omit the superscript $(\xnj)$ of the random 
variables (depending on the choice of the sequence $\{\xnj\}$) 
throughout the rest of the proof for notational convenience (e.g., 
we use $A_\lk(t)$ to denote $A^{(\xnj)}_\lk(t)$). Then, for all 
regular time $t > 0$, all links $l \in \Edge$ and $1 \le k \le 
L^{\max}$, we have
\[
\begin{split}
&\hat{p}_\lk(t) \\
&= \frac{d}{dt} \textstyle \int_0^t \hat{p}_\lk(\tau) d\tau 
= \lim_{\delta \rightarrow 0} \frac{\textstyle \int_0^{t+\delta} \hat{p}_\lk(\tau) d\tau -
\textstyle \int_0^t \hat{p}_\lk(\tau) d\tau} {\delta} \\
&\stackrel{(\ref{eq:fluid_hp})}= \lim_{\delta \rightarrow 0} \lim_{j \rightarrow \infty} 
\frac{ \sum^{\lfloor (t+\delta)\xnj \rfloor}_{\tau=\lceil t\xnj \rceil} \hat{P}_\lk(\tau)} {\delta \xnj} \\
&\stackrel{(\ref{eq:shadowarr})}= (1+\epsilon) \lim_{\delta \rightarrow 0} \lim_{j \rightarrow \infty} 
\frac{ \sum^{\lfloor (t+\delta)\xnj \rfloor}_{\tau=\lceil t\xnj \rceil} \frac {A_\lk(\tau)} {\tau}} {\delta \xnj} \\
&\stackrel{(\ref{eq:causality})}\le (1+\epsilon) \lim_{\delta \rightarrow 0} \lim_{j \rightarrow \infty} 
\frac{ \sum^{\lfloor (t+\delta)\xnj \rfloor}_{\tau=\lceil t\xnj \rceil} \frac {\sum_s 
\Hslk F_s(\tau) + \sum_i \sum_h Q_{i,h}(0)} {\tau}} {\delta \xnj} \\
&= (1+\epsilon ) \sum_s H^s_\lk \lim_{\delta \rightarrow 0} \lim_{j \rightarrow \infty} 
\frac{ \sum^{\lfloor (t+\delta )\xnj \rfloor}_{\tau=1} \frac {F_s(\tau )} {\tau}} {\lfloor(t+\delta ) 
\xnj\rfloor} \cdot \frac {\lfloor(t+\delta ) \xnj \rfloor} {\delta \xnj} \\
&~~~- (1+\epsilon ) \sum_s H^s_\lk \lim_{\delta \rightarrow 0} \lim_{j \rightarrow \infty}
\frac {\sum^{\lceil t\xnj \rceil - 1}_{\tau=1} \frac {F_s(\tau )} {\tau}} {\lceil t \xnj \rceil - 1} 
\cdot \frac {\lceil t \xnj \rceil - 1} {\delta \xnj} \\
&~~~+ (1+\epsilon) \lim_{\delta \rightarrow 0} \lim_{j \rightarrow \infty} 
\frac {\sum_i \sum_h Q_{i,h}(0)} {\delta \xnj} \cdot \sum^{\lfloor (t+\delta)\xnj \rfloor}_{\tau=\lceil t\xnj \rceil} \frac {1}  {\tau}  \\
&\le (1+\epsilon) \sum_s H^s_\lk \lambda_s \lim_{\delta \rightarrow 0} \left(  
\frac {t+\delta} {\delta} - \frac {t} {\delta} \right) + (1+\epsilon) \frac {1} {t} \\
&= (1+\epsilon) \left( \sum_s H^s_\lk \lambda_s + \frac {1} {t} \right),
\end{split}
\]
% where $(a)$ is from (\ref{eq:fluid_hp}), $(b)$ is from (\ref{eq:shadowarr}), 
% $(c)$ is from (\ref{eq:causality}), the first term 
% of $(d)$ are from (\ref{eq:slimit}), and the second term of $(d)$ is from 
where in the last inequality, the first term is from (\ref{eq:slimit}), 
and the second term is from the fact that: i) $\|q(0)\| + \|\hat{q}(0)\| \le 1$ 
implies $\lim_{j \rightarrow \infty} \frac {\sum_j \sum_h Q_{j,h}(0)} {\xnj} \le 1$; 
and ii) 
\begin{equation}
\label{eq:taulimit}
\begin{split}
& \lim_{j \rightarrow \infty} \int^{\lfloor (t+\delta)\xnj \rfloor}_{\tau=\lceil t\xnj \rceil}
\frac {1} {\tau+1} d\tau 
\le \lim_{j \rightarrow \infty} 
\sum^{\lfloor (t+\delta)\xnj \rfloor}_{\tau=\lceil t\xnj \rceil} \frac {1}  {\tau} 
\le \lim_{j \rightarrow \infty} \int^{\lfloor (t+\delta)\xnj \rfloor}_{\tau=\lceil t\xnj \rceil}
\frac {1} {\tau} d\tau \\
\Longleftrightarrow~~~& \lim_{j \rightarrow \infty} \log \left( \frac {\lfloor (t+\delta)\xnj \rfloor+1}
{\lceil t\xnj \rceil+1} \right) 
\le \lim_{j \rightarrow \infty} 
\sum^{\lfloor (t+\delta)\xnj \rfloor}_{\tau=\lceil t\xnj \rceil} \frac {1}  {\tau} 
\le \lim_{j \rightarrow \infty} \log \left( \frac {\lfloor (t+\delta)\xnj \rfloor}
{\lceil t\xnj \rceil} \right) \\
\Longleftrightarrow~~~& \lim_{j \rightarrow \infty} 
\sum^{\lfloor (t+\delta)\xnj \rfloor}_{\tau=\lceil t\xnj \rceil} \frac {1}  {\tau} 
= \log \frac {t+\delta} {t}.
\end{split}
\end{equation}
Combining i) and ii), we have
%\begin{equation}
\[
\begin{split}
&\textstyle \lim_{\delta \rightarrow 0} \lim_{j \rightarrow \infty}
\frac {\sum_i \sum_h Q_{i,h}(0)} {\delta \xnj} \cdot 
\sum^{\lfloor (t+\delta)\xnj \rfloor}_{\tau=\lceil t\xnj \rceil} \frac {1} {\tau} \\
&~~~~~\textstyle \le \lim_{\delta \rightarrow 0} \left( \frac {1} {\delta} \cdot \log \frac 
{t+\delta} {t} \right)= \frac {1} {t},
\end{split}
\]
%\end{equation} 
where the equality is from the L'Hospital's Rule.

So far, we have shown (\ref{eq:lk}). Note that when $k=1$, Eq.~(\ref{eq:causality}) 
reduces to $A_{l,1}(t) = \sum_{s \in \Flow} H^s_{l,1} F_s(t)$. Then, in the above 
derivation of $\hat{p}_\lk(t)$, the first inequality (which follows from 
(\ref{eq:causality})) becomes an equality and the right-hand side of this inequality 
becomes 
\[
\textstyle (1+\epsilon) \lim_{\delta \rightarrow 0} \lim_{j \rightarrow \infty} 
\frac{ \sum^{\lfloor (t+\delta)\xnj \rfloor}_{\tau=\lceil t\xnj \rceil} 
\frac {\sum_s H^s_{l,1} F_s(\tau)} {\tau}} {\delta \xnj}.
\]
Hence, we obtain (\ref{eq:l1}).
\end{proof}

\emph{Remark:} Lemma~\ref{lem:ple} holds when the exogenous arrival 
processes satisfy the SLLN, and the shadow arrivals are controlled as 
in (\ref{eq:shadowarr}). Note that Lemma~\ref{lem:ple} does not hold 
for data queues $Q_\lk$, since the data arrival processes do not satisfy
(\ref{eq:shadowarr}) due to their dependency on the service of the 
previous hop queues. Lemma~\ref{lem:ple} is important to proving the 
stability of the shadow queues, and implies that in the fluid limit 
model, the instantaneous arrival rate of shadow queues is strictly 
inside the optimal throughput region $\Lambda^*$ after a finite time. 

%The proof is provided in Appendix~\ref{app:lem:ple}.

Then, in the following lemma, we show that the fluid limit model for 
the sub-system consisting of shadow queues is stable\footnote{Similar 
to \cite{andrews04}, we consider a weaker criterion for the stability 
of the fluid limit model in Lemma~\ref{lem:sub-stable}, which can imply 
the stability of the original system from Lemma~\ref{lem:stab_cri}.} 
under HQ-MWS.

\begin{lemma}
\label{lem:sub-stable}
The fluid limit model for the sub-system of shadow queues $\hat{q}$ 
operating under HQ-MWS satisfies that: For any $\zeta>0$, there exists 
a finite $T_1 > 0$ such that for any fluid model solution with 
$\|\hat{q}(0)\| \le 1$, we have that with probability one,
%\begin{equation}
\[
\textstyle \|\hat{q}(t)\| \le \zeta, ~\text{for all}~ t \ge T_1,
\]
%\end{equation}
for any arrival rate vector strictly inside $\Lambda^{*}$. 
\end{lemma}

\begin{proof}
Suppose $\lambda$ is strictly inside $\Lambda^{*}$, we can find a small $\epsilon>0$ 
such that $(1+\epsilon)\lambda$ is strictly inside $\Lambda^{*}$. Then, there 
exists a vector $\phi \in Co(\Matching)$ such that $(1+\epsilon)\lambda < \phi$, 
i.e., $(1+\epsilon)\sum_s \sum_k \Hslk \lambda_s < \phi_l$, for all $l\in\Edge$. 
Let $\beta$ denote the smallest difference between the two vectors, which is defined 
as $\beta \triangleq \min_{l\in\Edge} (\phi_l - (1+\epsilon) \sum_s \sum_k \Hslk \lambda_s)$.
Clearly, we have $\beta>0$. Let $T^{\prime}$ be a finite time such that $T^{\prime} > \frac 
{(1+\epsilon)L^{\max}} {\beta}$, then we have $(1+\epsilon) \left( \sum_s \sum_k 
\Hslk \lambda_s + \frac {L^{\max}} {T^{\prime}} \right)< \phi_l$. Let $\phi_\lk 
\triangleq (1+\epsilon) \left(\sum_s \Hslk \lambda_s + \frac {1} {T^{\prime}} \right)
+ \frac {\phi_l - (1+\epsilon) \left( \sum_k \sum_s \Hslk \lambda_s + \frac {L^{\max}} 
{T^{\prime}} \right) } {L^{\max}}$.
% \[
% \begin{split}
% &\textstyle  \phi_\lk \triangleq (1+\epsilon) \left(\sum_s \Hslk \lambda_s + \frac {1} {T^{\prime}} \right) \\ 
% &~~~~~~~~~\textstyle  + \frac {\phi_l - (1+\epsilon) \left( \sum_k \sum_s \Hslk \lambda_s 
% + \frac {L^{\max}} {T^{\prime}} \right) } {L^{\max}},
% \end{split}
% \]
Then, we have
\begin{equation}
\label{eq:phiphi}
\textstyle \sum_k \phi_\lk = \phi_l,
\end{equation}
and from (\ref{eq:lk}), we have
\begin{equation}
\label{eq:lp}
%\begin{split}
\textstyle \hat{p}_\lk(t) \le (1+\epsilon) \left(\sum_s \Hslk \lambda_s 
+ \frac {1} {T^{\prime}} \right) < \phi_\lk,
%\end{split}
\end{equation}
for all regular time $t \ge T^{\prime}$. This implies that the 
instantaneous arrival rate of shadow queues is strictly inside 
the optimal throughput region $\Lambda^*$.

We consider a quadratic-form Lyapunov function $\hat{V}(\hat{q}(t)) = \frac 1 2 
\sum_l \sum_k (\hat{q}_\lk(t))^2$. It is sufficient to show that for any $\zeta_1 
> 0$, there exist $\zeta_2>0$ and a finite time $T^*>0$ such that at any regular 
time $t \ge T^*$, $\hat{V}(\hat{q}(t)) \ge \zeta_1$ implies $\frac{D^+}{dt^+} 
\hat{V}(\hat{q}(t)) \le -\zeta_2$. Since $\hat{q}(t)$ is differentiable for any 
regular time $t \ge T^{\prime}$, we can obtain the derivative of $\hat{V}(\hat{q}(t))$ as
\begin{equation}
\label{eq:dol}
\begin{split}
%\begin{eqnarray}
\textstyle \frac{D^+}{dt^+} \hat{V}(\hat{q}(t))
=&\textstyle \sum_l \sum_k \hat{q}_\lk(t) \cdot \left(\hat{p}_\lk(t) - \hat{\pi}_\lk(t)\right)  \\
%\stackrel{(b)}\le& \sum_l \sum_k \hat{q}_\lk(t) \cdot \left((1+\epsilon) \sum_s \Hslk \lambda_s - \hat{\pi}_\lk(t)\right)  \\
=&\textstyle  \sum_l \sum_k \hat{q}_\lk(t) \cdot \left(\hat{p}_\lk(t) - \phi_\lk\right) \\
&\textstyle + \sum_l \sum_k \hat{q}_\lk(t) \cdot \left(\phi_\lk - \hat{\pi}_\lk(t)\right),
%\end{eqnarray}
\end{split}
\end{equation}
where $\frac{D^+}{dt^+} \hat{V}(\hat{q}(t)) = \lim_{\delta \downarrow 0} \frac 
{\hat{V}(\hat{q}(t+\delta)) - \hat{V}(\hat{q}(t))} {\delta}$, and the first 
equality is from (\ref{eq:hdq2}).

Let us choose $\zeta_3>0$ such that $\hat{V}(\hat{q}(t)) \ge \zeta_1$ implies 
$\max_{l \in \Edge, 1 \le k \le L^{\max}} \hat{q}_\lk(t) \ge \zeta_3$. Then in 
the final result of (\ref{eq:dol}), we can conclude that the first term is 
bounded. That is,
\[
\begin{split}
&\textstyle \sum_l \sum_k \hat{q}_\lk(t) \cdot \left(\hat{p}_\lk(t) - \phi_\lk\right)
\le -\zeta_3 \min_{l,k} (\phi_\lk - \hat{p}_\lk(t)) \\
& \textstyle \le -\zeta_3 \min_{l,k} (\phi_\lk - (1+\epsilon) (\sum_s \Hslk \lambda_s 
+ \frac {1} {T^{\prime}}) ) \triangleq -\zeta_2 < 0,
\end{split}
\]
where the second inequality is from (\ref{eq:lp}). For the second term, 
since HQ-MWS chooses schedules that maximize the shadow queue length 
weighted rate, the service rate satisfies that
\begin{equation}
\label{eq:piphi}
\textstyle \hat{\pi}(t) \in \argmax_{\phi \in 
Co(\Matching)} \sum_l \hat{q}_{l,k^*(l)}(t) \cdot \phi_l,
\end{equation}
where i) $\hat{q}_{l,k^*(l)}(t) = \max_k \hat{q}_\lk(t)$, and ii) $\hat{\pi}_l(t) = 
\sum_k \hat{\pi}_\lk(t)$ with $\hat{\pi}_\lk(t)=0$ when $\hat{q}_\lk(t) < 
\hat{q}_{l,k^*(l)}(t)$. This implies that
%\[
%\begin{split}
$\sum_l \sum_k  \hat{q}_\lk(t) \cdot \phi_\lk 
\le \sum_l \sum_k \hat{q}_{l,k^*}(t) \cdot \phi_\lk  
= \sum_l \hat{q}_{l,k^*}(t) \cdot \phi_l
\le \sum_l \hat{q}_{l,k^*}(t) \cdot \hat{\pi}_l(t)  
= \sum_l \sum_k \hat{q}_\lk(t) \cdot \hat{\pi}_\lk(t)$, 
%\end{split}
%\]
for all $\phi \in Co(\Matching)$, where the first equality and the second
inequality are from (\ref{eq:phiphi}) and (\ref{eq:piphi}), respectively. 
Then, we obtain that the second term of (\ref{eq:dol}) is non-positive. 
This shows that $\hat{V}(\hat{q}(t)) \ge \zeta_1$ implies $\frac{D^+}{dt^+} 
\hat{V}(\hat{q}(t)) \le -\zeta_2$ for all regular time $t \ge T^*$. Hence, 
it immediately follows that for any $\zeta>0$, there exists a finite $T_1 
\ge T^* > 0$ such that $\|\hat{q}(t)\| \le \zeta$, for all $t \ge T_1$.
\end{proof}

We next present Lemma~\ref{lem:p} that is used to show that the sub-system 
consisting of data queues is stable under HQ-MWS in the fluid limit model. 
\begin{lemma}
\label{lem:p}
If data queues $q_{l,j}$ are stable for all $l \in \Edge$ and for all $j \le k$, 
then there exists a finite $T^k_1>0$ such that for all regular time $t\ge T^k_1$ 
and for all $l \in \Edge$, we have that with probability one,
%\begin{equation}
\[
\textstyle \hat{p}_{l,k+1}(t) \ge (1+\epsilon) \sum_s H^s_{l,k+1} \lambda_s.
\]
%\end{equation}
\end{lemma}
The proof follows a similar argument used in the proof for Lemma~\ref{lem:ple},
and is referred to Appendix~\ref{app:lem:p}.

In the following lemma, using a hop-by-hop inductive argument, we show 
that the fluid model for the sub-system of data queues is stable. 

\begin{lemma}
\label{lem:stable}
The fluid limit model of the sub-system of data queues $q$ operating 
under HQ-MWS is stable, i.e., there exists a finite $T_2 > 0$ such 
that, for any fluid model solution with $\|q(0)\| \le 1$, we have
%\begin{equation}
\[
\textstyle \|q(t)\| = 0, ~\text{for all}~ t \ge T_2,
\]
%\end{equation}
for any arrival rate vector strictly inside $\Lambda^{*}$. 
\end{lemma}

\begin{proof}
We prove the stability of data queues by induction.

Suppose $\lambda$ is strictly inside $\Lambda^{*}$, the sub-system 
of shadow queues $\hat{q}$ is stable from Lemma~\ref{lem:sub-stable}. 
Let us choose sufficiently small $\zeta>0$ such that $\zeta < \epsilon
\min_s \lambda_s$, then there exists a finite time $T_1>0$ such that 
we have $\| \hat{q}(t) \| \le \zeta$ for any regular time $t \ge T_1$. 
Thus, we have $\hat{\psi}_\lk(t) \ge \hat{p}_\lk(t) - \zeta$ from 
(\ref{eq:hdq1}), for all $t \ge T_1$. Hence, for all data queues and 
all regular time $t \ge T_1$, we have
\begin{equation}
\label{eq:pilambda-2}
\textstyle \pi_\lk(t) = \hat{\pi}_\lk(t) \ge \hat{p}_\lk(t) - \zeta,
\end{equation}
from (\ref{eq:pipi}) and (\ref{eq:hpp}). 

Now we show by induction that all data queues are stable in the 
fluid limit model.

\noindent {\bf Base Case:} 

First, note that $\pi_{l,1}(t) \ge (1+\epsilon) \sum_{s} H^s_{l,1} 
\lambda_s - \zeta$ from (\ref{eq:l1}) and (\ref{eq:pilambda-2}). 
Consider a sub-system that contains only queue $q_{l,1}$. From 
$p_{l,1}(t) = \sum_{s} H^s_{l,1} \lambda_s$ and (\ref{eq:dq2}), 
we have $\frac{d}{dt} q_{l,1}(t) = p_{l,1}(t) - \pi_{l,1}(t) \le 
- \epsilon \sum_{s} H^s_{l,1} \lambda_s + \zeta < 0$, 
if $q_{l,1}(t) > 0$. This implies that the sub-system that 
contains only $q_{l,1}$ is stable, for all $l \in \Edge$.

\noindent {\bf Induction Step:} 

Next, we show that, if $q_{l,j}$ is stable for all $l \in \Edge$ 
and all $j \le k$, then each queue $q_{l,k+1}$ is also stable for 
all $l \in \Edge$, where $1 \le k < L^{\max}$.

Since $q_{l,j}(t)$ is stable for all $l \in \Edge$ and all $j \le k$, 
i.e., there exists a finite $T^k_1 > 0$ such that $q_{l,j}(t) = 0$ for 
all regular time $t \ge T^k_1$, then $u_{s,k+1}(t) = u_\sk(t) + q_\sk(0) 
= \cdots = u_{s,1}(t) + \sum_{h \le k}q_{s,h}(0) = \lambda_s t + 
\sum_{h \le k}q_{s,h}(0)$ for all $s \in \Flow$ and for all regular 
time $t\ge T^k_1$. Thus, we have 
%\begin{equation}
%\label{eq:alambda}
%\textstyle 
$a_{l,k+1}(t) = t \sum_s H^s_{l,k+1} \lambda_s + \sum_s H^s_{l,k+1} \sum_{h \le k}q_{s,h}(0)$
%\end{equation}
from (\ref{eq:ar}), and $p_{l,k+1}(t) = \sum_s H^s_{l,k+1} \lambda_s$ 
from (\ref{eq:ap}) by taking derivative, for all $l \in \Edge$ and all 
regular time $t\ge T^k_1$. Then, note that we have $\hat{p}_{l,k+1}(t) 
\ge (1+\epsilon) \sum_s H^s_{l,k+1} \lambda_s$ from Lemma~\ref{lem:p}. 
Hence, we have $\pi_{l,k+1}(t) \ge (1+\epsilon) \sum_s H^l_{s,k+1}
\lambda_s - \zeta$ from (\ref{eq:pilambda-2}). Therefore, we have 
$\frac{d}{dt} q_{l,k+1}(t) = p_{l,k+1}(t) - \pi_{l,k+1}(t) \le -\epsilon 
\sum_s H^l_{s,k+1} \lambda_s + \zeta < 0$, if $q_{l,k+1}(t) > 0$. This 
implies that $q_{l,k+1}$ is stable for all $l \in \Edge$.

Therefore, the result follows by induction.
\end{proof}

The following lemma says that the fluid limit model of joint data
queues and shadow queues is stable, which follows immediately from 
Lemmas~\ref{lem:sub-stable} and \ref{lem:stable}.
\begin{lemma}
\label{lem:joint-stable}
The fluid limit model of the joint system of data queues $q$ and
shadow queues $\hat{q}$ operating under HQ-MWS satisfies that:
For any $\zeta>0$, there exists a finite $T_2 > 0$ such that for 
any fluid model solution with $\|q(0)\|+\|\hat{q}(0)\| \le 1$, 
we have that with probability one,
%begin{equation}
\[
\|q(t)\|+\|\hat{q}(t)\| \le \zeta, ~\text{for all}~ t \ge T_2,
\]
%\end{equation}
for any arrival rate vector strictly inside $\Lambda^{*}$. 
\end{lemma}

Now, consider any fixed sequence of processes $\{ \frac {1} {x} \Markov^{(x)}(xt), 
x=1,2,\cdots\}$ (for simplicity also denoted by $\{x\}$). By Lemmas~\ref{lem:fluid} 
and \ref{lem:joint-stable}, we have that for any fixed $\xi_1>0$, we can always 
choose a large enough integer $T>0$ such that for any subsequence $\{\xn\}$ of 
$\{x\}$, there exists a further (sub)subsequence $\{\xnj\}$ such that
%\begin{equation}
%\label{eq:qs}
\[
\begin{split}
\textstyle \lim_{j \rightarrow \infty} & \textstyle \frac {1} {\xnj} (\|Q^{(\xnj)}(\xnj T)\|
+ \lceil \|\hat{Q}^{(\xnj)}(\xnj T)\| \rceil) \\
& = \|q(T)\|+\|\hat{q}(T)\| \le \xi_1
\end{split}
\]
%\end{equation}
almost surely. This, along with (\ref{eq:A0}), implies that 
\[
\textstyle \lim_{j \rightarrow \infty} \frac {1} {\xnj} \| \Markov^{(\xnj)} (\xnj T) \| 
\le \xi_1 
\]
almost surely, which in turn implies (for small enough $\zeta_1$) that
\begin{equation}
\label{eq:mc_conv}
\textstyle \limsup_{x \rightarrow \infty} \frac {1} {x} \| \Markov^{(x)} (x T) \| 
\le \xi_1 \triangleq 1 - \xi < 1
\end{equation}
almost surely. This is because there must exist a subsequence of $\{ x \}$ that 
converges to the same limit as $\limsup_{x \rightarrow \infty} \frac {1} {x} \| 
\Markov^{(x)} (x T) \|$.

Next, we will show that the sequence $\{ \frac {1} {x} \| \Markov^{(x)} (x T) \|, 
x=1,2,\cdots \}$ is uniformly integrable.
Note that link capacities are all finite (equals one, as we assumed in the system 
model), then for all time slots $t>0$, we have that 
\begin{equation}
\label{eq:hp_bound}
\textstyle \hat{P}_\lk(t) = (1+\epsilon) \frac {A_\lk(t)} {t} 
\le (1+\epsilon) \frac {|\Edge| t + \sum_{s \in \Flow} F_s(t)} {t},
\end{equation}
for all $l$ and $k$. Define a random variable 
\[
\textstyle
\Theta(T) \triangleq \frac {1} {x} 
\left( (1+|\Edge| \cdot L^{\max})(x + \sum_s F^{(x)}_s(xT)) + \sum_l \sum_k \sum^{xT}_{\tau = 1} 
\hat{P}^{(x)}_\lk(\tau) + 2 \right). 
\]
Note that we have
\[
\textstyle \sum_l \sum_k Q^{(x)}_\lk(xT) \le x + \sum_s F^{(x)}_s(xT),
\]
\[
\textstyle \hat{Q}^{(x)}_\lk(xT) \le \sum^{xT}_{\tau = 1} \hat{P}^{(x)}_\lk(\tau)
\]
and 
\[
\textstyle A^{(x)}_\lk(xT) \le x + \sum_s F^{(x)}_s(xT).
\]
Then, we have
\[
\textstyle \frac {1} {x} \| \Markov^{(x)} (x T) \| 
= \textstyle \frac {1} {x} ( \sum_l \sum_k Q^{(x)}_\lk(xT) 
+ \lceil \sum_l \sum_k \hat{Q}^{(x)}_\lk(xT) \rceil 
+ \lceil \frac {1} {xT+1} \sum_l \sum_k A^{(x)}_\lk(xT)) \rceil 
\le \Theta(T),
\]
and
\[
\begin{split}
\Expect [\Theta(T)] &\le  \textstyle \frac {1} {x} \left( (1+|\Edge| \cdot L^{\max}) (x + \sum_s \lambda_s xT) + (1+\epsilon) \sum_l \sum_k \sum^{xT}_{\tau = 1} (|\Edge| + \sum_s \lambda_s)  + 2 \right) \\
& \textstyle \le \frac {1} {x} \left( x(1+|\Edge| \cdot L^{\max})(1 + T \sum_s \lambda_s)  +  (1+\epsilon) xT \cdot |\Edge| \cdot L^{\max} (|\Edge|+\sum_s \lambda_s) + 2 \right) \\
& \textstyle \le (1+|\Edge| \cdot L^{\max})(1 + T \sum_s \lambda_s) + (1+\epsilon) T \cdot |\Edge| \cdot L^{\max} (|\Edge|+\sum_s \lambda_s)  + 2 \\
& < \infty,
\end{split}
\]
where the first inequality is from (\ref{eq:hp_bound})
and the assumption on our arrival processes. 

Therefore, it follows from the Dominated Convergence Theorem that the 
sequence $\{ \frac {1} {x} \| \Markov^{(x)} (x T) \|, x=1,2,\cdots \}$ 
is uniformly integrable. Then, the almost surely convergence in 
(\ref{eq:mc_conv}) along with uniform integrability implies the following 
convergence in the mean:
\[
\textstyle \limsup_{x \rightarrow \infty} \Expect [ \frac {1} {x} \| \Markov^{(x)} (x T) \| ] \le 1-\xi.
\] 

Since the above convergence holds for any sequence of processes 
$\{\frac {1} {x} \Markov^{(x)}(xT), x=1,2,\cdots\}$, the condition 
of type (\ref{eq:stab_cri}) in Lemma~\ref{lem:stab_cri} is satisfied. 
This completes the proof of Proposition~\ref{pro:hq}.
}

%%%%%%%%%%%%%%%%%%%%%%%%%%%%%%%%%%%%%%%%%%%%%%%%%%%%%%%%%%%%%%
\section{Proof of Lemma~\ref{lem:fluid}} \label{app:lem:fluid}
%%%%%%%%%%%%%%%%%%%%%%%%%%%%%%%%%%%%%%%%%%%%%%%%%%%%%%%%%%%%%%
\high{
First, we prove the convergence and continuity properties for the processes associated 
with data queues. It follows from the strong law of large numbers that $\frac{1}{\xn}F^{(\xn)}_s(\xn t) 
\rightarrow  \lambda_s t$, hence, the convergence (\ref{eq:fluid_f}) holds, and each 
of the limiting functions $f_s$ is Lipschitz continuous. Also, note that for any fixed 
$0 \le t_1 \le t_2$, due to finite link capacities (in particular, all equal to one under
our unit capacity assumption), we have that
\begin{equation}
\label{eq:lipschitz}
\textstyle \frac {1} {\xn} \left(D^{(\xn)}_{\lk}(\xn t_2) - D^{(\xn)}_{\lk}(\xn t_1) \right) \le t_2 - t_1.
\end{equation}
Thus, the sequence of functions $\{\frac {1} {\xn} D^{(\xn)}_{\lk}(\xn \cdot)\}$
is uniformly bounded and uniformly equicontinuous. Consequently, by the Arzela-Ascoli 
Theorem, there must exist a subsequence under which (\ref{eq:fluid_d}) holds. Note that 
(\ref{eq:lipschitz}) also implies that each of the limiting functions $d_\lk$ is Lipschitz
continuous. Recall that $U_\sk(t)$ denotes the cumulative number of packets transmitted 
from the $(k-1)$-st hop to the $k$-th hop for flow $s$ up to time slot $t$, then convergence 
(\ref{eq:fluid_u}) holds similarly as (\ref{eq:fluid_d}) for $k>1$, and holds from 
(\ref{eq:fluid_f}) for $k=1$. Hence, convergence (\ref{eq:fluid_a}) trivially follows 
from the definition of $A_\lk(t)$ and (\ref{eq:fluid_u}). Similarly, each of the limiting 
functions $u_\sk$ and $a_\lk$ is Lipschitz continuous.

Since the sequence $\{\frac {1} {\xn} Q^{(\xn)}_{\lk}(0)\}$ are bounded by 1 from
(\ref{eq:init_config}), there exists a further subsequence (of the subsequence already
chosen above, and for simplicity still denoted by $\xnj$) such that 
$\frac{1}{\xnj}Q^{(\xnj)}_\lk(0) \rightarrow  q_\lk(0)$. Hence, convergence 
(\ref{eq:fluid_q}) trivially follows from the queue evolution equation (\ref{eq:q_evolution}) 
and convergences (\ref{eq:fluid_a}) and (\ref{eq:fluid_d}). Also, it follows that
each of the limiting functions $q_\lk$ is Lipschitz continuous.

Recall that $\Psi_\lk(t) = D_\lk(t) - D_\lk(t-1)$ and $P_\lk(t) = A_\lk(t) - A_\lk(t-1)$, 
hence, the sequences $\{\frac{1}{\xnj}\int_0^{\xnj t} \Psi^{(\xnj)}_\lk(\tau)d\tau\}$
and $\{\frac{1}{\xnj}\int_0^{\xnj t} P^{(\xnj)}_\lk(\tau)d\tau \}$ are identical to the 
sequences $\{\frac{1}{\xnj}D^{(\xnj)}_\lk(\xnj t)\}$ and $\{\frac{1}{\xnj}A^{(\xnj)}_\lk(\xnj t)\}$, 
respectively. This in turn implies that the convergences (\ref{eq:fluid_psi}) and 
(\ref{eq:fluid_d}) hold, where $\int_0^t \psi_\lk(\tau) d\tau = d_\lk(t)$ and $\int_0^t 
p_\lk(\tau) d\tau = a_\lk(t)$. The convergence (\ref{eq:fluid_pi}) follows from an 
inequality similar to (\ref{eq:lipschitz}) by applying the Arzela-Ascoli Theorem.

Using similar arguments, we can prove the results for the processes associated with 
the shadow queues. This completes the proof of Lemma~\ref{lem:fluid}.
}

%%%%%%%%%%%%%%%%%%%%%%%%%%%%%%%%%%%%%%%%%%%%%%%%%%%%%%
\section{Proof of Lemma~\ref{lem:p}} \label{app:lem:p}
%%%%%%%%%%%%%%%%%%%%%%%%%%%%%%%%%%%%%%%%%%%%%%%%%%%%%%
Note that the total number of packets waiting in the 
previous hops for $Q_{l,k+1}$ at time slot $t$ is no 
greater than $\sum_i \sum_{h \le k} Q_{i,h}(t)$. Then, 
we have
\begin{equation}
\label{eq:afq}
\textstyle A_{l,k+1}(t) \ge \sum_s H^s_{l,k+1} F_s(t) - \sum_i 
\sum_{h \le k} Q_{i,h}(t).
\end{equation}
Since $q_{i,h}$ is stable for all $i \in \Edge$ and all $h \le k$, 
there exists a finite $T^k_1>0$ such that $\sum_i \sum_{h \le k} 
q_{i,h}(t) = 0$, for all regular time $t \ge T^k_1$. Let $\delta>0$ 
be fixed, and consider all times $\nu \in [t,t+\delta]$, where $t 
\ge T^k_1$. Recall that $\xnj$ is a positive subsequence for which 
the convergence to the fluid limit holds \emph{u.o.c}. For an 
arbitrary $\theta>0$, there exists a large enough $j$ so that
\begin{equation}
\label{eq:theta1}
\left|\frac {\sum_i \sum_{h \le k} Q_{i,h}(\xnj \nu)} {\xnj} 
- \sum_i \sum_{h \le k} q_{i,h}(\nu) \right| < \theta,
\end{equation}
for all $\nu \in [t,t+\delta]$.

Consider time slots $\Upsilon \triangleq \{ \lceil \xnj t \rceil,
\lceil \xnj t \rceil + 1, \cdots, \lfloor \xnj (t+\delta) 
\rfloor \}$. Eq.~(\ref{eq:theta1}) can be rewritten as
\begin{equation}
\label{eq:theta2}
\textstyle \sum_i \sum_{h \le k} Q_{i,h}(\tau) < \theta \xnj,
\end{equation}
for all time slots $\tau \in \Upsilon$. Then for all $t \ge T^k_1$
and all $l \in \Edge$, we have
\[
\begin{split}
\hat{p}_{l,k+1}(t) &= \frac{d}{dt} \textstyle \int_0^t \hat{p}_{l,k+1}(\tau) d\tau 
= \lim_{\delta \rightarrow 0} \frac{\textstyle \int_0^{t+\delta} \hat{p}_{l,k+1}(\tau) d\tau -
\textstyle \int_0^t \hat{p}_{l,k+1}(\tau) d\tau} {\delta} \\
&= \lim_{\delta \rightarrow 0} \lim_{\xnj \rightarrow \infty} 
\frac{ \sum^{\lfloor (t+\delta)\xnj \rfloor}_{\tau=\lceil t\xnj \rceil} \hat{P}_{l,k+1}(\tau)} {\delta \xnj} \\
&\stackrel{(a)}= (1+\epsilon) \lim_{\delta \rightarrow 0} \lim_{\xnj \rightarrow \infty} 
\frac{ \sum^{\lfloor (t+\delta)\xnj \rfloor}_{\tau= \lceil t\xnj \rceil} \frac {A_{l,k+1}(\tau)} {\tau}} {\delta \xnj} \\
&\stackrel{(b)}\ge (1+\epsilon) \lim_{\delta \rightarrow 0} \lim_{\xnj \rightarrow \infty} 
\frac{ \sum^{\lfloor (t+\delta)\xnj \rfloor}_{\tau=\lceil t\xnj \rceil} \frac {\sum_s H^s_{l,k+1} F_s(\tau) 
- \sum_i \sum_{h \le k} Q_{i,h}(\tau)} {\tau}} {\delta \xnj} \\
&\stackrel{(c)}> (1+\epsilon ) \sum_s H^s_{l,k+1} \lim_{\delta \rightarrow 0} \lim_{\xnj \rightarrow \infty} 
\frac{ \sum^{\lfloor (t+\delta )\xnj \rfloor}_{\tau=1} \frac {F_s(\tau )} {\tau}} {\lfloor(t+\delta ) 
\xnj\rfloor} \cdot \frac {\lfloor(t+\delta ) \xnj \rfloor} {\delta \xnj} \\
&~~~- (1+\epsilon ) \sum_s H^s_{l,k+1} \lim_{\delta \rightarrow 0} \lim_{\xnj \rightarrow \infty}
\frac {\sum^{\lceil t\xnj \rceil - 1}_{\tau=1} \frac {F_s(\tau )} {\tau}} {\lceil t \xnj \rceil - 1} 
\cdot \frac {\lceil t \xnj \rceil - 1} {\delta \xnj} \\
&~~~- (1+\epsilon) \lim_{\delta \rightarrow 0} \lim_{\xnj \rightarrow \infty} 
\left( \frac {1} {\delta} \cdot  \frac {\theta \xnj} {\xnj} \cdot \sum^{\lfloor (t+\delta)\xnj \rfloor}_{\tau=\lceil t\xnj \rceil} \frac {1}  {\tau} \right)  \\
&\stackrel{(d)}= (1+\epsilon) \sum_s H^s_{l,k+1} \lambda_s \lim_{\delta \rightarrow 0} \left(  
\frac {t+\delta} {\delta} - \frac {t} {\delta} \right) - (1+\epsilon) \lim_{\delta \rightarrow 0} 
\left( \frac {\theta} {\delta} \cdot \log \frac {t+\delta} {t} \right) \\
&= (1+\epsilon) \left( \sum_s H^s_{l,k+1} \lambda_s - \frac {\theta} {t} \right),
\end{split}
\]
where $(a)$, $(b)$ and $(c)$ are from (\ref{eq:shadowarr}), (\ref{eq:afq})
and (\ref{eq:theta2}), respectively, and $(d)$ is from (\ref{eq:slimit}) 
and (\ref{eq:taulimit}).

Since $\theta>0$ can be arbitrary, we complete the proof by letting $\theta \rightarrow 0$.

%%%%%%%%%%%%%%%%%%%%%%%%%%%%%%%%%%%%%%%%%%%%%%%%%%%%%%%%%%%%%%%%%%%%%%%%%%%%%%%%
\section{Stability of the Shadow Queues under LQ-MWS} \label{app:shadow-stab-lq}
%%%%%%%%%%%%%%%%%%%%%%%%%%%%%%%%%%%%%%%%%%%%%%%%%%%%%%%%%%%%%%%%%%%%%%%%%%%%%%%%
Similarly to (\ref{eq:fluid_f})-(\ref{eq:fluid_hp}), we can establish the fluid 
limits of the system: $(f, u, q, \pi, \psi, a, d, p, \hat{q}, \hat{\pi}, \hat{\psi}$, $
\hat{a}$, $\hat{d}, \hat{p})$, and we have the following fluid model equations:
\begin{eqnarray}
&& \textstyle	f_s(t) = \lambda_s t, \label{eq:lflt} \\
&& \textstyle	q_l(t) = q_l(0) + a_l(t) - d_l(t), \label{eq:lqad} \\
&& \textstyle	a_l(t) = \sum_s \sum_k \Hslk u_\sk(t), \label{eq:lar} \\
&& \textstyle	a_l(t) = \int_0^t p_l(\tau) d\tau, \label{eq:lap} \\
&& \textstyle	d_l(t) = \int_0^t \psi_l(\tau) d\tau, \label{eq:ldp} \\
&& \textstyle	\psi_l(t) \le \pi_l(t), \label{eq:lpp} \\
&& \textstyle	\frac{d}{dt} q_l(t) = p_l(t) - \psi_l(t),  \label{eq:ldq1} \\
&& \textstyle	\frac{d}{dt} q_l(t) = \left\{
\begin{array}{ll}
p_l(t) - \pi_l(t), &~\text{if}~q_l(t) > 0,\\
(p_l(t) - \pi_l(t))^+, &~\text{otherwise},
\end{array}
\right. \label{eq:ldq2} \\
%\end{eqnarray}
%\begin{eqnarray}
&& \textstyle	\hat{q}_l(t) = \hat{q}_l(0) + \hat{a}_l(t)  - \hat{d}_l(t),  \label{eq:lhqad}\\
&& \textstyle	\hat{a}_l(t) = \int_0^t \hat{p}_l(\tau) d\tau, \\
&& \textstyle	\hat{d}_l(t) = \int_0^t \hat{\psi}_l(\tau) d\tau, \\
&& \textstyle	\hat{\psi}_l(t) \le \hat{\pi}_l(t),  \label{eq:lhpp} \\
&& \textstyle	\frac{d}{dt} \hat{q}_l(t) = \hat{p}_l(t) - \hat{\psi}_l(t), \label{eq:lhdq1}\\
&& \textstyle	\frac{d}{dt} \hat{q}_l(t) = \left\{
\begin{array}{ll}
\hat{p}_l(t) - \hat{\pi}_l(t), &~\text{if}~\hat{q}_l(t) > 0,\\
(\hat{p}_l(t) - \hat{\pi}_l(t))^+, &~\text{otherwise},
\end{array}
\right. \label{eq:lhdq2} \\
&& \textstyle	\|q(0)\| + \| \hat{q}(0)\| = 1, \label{eq:linit} \\
&& \textstyle	\pi_l(t) = \hat{\pi}_l(t). \label{eq:lpipi}
\end{eqnarray}

We present a lemma similar to Lemma~\ref{lem:ple}. This will be used 
to show that the fluid limit model for the sub-system consisting of 
shadow queues is stable under LQ-MWS. We omit its proof since it follows 
the same line of analysis for the proof of Lemma~\ref{lem:ple}.

\begin{lemma}
\label{lem:lple}
For all (scaled) time $t>0$ and for all links $l \in \Edge$, we have that 
with probability one,
\begin{equation}
\label{eq:lple}
\textstyle \hat{p}_l(t) \le (1+\epsilon) \left( \sum_s \sum_k \Hslk \lambda_s 
+ \frac {1} {t} \right).
\end{equation}
\end{lemma}

Now, we can show that the fluid limit model for the sub-system of shadow 
queues $\hat{q}$ is stable under LQ-MWS.  

\begin{lemma}
\label{lem:sub-stable-lq}
The fluid limit model for the sub-system of shadow queues $\hat{q}$ 
operating under LQ-MWS satisfies that: For any $\zeta>0$, there exists 
a finite $T_3 > 0$ such that for any fluid model solution with 
$\|\hat{q}(0)\| \le 1$, we have that with probability one,
\begin{equation}
\|\hat{q}(t)\| \le \zeta, ~\text{for all}~ t \ge T_3,
\end{equation}
for any arrival rate vector strictly inside $\Lambda^{*}$. 
\end{lemma} 

The proof is similar to that of Lemma~\ref{lem:sub-stable} and is thus omitted.

%%%%%%%%%%%%%%%%%%%%%%%%%%%%%%%%%%%%%%%%%%%%%%%%%%%%%%%%%%%%%%
\section{Proof of Proposition~\ref{pro:plq}} \label{app:pro:plq}
%%%%%%%%%%%%%%%%%%%%%%%%%%%%%%%%%%%%%%%%%%%%%%%%%%%%%%%%%%%%%%
To show the stability of the network under PLQ-MWS, it is enough to
show that the fluid limit model of the joint system of data queues
and shadow queues is stable. Since the fluid limit model for the 
sub-system of shadow queues is stable from Lemma~\ref{lem:sub-stable-lq},
it remains to show that the fluid model for the sub-system of data
queues is stable, i.e., it is equivalent to show that all the sub-queues 
for hop-class $k$ packets are stable for each $1\le k \le L^{\max}$. 
We will prove the stability of sub-queues via a hop-by-hop inductive 
argument. 

Let $Q_\lk(t)$ denote the number of packets of hop-class $k$ 
at $Q_l$ at time slot $t$, and let $A_\lk(t)$, $D_\lk(t)$, $\Pi_\lk(t)$, 
$\Psi_\lk(t)$ and $P_\lk(t)$ denote the cumulative arrival, cumulative 
departure, service, departure and arrival for packets of hop-class $k$ 
at $Q_l$, respectively. As before, we establish the fluid limits of the 
system, and obtain (\ref{eq:lflt})-(\ref{eq:lpipi}) and the following 
additional fluid model equations: for all (scaled) time $t \ge 0$,
\begin{eqnarray}
&& \textstyle	a_\lk(t) = \sum_s \Hslk u_\sk(t), \label{eq:lkar} \\
% \end{eqnarray}
% \begin{eqnarray}
&& \textstyle	a_\lk(t) = \int_0^t p_\lk(\tau) d\tau, \label{eq:lkap} \\
&& \textstyle	\frac{d}{dt} q_\lk(t) = p_\lk(t) - \psi_\lk(t), \label{eq:lkdp1} \\
&& \textstyle	\frac{d}{dt} q_\lk(t) = \left\{
\begin{array}{ll}
p_\lk(t) - \pi_\lk(t), &~\text{if}~q_\lk(t) > 0, \\
(p_\lk(t) - \pi_\lk(t))^+, &~\text{otherwise}.
\end{array}
\right. \label{eq:lkdq2}
\end{eqnarray}

Clearly, packets of hop-class $k$ at link $l$ will not be transmitted 
under PLQ-MWS unless link $l$ is active at time slot $t$ and $\sum_{j<k} 
Q_{l,j}(t) < c_l$ (Equivalently, $Q_{l,j}(t)=0$ for all $j<k$ in our 
setting, since $c_l =1$.), i.e., for all $1 \le k \le L^{\max}$, we have
\begin{equation}
\textstyle \Pi_\lk(t)=\left(\Pi_l(t) - \sum_{j<k} Q_{l,j}(t) \right)^+,
\end{equation}
where $\Pi_l(t)=1$, if link $l$ is active at time slot $t$, and
$\Pi_l(t)=0$, otherwise. Hence, we have an additional fluid model 
equation as follows: 
\begin{equation}
\label{eq:pilk}
\textstyle \pi_\lk(t) = \pi_l(t) - \sum_{j<k} \psi_{l,j}(t), 
\end{equation}
for all $1 \le k \le L^{\max}$, and in particular, we have 
\begin{equation}
\label{eq:pil1}
\textstyle \pi_{l,1}(t)=\pi_l(t),
\end{equation}
for all $l\in\Edge$ and for all $t\ge0$.

From Lemma~\ref{lem:sub-stable-lq}, the fluid limit model for the 
sub-system consisting of shadow queues is stable, i.e., there exists 
a finite $T_3 > 0$ such that, for all $l\in\Edge$ and for all time $t\ge T_3$,
\begin{equation}
\label{eq:pipip}
\textstyle \pi_l(t) = \hat{\pi}_l(t) \ge \hat{p}_l(t).
\end{equation}

Next, we show the stability of sub-queues by induction. 

\noindent {\bf Base Case:} 

We first show that sub-queues $q_{l,1}$ are stable for all $l\in\Edge$. 
Note that $\Expect [\hat{P}_l(t)]=(1+\epsilon) \frac {A_l(t)} {t}\ge 
(1+\epsilon) \frac {\sum_s H^s_{l,1} F_s(t)} {t}$, and following the 
same line of analysis for the proof of Lemma~\ref{lem:p}, we show that,
%\begin{equation}
%\label{eq:lpge}
\[
\textstyle \hat{p}_l(t) \ge (1+\epsilon) \sum_s H^s_{l,1} \lambda_s,
\]
%\end{equation}
for all $t\ge0$. This, along with (\ref{eq:pil1}) and (\ref{eq:pipip}), 
implies that 
%\begin{equation}
\[
\textstyle \pi_{l,1}(t)\ge (1+\epsilon) \sum_s H^s_{l,1} \lambda_s,
\]
%\end{equation}
for all $l \in \Edge$ and for all time $t\ge T_3$.

Consider the sub-system that only contains sub-queue 
$q_{l,1}$, and note that $p_{l,1}(t) = \sum_s H^s_{l,1} \lambda_s$, 
then for all $t \ge T_3$, we have $\frac{d}{dt} q_{l,1}(t) = p_{l,1}(t) 
- \pi_{l,1}(t) \le -\epsilon \sum_s H^s_{l,1} \lambda_s < 0$, if 
$q_{l,1}(t) > 0$. This implies that the sub-system that 
consists of $q_{l,1}$ is stable, for all $l \in \Edge$.

\noindent {\bf Induction Step:} 

Next, we show that, if sub-queues $q_{l,j}$ for all $l \in \Edge$ 
and all $j \le k$ is stable, then each sub-queue $q_{l,k+1}$ 
for all $l \in \Edge$ is also stable, along with the stability
of $q_{l,j}$ for all $l \in \Edge$ and all $j \le k$.

Recall that $U_\sk(t)$ is the number of packets transmitted from 
the $(k-1)$-st hop to the $k$-th hop for flow $s$ up to time slot 
$t$, and $u_\sk(t)$ is its fluid limit. Since $q_{l,j}(t)$ is stable 
for all $l \in \Edge$ and all $j \le k$, i.e., there exists a finite 
$T^k_2>0$ such that $q_{l,j}(t) = 0$ for all regular time $t \ge T^k_2$, 
then $u_{s,k+1}(t) = u_\sk(t) + q_\sk(0) = \cdots = u_{s,1}(t) + 
\sum_{h \le k} q_{s,h}(0) = \lambda_s t + \sum_{h \le k}q_{s,h}(0)$ 
for all $s \in \Flow$, for all regular time $t\ge T^k_2$. Thus, for 
all $l \in \Edge$ and for all $j \le k+1$, we have $a_{l,j}(t) = 
t \sum_s H^s_{l,j} \lambda_s + \sum_s H^s_{l,j} \sum_{h < j}q_{s,h}(0)$
from (\ref{eq:lkar}), and $p_{l,j}(t) = \sum_s H^s_{l,j} \lambda_s$ 
from (\ref{eq:lkap}) by taking derivative, for all $l \in \Edge$ and all 
regular time $t\ge T^k_1$.
% Then, for 
% all $l \in \Edge$ and for all $j \le k+1$, we obtain 
% from (\ref{eq:lkar}) and (\ref{eq:lkap}) that 
% \begin{eqnarray}
% && \textstyle a_{l,j}(t) = t \sum_s H^s_{l,j} \lambda_s, \label{eq:lalambda} \\
% && \textstyle p_{l,j}(t) = \sum_s H^s_{l,j} \lambda_s. \label{eq:pljl}
% \end{eqnarray}
Hence, from (\ref{eq:lkdp1}) and the stability of $q_{l,j}$ (i.e., 
$\frac{d}{dt} q_{l,j}(t) = 0$) for all $j \le k$, we have that for all 
$j \le k$,
\begin{equation}
\label{eq:psil}
\textstyle \psi_{l,j}(t)=p_{l,j}(t)=\sum_s H^s_{l,j} \lambda_s.
\end{equation}
Note that since
%\begin{equation}
\[
\textstyle \Expect [\hat{P}_l(t)] = (1+\epsilon) \frac {A_l(t)} {t} 
\ge (1+\epsilon) \frac {\sum_s \sum_{j \le k+1} H^s_{l,j} A_{l,j}(t)} {t},
\]
%\end{equation}
we can obtain that
\begin{equation}
\label{eq:pll}
\textstyle \hat{p}_l(t) \ge (1+\epsilon) \sum_s \sum_{j \le k+1} H^s_{l,j} \lambda_s, 
\end{equation}
following the same line of analysis of Lemma~\ref{lem:p}. Hence, from 
(\ref{eq:pilk}), (\ref{eq:pipip}), (\ref{eq:psil}) and (\ref{eq:pll}), 
we have that for all $j \le k$, 
\begin{equation}
\label{eq:pil}
\textstyle \pi_{l,k+1}(t) \ge (1+\epsilon) \sum_s H^l_{s,k+1} \lambda_s + 
\epsilon \sum_s \sum_{j \le k} H^s_{l,j} \lambda_s.
\end{equation} 
This implies that for all time $t \ge T^k_2$, $\frac{d}{dt} q_{l,k+1}(t) = 
p_{l,k+1}(t) - \pi_{l,k+1}(t) \le -\epsilon \sum_s \sum_{j \le k+1} H^s_{l,j} 
\lambda_s < 0$, if $q_{l,k+1}(t)>0$. Hence, we can conclude that $q_{l,k+1}$ 
is stable for all $l \in \Edge$.

Now by induction, we can show that all the data queues in fluid limits are stable.
With Lemma~\ref{lem:sub-stable-lq}, this implies that the fluid limit model of the 
joint system of data queues and shadow queues is stable. Then, we can conclude 
Proposition~\ref{pro:plq} following the same arguments used in the proof of 
Proposition~\ref{pro:hq}.

%%%%%%%%%%%%%%%%%%%%%%%%%%%%%%%%%%%%%%%%%%%%%%%%%%%%%%%%%%
\section{Proof of Lemma~\ref{lem:mon}} \label{app:lem:mon}
%%%%%%%%%%%%%%%%%%%%%%%%%%%%%%%%%%%%%%%%%%%%%%%%%%%%%%%%%%
Recall that $\Route(s)$ denotes the loop-free route of the flow $s$. We prove 
Lemma~\ref{lem:mon} in a constructive way, i.e., for a network where flows do 
not form loops, we will give an algorithm that generates a ranking such that 
the following statements in Lemma~\ref{lem:mon} hold: 1) \high{for} any flow 
$s \in \Flow$, the ranks are monotonically increasing when one traverses the 
links on the route of the flow $s$ from $l^s_1$ to $l^s_{|\Route(s)|}$, i.e., 
$r(l^s_i) < r(l^s_{i+1})$ for all $1 \le i < |\Route(s)|$; and 2) the packet 
arrivals at a link are either exogenous, or forwarded from links with a smaller 
rank.

We start with some useful definitions. 
\begin{definition}
\label{def:connected}
Two flows $s_1, s_2 \in \Flow$ are \emph{connected}, if they have common (directed) 
links on their routes, i.e., $\Route(s_1) \bigcap \Route(s_2) \neq \emptyset$, and 
\emph{disconnected}, otherwise. A sequence of flows $(\tau_1, \cdots, \tau_n)$ 
is a \emph{communicating sequence}, if every two adjacent flows $\tau_i$ and 
$\tau_{i+1}$ are connected with each other. Two flows $s_1$ and $s_2$ 
\emph{communicate}, if there exists a communicating sequence between $s_1$ and $s_2$.
\end{definition}

\begin{definition}
\label{def:component}
Let $\Flow(l) \subseteq \Flow$ denote the set of flows passing through link 
$l$, and let $\Flow(\Component) \triangleq \bigcup_{l \in \Component} \Flow(l)$ 
denote the set of flows passing through a set of links $\Component \subseteq \Edge$. 
A non-empty set of links $\Component$ is called a \emph{component}, if the following 
conditions are satisfied:
\begin{enumerate}
\item $\Component = \bigcup_{s \in \Flow(\Component)} \Route(s)$.
\item Either $|\Flow(\Component)| = 1$, or any two flows $s_1, s_2 
\in \Flow(\Component)$ communicate.
%\item If $s^{\prime}_1 \in \Flow(\Component)$ and $s^{\prime}_2 
%\notin \Flow(\Component)$, then $s^{\prime}_1$ and $s^{\prime}_2$ do not communicate.
%\item There is no proper subset $\Component^{\prime} \subset \Component$ that satisfies 
%conditions 1)-3).
\end{enumerate}
\end{definition}

\begin{definition}
\label{def:fl}
Consider a component $\Component$, a sequence\footnote{By slightly abusing the notation, 
we also use $(s_1,s_2,\cdots,s_N)$ to denote the set of unique elements of the sequence.} of 
flows $(s_1,s_2,\cdots,s_N) \subseteq \Flow(\Component)$, where $N \ge 2$, is said to form a 
\emph{flow-loop}, if one can find two links $l^{s_n}_{i_n}$ and $l^{s_n}_{j_n}$ for each 
$n=1,2,\cdots,N$, satisfying
\begin{enumerate}
\item $i_n < j_n$ for each $1 \le n \le N$,
\item $\left\{
\begin{array}{ll}
l^{s_n}_{j_n} = l^{s_{n+1}}_{i_{n+1}} ~\text{for each}~ n < N, \\
l^{s_N}_{j_N} = l^{s_1}_{i_1}.
\end{array}
\right.$
\end{enumerate}
\end{definition}

An example of a component that contains a flow-loop is presented in Fig.~\ref{fig:fl},
where the network consists of seven links and six flows. The routes of the flows are 
as follows: $\Route(s_1)=(1,2,3),\Route(s_2)=(3,4),\Route(s_3)=(4,5),\Route(s_4)=(5,6),
\Route(s_5)=(6,7),\Route(s_6)=(7,2)$.

\begin{definition}
A component $\Component$ is called a \emph{flow-tree}, if $\Component$
does not contain any flow-loops.
\end{definition}

\begin{definition}
\label{def:fp}
Consider a component $\Component$, a link $l \in \Component$ is called a 
\emph{starting link}, if there exists a flow $s^{\prime} \in \Flow(\Component)$ 
such that $H^{s^{\prime}}_{l,1}=1$ and $\Hslk=0$ for all other $s \in \Flow(\Component)$ 
and all $k \ge 2$, i.e., a starting link has only exogenous arrivals. Similarly, 
a link $l \in \Component$ is called an \emph{ending link}, if there exists a flow 
$s^{\prime \prime} \in \Flow(\Component)$ such that, $H^{s^{\prime \prime}}_{l,|\Route(s^{\prime 
\prime})|}=1$, and $\Hslk=0$ for all other $s \in \Flow(\Component)$ and all $k < |\Route(s)|$,
i.e., an ending link transmits only packets that will leave the system immediately. 
A path $\Path=( l_{\Path,1}, l_{\Path,2}, \cdots, l_{\Path, len(\Path)} )$, where 
$len(\Path)$ denotes the length of path $\Path$ and $l_{\Path,i}$ denotes the $i$-th 
hop link of $\Path$, is called a \emph{flow-path}, if the following conditions are satisfied:
\begin{enumerate}
\item Links $l_{\Path, 1}$ and $l_{\Path, len(\Path)}$ are the only starting and 
ending link on the path $\Path$, respectively.
\item Either $len(\Path)=1$, or for each $1 \le i < len(\Path)$, there exists a flow 
$s$ such that, $l_{\Path,i} \in \Route(s)$ and $l_{\Path,i+1} \in \Route(s)$, i.e., 
two adjacent links $l_{\Path,i}$ and $l_{\Path,i+1}$ are on the route of some flow.
\end{enumerate}
\end{definition} 

In general, a flow-tree consists of multiple (possibly overlapped) flow-paths.
An illustration of flow-loop, flow-path, and flow-tree is presented in Fig.~\ref{fig:component}.
It is clear from Definition~\ref{def:fl} that, if there exists a flow-loop 
in a component, this component must contain a cycle of links, while the opposite 
is not necessarily true. For example, the components in Figs.~\ref{fig:ft8} and 
\ref{fig:ft9} both contain a cycle, while neither of them contains a flow-loop.

\begin{figure}[t]
\centering
% \subfigure[A component containing a flow-loop]{\epsfig{file=fl.eps,width=0.4\linewidth}\label{fig:fl}}~~~~~~~~~~
% \subfigure[A flow-tree with one flow-path]{\epsfig{file=ft.eps,width=0.4\linewidth}\label{fig:ft8}} \\
% \subfigure[A flow-tree with five flow-paths]{\epsfig{file=rank9.eps,width=0.7\linewidth}\label{fig:ft9}}
\subfigure[A component containing a flow-loop]{\epsfig{file=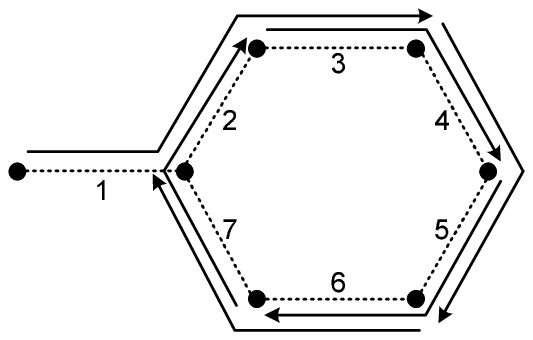,width=0.3\linewidth}\label{fig:fl}}~~~~~~~~~~
\subfigure[A flow-tree with one flow-path]{\epsfig{file=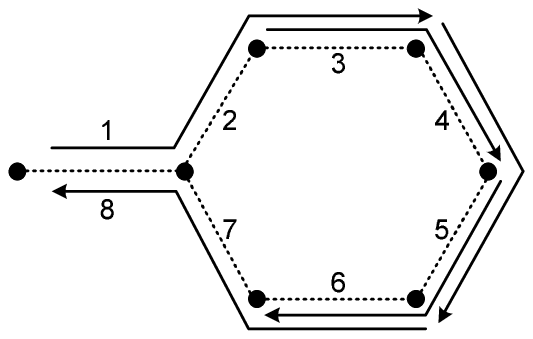,width=0.3\linewidth}\label{fig:ft8}} \\
\subfigure[A flow-tree with five flow-paths]{\epsfig{file=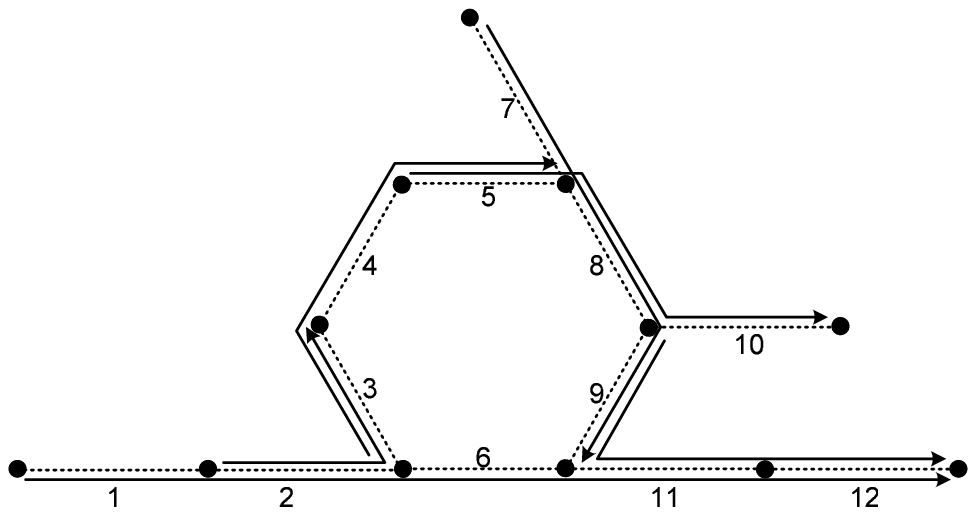,width=0.5\linewidth}\label{fig:ft9}}
\caption{Examples of different types of components. Links and flows are denoted by dashed 
lines with numbers and solid lines with arrows, respectively. 
Note that links without data flows are omitted (not numbered), and two numbers 
labeled beside a dashed line stand for two links with opposite directions, e.g., 
links 1 and 8 in Fig.~\ref{fig:ft8}. 
In Fig.~\ref{fig:fl}, all flows together forms a flow-loop $(2,3,4,5,6,7)$, 
and the component is not a flow-tree. 
In Fig.~\ref{fig:ft8}, the component is a flow-tree and consists of one 
single flow-path: $(1,2,3,4,5,6,7,8)$.
In Fig.~\ref{fig:ft9}, the component is a flow-tree and consists of five 
flow-paths: $\Path_1=(1,2,3,4,5,8,10), \Path_2=(1,2,6,11,12), \Path_3=(7,8,10),
\Path_4=(7,8,9,11,12)$ and $\Path_5=(1,2,3,4,5,8,9,11,12)$.}
\label{fig:component}
\end{figure}

Now, we describe Algorithm~\ref{alg:ra}, which is used to generate a ranking 
for a network without flow-loops such that the monotone property in Lemma~\ref{lem:mon}
holds. 

Let $\Edge(\Path)$ denote the set of links belonging to flow-path $\Path$. Let 
$\Tree$ denote a flow-tree, and let $\Pathset(\Tree)$ denote the set of all 
flow-paths in $\Tree$, i.e., $\Pathset(\Tree) \triangleq \{\Path~\text{is a 
flow-path} ~|~ \Edge(\Path) \subseteq \Tree\}$. Let $\Path_k(\Tree)$ denote the flow-path
chosen in the $k$-th while-loop when running Algorithm~\ref{alg:ra} for $\Tree$,
and let $\Pathset_k(\Tree) \triangleq \bigcup_{j<k} \Path_k(\Tree)$. Let $r(l)$ 
denote the rank of link $l \in \Tree$, and let $\Pathset(l)$ denote the set of 
flow-paths passing through link $l$, i.e., $\Pathset(l) \triangleq \{ \Path \in 
\Pathset(\Tree) ~|~ l \in \Edge(\Path) \}$. Let $\Gamma_k(l) \triangleq \{ l^{\prime} 
\in \bigcup_{\Path \in \Pathset(l) \bigcap \Pathset_k(\Tree)} \Edge(\Path) ~|~ 
r(l^{\prime}) > r(l)\}$ denote the set of links that belong to the flow-paths of 
$\Pathset(l) \bigcap \Pathset_k(\Tree)$ (i.e., flow-paths that pass through link 
$l$ and are chosen in the $j$-th while-loop for $j<k$) and have a rank greater 
than $r(l)$. 

\begin{algorithm}[t]
\caption{Rank Assignment} \label{alg:ra}
\begin{algorithmic}[1]
\Procedure{AssignRank}{$\Tree$}
\State $r(l) \gets -1$ for all $l \in \Tree$ \label{line:init}
\State $\Pathset^{\prime} \gets \Pathset(\Tree)$ 
\While{$\Pathset^{\prime} \neq \emptyset$} \label{line:assign1}
\State pick a flow-path $\Path \in \Pathset^{\prime}$
\State $count \gets 1$ \label{line:count}
\For{$1 \le i \le len(P)$}
\If{$r(l_{\Path,i})=-1$} 
\State $r(l_{\Path,i}) \gets count$ \label{line:case1}
\ElsIf{$r(l_{\Path,i}) \ge count$}
\State $count \gets r(l_{\Path,i})$ \label{line:case2}
\Else
\ForAll{$l \in \Gamma_k(l_{\Path,i})$} \label{line:case31} 
\State $r(l) \gets r(l) + (count - r(l_{\Path,i}))$
\EndFor \label{line:case32}
\State $r(l_{\Path,i}) \gets count$ \label{line:case3}
\EndIf
\State $count \gets count + 1$ \label{line:count++} 
\EndFor
\State $\Pathset^{\prime} \gets \Pathset^{\prime} \backslash \{ \Path \}$ \label{line:pathset}
\EndWhile \label{line:assign2}
\EndProcedure
\end{algorithmic}
\end{algorithm}

The details of ranking are provided in Algorithm 1. In line~\ref{line:init}, 
we do initialization by setting the rank of all links of $\Tree$ to $-1$. 
In lines~\ref{line:assign1}-\ref{line:assign2}, we pick a flow-path $\Path 
\in \Pathset^{\prime}$, and assign a rank to each link of $\Path$ starting 
from link $l_{\Path,1}$. We may update a link's rank if we already assigned 
a rank to that link. The set of flow-paths $\Pathset^{\prime}$ is updated in 
line~\ref{line:pathset}. The while-loop continues until $\Pathset^{\prime}$ 
becomes empty. We set $count=1$ in line~\ref{line:count}, and assign a rank 
to links $l_{\Path,i}$ for each $1 \le i \le len(\Path)$. For each link $l_{\Path,i}$, 
we consider the following three cases: 1) $r(l_{\Path,i})=-1$; 2) $r(l_{\Path,i}) 
\ge count$; 3) $0 < r(l_{\Path,i}) < count$. 

\noindent {\bf Case 1):} link $l_{\Path,i}$ has not been assigned a rank
yet. We set $r(l_{\Path,i}) = count$ in line~\ref{line:case1}.

\noindent {\bf Case 2):} link $l_{\Path,i}$ already has a rank that is no 
smaller than the current $count$. In this case, the rank does not need an 
update, and we set $count = r(l_{\Path,i})$ in line~\ref{line:case2}.

\noindent {\bf Case 3):} link $l_{\Path,i}$ already has a rank that is smaller 
than the current $count$. In this case, we update the rank of some other links
as well as that of link $l_{\Path,i}$. Specifically, for all the links $l \in 
\Gamma_k(l_{\Path,i})$, i.e., links that belong to the flow-paths in $\Pathset(l) 
\bigcap \Pathset_k(\Tree)$ and have a rank greater than $r(l_{\Path,i})$, we 
increase their ranks by $count-r(l_{\Path,i})$ in lines~\ref{line:case31}-\ref{line:case32}.
Then, we update the rank of link $l_{\Path,i}$ by setting it to $count$ in 
line~\ref{line:case3}.

After considering all three cases, we increase the value of $count$ 
by 1 in line~\ref{line:count++}. 

The intention of this ranking is to assign a rank to each link such that the ranks 
are monotonically increasing when one traverses any flow-path from its starting
link. Algorithm~\ref{alg:ra} may give different ranking to a given flow-tree 
depending on the order of choosing flow-paths. We give two examples for illustration 
as follows. In Fig.~\ref{fig:ft8}, one (and the unique one in this case) example 
of the ranking for the flow-tree is $(1,2,3,4,5,6,7,8)$ for links 1-8. In 
Fig.~\ref{fig:ft9}, one example of the ranking for the flow-tree is $(1,2,3,
4,5,3,1,6,7,7,8,9)$ for links 1-12. The evolution of the ranking for the 
flow-tree in Fig.~\ref{fig:ft9} is presented in Table~\ref{tab:rank}, where 
flow-path $\Path_i$ is chosen in the $i$-th while-loop, for $i=1,2,3,4,5$.

\begin{table}[!t]
\caption{The evolution of the ranking for the flow-tree in Fig.~\ref{fig:ft9}}
\label{tab:rank}
\centering
\begin{tabular}{|c|r@{}c@{,}c@{,}c@{,}c@{,}c@{,}c@{,}c@{,}c@{,}c@{,}c@{,}c@{,}c@{}l|}
\hline
\multicolumn{1}{|c|}{Iteration $k$} & \multicolumn{14}{|c|}{Ranking of links~$1-12$} \\
\hline
%\hline
0 &  ( & -1 & -1 & -1 & -1 & -1 & -1 & -1 & -1 & -1 & -1 & -1 & -1 & ) \\
1 &  ( &  1 &  2 &  3 &  4 & 5  & -1 & -1 &  6 & -1 &  7 & -1 & -1 & ) \\
2 &  ( &  1 &  2 &  3 &  4 & 5  &  3 & -1 &  6 & -1 &  7 &  4 &  5 & ) \\
3 &  ( &  1 &  2 &  3 &  4 & 5  &  3 &  1 &  6 & -1 &  7 &  4 &  5 & ) \\
4 &  ( &  1 &  2 &  3 &  4 & 5  &  3 &  1 &  6 &  7 &  7 &  8 &  9 & ) \\
5 &  ( &  1 &  2 &  3 &  4 & 5  &  3 &  1 &  6 &  7 &  7 &  8 &  9 & ) \\
\hline
\end{tabular}
\end{table}

Since we assume $\sum_s \sum^{|\Route(s)|}_{k=1}\Hslk \ge 1$ for all 
$l \in \Edge$, a network graph $\Graph$ can be decomposed into multiple
disjoint components. Clearly, a network with no flow-loops is equivalent
to that all the components of the network are flow-trees. Without loss
of generality, in the rest of the proof, we assume that the network that
we consider consists of one single component, which is a flow-tree under 
the condition of Lemma~\ref{lem:mon}. The same argument applies to the 
case with multiple disjoint components. We claim the following lemma and 
provide its proof in Appendix~\ref{app:lem:alg}.
\begin{lemma}
\label{lem:alg}
Algorithm~\ref{alg:ra} assigns a rank to each link of a flow-tree $\Tree$
such that for any flow-path $\Path \in \Pathset(\Tree)$, the ranks are 
monotonically increasing when one traverses the links of $\Path$ from
$l_{\Path,1}$ to $l_{\Path,len(\Path)}$, i.e., $r(l_{\Path,i}) < r(l_{\Path,i+1})$
for all $1 \le i < len(\Path)$ and for any $\Path \in \Pathset(\Tree)$.
\end{lemma}

Now, consider any flow $s \in \Flow$. The statement 1) holds trivially 
for the case of $|\Route(s)|=1$. Hence, we assume that $|\Route(s)| > 1$. 
It is clear that for any $1 \le i < |\Route(s)|$, the links $l^s_i$ and 
$l^s_{i+1}$ must belong to some flow-path $\Path \in \Pathset(\Edge)$, 
where $\Edge$ is assumed to be a flow-tree. Therefore, the statement 1)
follows from Lemma~\ref{lem:alg}.

Note that the packet arrivals at a link are either exogenous or from the 
previous hop on the route of some flow passing through it. Owing to the 
monotonically increasing rank assignment, it is clear that these previous 
hop links have a smaller rank. Hence, the statement 2) immediately follows 
from statement 1). This completes the proof of Lemma~\ref{lem:mon}.

%%%%%%%%%%%%%%%%%%%%%%%%%%%%%%%%%%%%%%%%%%%%%%%%%%%%%%%%%%
\section{Proof of Lemma~\ref{lem:alg}} \label{app:lem:alg}
%%%%%%%%%%%%%%%%%%%%%%%%%%%%%%%%%%%%%%%%%%%%%%%%%%%%%%%%%%
We want to show that Algorithm~\ref{alg:ra} assigns a rank to 
each link of flow-tree $\Tree$ satisfying that $r(l_{\Path,i}) 
< r(l_{\Path,i+1})$, for all $1 \le i < len(\Path)$ and for any 
$\Path \in \Pathset(\Tree)$. We use the method of induction. 

Recall that $P_k(\Tree)$ denotes the flow-path chosen in the 
$k$-th while-loop, and $\Pathset_k(\Tree) = \bigcup_{j<k} 
\Path_k(\Tree)$. We denote $P_k(\Tree)$ and $\Pathset_k(\Tree)$ 
by $P_k$ and $\Pathset_k$, respectively, whenever there is no
confusion.

\noindent {\bf Base Case:} 

It is trivial for the case of $k=1$. Since we initialize $r(l_{\Path_1,i})
= -1$ for all $1 \le i \le len(\Path_1)$, we should have $r(l_{\Path_1,i})
= i$ for all $1 \le i \le len(\Path_1)$ from lines~\ref{line:case1} and 
\ref{line:count++} of Algorithm~\ref{alg:ra}, after running the first 
while-loop.

\noindent {\bf Induction Step:} 

We show that after running the $k$-th while-loop of Algorithm~\ref{alg:ra}, if 
\begin{equation}
\label{eq:assump}
r(l_{\Path_j,i}) < r(l_{\Path_j,i+1}) ~\text{for all}~ 1 \le i < len(\Path_j) 
~\text{and for all}~ j \le k, 
\end{equation}
then after running the $(k+1)$-st while-loop the same result holds for all 
$j \le k+1$. In other words, once Algorithm~\ref{alg:ra} assigns the ranks 
for links of a flow-path in a monotonically increasing way, then this property 
does not change afterward. We also prove this induction step using method of 
induction.

We first show that if (\ref{eq:assump}) holds, then after the first iteration 
(for assigning a rank to link $l_{\Path_{k+1},1}$) of the $(k+1)$-st while-loop, 
(\ref{eq:assump}) still holds. When we start the $(k+1)$-st while-loop, we have 
$count = 1$, and $r(l_{\Path_{k+1},1})$ must be in one of the following two cases: 
1) $r(l_{\Path_{k+1},1})=-1$ if the rank of link $l_{\Path_{k+1},1}$ is not assigned 
yet, or 2) $r(l_{\Path_{k+1},1}) \ge count$, otherwise. Then, Algorithm~\ref{alg:ra} will 
assign a rank of 1 to link $l_{\Path_{k+1},1}$ in the former case (line~\ref{line:case1}), 
or will not change its rank in the latter case (line~\ref{line:case2}). Hence, 
(\ref{eq:assump}) still holds.

Now suppose that after assigning the ranks of links up to link $l_{\Path_{k+1},n}$,
which is the $n$-th hop of the flow-path chosen in the $(k+1)$-st while-loop, we have 
$r(l_{\Path_{k+1},m-1}) < r(l_{\Path_{k+1},m})$ for all $1 < m \le n$, and (\ref{eq:assump})
holds. Then we want to show that after assigning a rank to the next hop $l_{\Path_{k+1}, 
n+1}$, we still have both $r(l_{\Path_{k+1},m-1}) < r(l_{\Path_{k+1},m})$ for all $m \le 
n+1$, and (\ref{eq:assump}). We show this when $n=2$ for ease of presentation. One can 
easily extend the analysis to the case when $n \ge 2$. After assigning a rank to link 
$l_{\Path_{k+1},1}$, we have $count = r(l_{\Path_{k+1},1}) + 1$ from line~\ref{line:count++}
of Algorithm~\ref{alg:ra}. At this moment, the rank of link $l_{\Path_{k+1},2}$ is either 
1) $r(l_{\Path_{k+1},2})=-1$, 2) $r(l_{\Path_{k+1},2}) \ge count$, or 3) 
$0 < r(l_{\Path_{k+1},2}) < count$. We discuss the three cases as follows.

\noindent {\bf Case 1): $r(l_{\Path_{k+1},2})=-1$.}

In this case, since Algorithm~\ref{alg:ra} sets $r(l_{\Path_{k+1},2})$ 
to $count$ from line~\ref{line:case1}, we have $r(l_{\Path_{k+1},2}) > 
r(l_{\Path_{k+1},1})$. The rank of links of $\Path_j$ for all $j \le k$ 
is not changed, and (\ref{eq:assump}) still holds.

\noindent {\bf Case 2): $r(l_{\Path_{k+1},2}) \ge count$.}

In this case, since Algorithm~\ref{alg:ra} does not change the rank 
$r(l_{\Path_{k+1},2})$, we have $r(l_{\Path_{k+1},2}) \ge count > 
r(l_{\Path_{k+1},1})$. The rank of links of $\Path_j$ for all 
$j \le k$ is not changed, and (\ref{eq:assump}) still holds.

\noindent {\bf Case 3): $0 < r(l_{\Path_{k+1},2}) < count$.}

Note that in this case, we have $r(l_{\Path_{k+1},1}) \ge r(l_{\Path_{k+1},2})$ 
before assigning a new rank to link $l_{\Path_{k+1},2}$. Since Algorithm~\ref{alg:ra}
sets $r(l_{\Path_{k+1},2})$ to $count$ in line~\ref{line:case3}, we will have 
$r(l_{\Path_{k+1},2}) > r(l_{\Path_{k+1},1}) = count-1$. Now what remains to show 
is that after the rank update for links of $\Gamma_k(l_{\Path_{k+1},2})$ in 
lines~\ref{line:case31}-\ref{line:case32}, we still have $r(l_{\Path_{k+1},2}) 
> r(l_{\Path_{k+1},1})$ and (\ref{eq:assump}) still holds.

Recall that $\Gamma_k(l) = \{ l^{\prime} \in \bigcup_{\Path \in \Pathset(l) \bigcap 
\Pathset_k} \Edge(\Path) ~|~ r(l^{\prime}) > r(l)\}$ denotes the set of links 
that belong to the flow-paths of $\Pathset(l) \bigcap \Pathset_k(\Tree)$ (i.e., 
flow-paths that pass through link $l$ and are chosen in the $j$-th while-loop for $j<k$)
and have a rank greater than $r(l)$. Let $\Omega \triangleq \Gamma_{k+1}(l_{\Path_{k+1},2}) 
\bigcup \{l_{\Path_{k+1},2}\}$ denote the union of $\Gamma_{k+1}(l_{\Path_{k+1},2})$
and $\{l_{\Path_{k+1},2}\}$. Algorithm~\ref{alg:ra} updates only the rank of the links 
in $\Omega$ by adding the rank with $count - r(l_{\Path_{k+1},2})$. We claim that 
$l_{\Path_{k+1},1} \notin \Omega$, i.e., the rank $r(l_{\Path_{k+1},1})$ is not 
changed after the update, which implies that $r(l_{\Path_{k+1},2}) > r(l_{\Path_{k+1},1})$
still holds after the update. We prove this claim by contradiction. Suppose that 
$l_{\Path_{k+1},1} \in \Omega$, then there exists a flow-path $\Path^{\prime} \in 
\Pathset(l_{\Path_{k+1},2}) \bigcap \Pathset_{k+1}$ such that $l_{\Path_{k+1},1}, 
l_{\Path_{k+1},2} \in \Edge(\Path^{\prime})$ and link $l_{\Path_{k+1},2}$ appears 
earlier than $l_{\Path_{k+1},1}$ on the flow-path $\Path^{\prime}$. This implies 
that flow-paths $\Path^{\prime}$ and $\Path_{k+1}$ form a flow-loop, which contradicts 
with the definition of flow-tree. 

Next, we want to show that (\ref{eq:assump}) still holds after the rank update. Note that
before the rank update, due to (\ref{eq:assump}), two adjacent links $l_{\Path_j,i}$ and 
$l_{\Path_j,i+1}$ satisfy that $r(l_{\Path_j,i}) < r(l_{\Path_j,i+1})$ for any $j \le k$ 
and any $i<len(\Path_j)$. We want to show that, after the rank update, we still have 
$r(l_{\Path_j,i}) < r(l_{\Path_j,i+1})$. We consider the following four cases for two 
adjacent links $l_{\Path_j,i}$ and $l_{\Path_j,i+1}$.

\noindent {\bf Case i): $l_{\Path_j,i} \in \Omega$ and $l_{\Path_j,i+1} \in \Omega$.}

In this case, since Algorithm~\ref{alg:ra} increases the rank of links 
$l_{\Path_j,i}$ and $l_{\Path_j,i+1}$ by $count - r(l_{\Path_{k+1},2})$, 
we still have $r(l_{\Path_j,i}) < r(l_{\Path_j,i+1})$ after the update.
 
\noindent {\bf Case ii): $l_{\Path_j,i} \notin \Omega$ and $l_{\Path_j,i+1} \notin \Omega$.}

In this case, since Algorithm~\ref{alg:ra} does not change the rank of 
links $l_{\Path_j,i}$ and $l_{\Path_j,i+1}$, we still have $r(l_{\Path_j,i}) 
< r(l_{\Path_j,i+1})$ after the update.

\noindent {\bf Case iii): $l_{\Path_j,i} \notin \Omega$ and $l_{\Path_j,i+1} \in \Omega$.}

In this case, since Algorithm~\ref{alg:ra} increases the rank of link $l_{\Path_j,i+1}$ 
by $count - r(l_{\Path_{k+1},2})$ and does not change the rank of links $l_{\Path_j,i}$,
we still have $r(l_{\Path_j,i}) < r(l_{\Path_j,i+1})$ after the update.

\noindent {\bf Case iv): $l_{\Path_j,i} \in \Omega$ and $l_{\Path_j,i+1} \notin \Omega$.}

This is an infeasible case from the definition of $\Omega$ and (\ref{eq:assump}) of 
the previous step. Note that since links $l_{\Path_j,i}$ and $l_{\Path_j,i+1}$ are 
two adjacent links on the flow-path $\Path_j$, there exists a flow $s$ such that 
$l_{\Path_j,i}, l_{\Path_j,i+1} \in \Route(s)$ from the definition of flow-path 
(Definition~\ref{def:fp}), we should have $r(l_{\Path_j,i}) < r(l_{\Path_j,i+1})$ 
before the rank update. Hence if $l_{\Path_j,i} \in \Omega$, we should have 
$l_{\Path_j,i+1} \in \Omega$ from the definition of $\Omega$.

We can show the property of monotonically increasing ranking for Case 3) 
by combining sub-cases i), ii), iii) and iv). Results for Cases 1), 2) 
and 3) complete the induction step when $n=2$. One can easily extends 
the analysis to the case when $n \ge 2$, and this completes the proof.

%%%%%%%%%%%%%%%%%%%%%%%%%%%%%%%%%%%%%%%%%%%%%%%%%%%%%%%%%%%%%%
\section{Proof of Proposition~\ref{pro:flq}} \label{app:pro:flq}
%%%%%%%%%%%%%%%%%%%%%%%%%%%%%%%%%%%%%%%%%%%%%%%%%%%%%%%%%%%%%%
We want to show that, a network where flows do not form loops,
i.e., all the components are flow-trees, is stable under FLQ-MWS 
for any traffic with arrival rate vector that is strictly inside 
$\Lambda^{*}$.

We know from Lemma~\ref{lem:mon} that, there exists a ranking $R(\Edge)$
such that the monotone property holds. Without loss of generality, we
assume that the minimum rank is 1, and use $r(\Edge) \triangleq \max_{l 
\in \Edge} r(l)$ to denote the maximum rank among all the links. We give 
the following definitions that are used in the proof.

\begin{definition}
\label{def:depth}
We divide $\Edge$ into $r(\Edge)$ disjoint subsets: $R_k \triangleq 
\{l \in \Edge ~|~ r(l)=k \}$, for $1 \le k \le r(\Edge)$. Then $R_k$ 
is called the \emph{depth-$k$ set}, and a link $l_k \in R_k$ is called 
a \emph{depth-$k$ link}. 
\end{definition}

Recall that the fluid limit model for the sub-system consisting of shadow 
queues is stable from Lemma~\ref{lem:sub-stable-lq}. We show by induction 
that all data queues are stable.

\noindent {\bf Base Case:} 

First, Lemma~\ref{lem:mon} implies that for any $l_1 \in R_1$, its 
arrivals are exogenous, i.e., $A_{l_1}(t) = \sum_s H^s_{l_1,1} F_s(t)$. 
Following the same line of analysis for the proof of Proposition~\ref{pro:hq}, 
we can show that $\pi_{l_1}(t) \ge (1+\epsilon) \sum_{s} H^s_{l_1,1} \lambda_s$ and
$p_{l_1}(t) = \sum_{s} H^s_{l_1,1} \lambda_s$, then $\frac{d}{dt} q_{l_1}(t) 
= p_{l_1}(t) - \pi_{l_1}(t) \le -\epsilon \sum_{s} H^s_{l_1,1} \lambda_s < 0$, 
if $q_{l_1}(t) > 0$. This implies that $q_{l_1}(t)$ is stable, for all 
$l_1 \in R_1$.

\noindent {\bf Induction Step:} 

Next, we show that, if $q_l$ is stable for all $l \in \bigcup_{j \le k} R_j$, 
then $q_{l_{k+1}}$ is also stable for all $l_{k+1} \in R_{k+1}$, along with
the stability of all $q_l$, for $1 \le k < K$.

Lemma~\ref{lem:mon} implies that for any $l_{k+1} \in R_{k+1}$, its arrivals 
are either exogenous or from certain links of $\bigcup_{j \le k} R_j$. Since $q_l$ 
is stable for all $l \in \bigcup_{j \le k} R_j$, following the same line of analysis 
for the proof of Proposition~\ref{pro:hq}, we can show that there exists a finite 
time $T^k_3 > 0$ such that, for all time $t \ge T^k_3$, we have $\pi_{l_{k+1}}(t) 
\ge (1+\epsilon) \sum_{s: l_{k+1} \in \Route(s)} \lambda_s$ and $p_{l_{k+1}}(t) = 
\sum_{s: l_{k+1} \in \Route(s)} \lambda_s$. Therefore, for all time $t \ge T^k_3$,
we have $\frac{d}{dt} q_{l_{k+1}}(t) = p_{l_{k+1}}(t) - \pi_{l_{k+1}}(t) \le -\epsilon 
\sum_{s: l_{k+1} \in \Route(s)} \lambda_s < 0$, if $q_{l_{k+1}}(t) > 0$. This implies 
that $q_{l_{k+1}}$ is stable for all $l_{k+1} \in R_{k+1}$.

Therefore, the fluid limit model for the sub-sytem of data queues is stable 
from the induction. With Lemma~\ref{lem:sub-stable-lq}, this implies that 
the fluid limit model of the joint system of data queues and shadow queues 
is stable. Then, we complete the proof following the same arguments used in 
the proof of Proposition~\ref{pro:hq}.

%%%%%%%%%%%%%%%%%%%%%%%%%%%%%%%%%%%%%%%%%%%%%%%%%%%%%%%%%%%%%%%%%%
\section{Proof of Lemma~\ref{lem:lq-csma}} \label{app:lem:lq-csma}
%%%%%%%%%%%%%%%%%%%%%%%%%%%%%%%%%%%%%%%%%%%%%%%%%%%%%%%%%%%%%%%%%%
Given any $\gamma \in (0,1)$, suppose that $\lambda$ is strictly inside 
$(1-\gamma) \Lambda^{*}$ , then there exists a sufficiently small 
$\epsilon > 0$ such that $(1+\epsilon)\lambda$ is strictly inside 
$(1-\gamma) \Lambda^{*}$, and we can find a vector $\phi \in (1-\gamma)
Co(\Matching)$ such that $(1+\epsilon)\lambda < \phi$, 
i.e., $(1+\epsilon)\sum_s \sum_k \Hslk \lambda_s < \phi_l$ for all 
$l\in\Edge$. Let $\beta \triangleq \min_{l\in\Edge} (\phi_l - (1+\epsilon) 
\sum_s \sum_k \Hslk \lambda_s)$. By definition, we have $\beta > 0$. 
Let $T^{\prime}$ be a finite time such that $T^{\prime} > \frac {(1+\epsilon)} 
{\beta}$. Then, for all regular time $t \ge T^{\prime}$, we have
\begin{equation}
\label{eq:llp}
\hat{p}_l(t) \le (1+\epsilon) \left(\sum_s \sum_k \Hslk \lambda_s 
+ \frac {1} {t} \right) < \phi_l,
\end{equation}
from Lemma~\ref{lem:lple}.
This implies that the instantaneous arrival rate of shadow queues is strictly 
inside $(1-\gamma)$ fraction of the optimal throughput region $\Lambda^*$.

Let $W_l(\hat{q}_l) \triangleq \int^{\hat{q}_l}_0 g_l(y) dy$ and consider a 
Lyapunov function $\hat{V}(\hat{q}(t)) \triangleq \sum_l W_l(\hat{q}_l(t))$. 
\high{It is sufficient to show that for any $\zeta_1 > 0$, there exists a
$\zeta_2 > 0$ such that $\hat{V}(\hat{q}(t)) \ge \zeta_1$ implies 
$\frac{D^+}{dt^+} \hat{V}(\hat{q}(t)) \le -\zeta_2$, for any regular time 
$t \ge T^{\prime}$.} Since $W_l(\hat{q}_l)$'s and $\hat{q}_l$'s are 
differentiable, for any regular time $t \ge T^{\prime}$, we can obtain the 
derivative of $\hat{V}(\hat{q}(t))$ as
\begin{equation}
\label{eq:ldol}
\begin{split}
\frac{D^+}{dt^+} \hat{V}(\hat{q}(t)) =& \sum_{l \in \Edge} g_l(\hat{q}_l(t)) \cdot \frac {d} {dt} \hat{q}_l(t)
= \sum_{l \in \Edge} g_l(\hat{q}_l(t)) \cdot \left(\hat{p}_l(t) - \hat{\pi}_l(t)\right)  \\
=& \sum_{l \in \Edge} g_l(\hat{q}_l(t)) \cdot \left(\hat{p}_l(t) - \phi_l\right)
+ \sum_{l \in \Edge} g_l(\hat{q}_l(t)) \cdot \left(\phi_l - \hat{\pi}_l(t)\right).
\end{split}
\end{equation}

Let us choose $\zeta_3>0$ such that $\hat{V}(\hat{q}(t)) \ge \zeta_1$ implies 
$\max_l \hat{q}_l(t) \ge \zeta_3$. Then following a similar argument as in the 
proof of Lemma~\ref{lem:sub-stable}, for the final result of (\ref{eq:dol}), 
we can conclude that the first term is bounded as follows:
\[
\sum_{l \in \Edge} g_l(\hat{q}_l(t)) \cdot \left(\hat{p}_l(t) - \phi_l\right)
\le -\zeta_2 < 0,
\]
and that the second term becomes non-positive due to the following. We first note 
that $\|\hat{q}(t)\| 
> 0$ from $\hat{V}(\hat{q}(t)) > 0$. Then at time slots $\Upsilon \triangleq \{ 
\lceil \xnj t \rceil, \lceil \xnj t \rceil + 1, \cdots, \lfloor \xnj (t+\delta) 
\rfloor \}$, for any $Q_B>0$, we have $\|\hat{Q}(\tau)\| \ge Q_B$ for all time 
slots $\tau \in \Upsilon$ with large enough $j$ and small enough $\delta$. From 
Lemma~\ref{lem:weight}, given any $\theta \in (0,1)$, for all time slots $\tau \in 
\Upsilon$, with probability greater than $1-\theta$, LQ-CSMA chooses a schedule 
$M(\tau) \in \Matching$ that satisfies
\begin{equation}
%\label{eq:gweight}
\sum_{l \in \Edge} g_l(\hat{Q}_l(\tau)) \cdot M_l(\tau) \ge (1-\gamma) \max_{M \in \Matching}
\sum_{l \in \Edge} g_l(\hat{Q}_l(\tau)) \cdot M_l.
\end{equation}
Hence, similar as in Chapter 4 of \cite{shah04}, from condition (\ref{eq:gcond}),
with probability greater than $1-\theta$, the fluid limit $\hat{\pi}(t)$ under 
LQ-CSMA satisfies
\begin{equation}
\begin{split}
\sum_{l \in \Edge} g_l(\hat{q}_l(t)) \cdot \hat{\pi}_l(t) &\ge (1-\gamma) \max_{\phi^{\prime} 
\in Co(\Matching)} \sum_{l \in \Edge} g_l(\hat{q}_l(t)) \cdot \phi^{\prime}_l \\
&= \max_{\phi \in (1-\gamma) Co(\Matching)} \sum_{l \in \Edge} g_l(\hat{q}_l(t)) \cdot \phi_l.
\end{split}
\end{equation}

Therefore, $\hat{V}(\hat{q}(t)) \ge \zeta_1$ implies $\frac{D^+}{dt^+} \hat{V}(\hat{q}(t)) 
\le -\zeta_2$. This completes the proof.

\bibliographystyle{IEEEtran}
\bibliography{jibo}

\begin{thebibliography}{10}
\providecommand{\url}[1]{#1}
\csname url@samestyle\endcsname
\providecommand{\newblock}{\relax}
\providecommand{\bibinfo}[2]{#2}
\providecommand{\BIBentrySTDinterwordspacing}{\spaceskip=0pt\relax}
\providecommand{\BIBentryALTinterwordstretchfactor}{4}
\providecommand{\BIBentryALTinterwordspacing}{\spaceskip=\fontdimen2\font plus
\BIBentryALTinterwordstretchfactor\fontdimen3\font minus
  \fontdimen4\font\relax}
\providecommand{\BIBforeignlanguage}[2]{{%
\expandafter\ifx\csname l@#1\endcsname\relax
\typeout{** WARNING: IEEEtran.bst: No hyphenation pattern has been}%
\typeout{** loaded for the language `#1'. Using the pattern for}%
\typeout{** the default language instead.}%
\else
\language=\csname l@#1\endcsname
\fi
#2}}
\providecommand{\BIBdecl}{\relax}
\BIBdecl

\bibitem{tassiulas92}
L.~Tassiulas and A.~Ephremides, ``{Stability properties of constrained queueing
  systems and scheduling policies for maximum throughput in multihop radio
  networks},'' \emph{IEEE Transactions on Automatic Control}, vol.~37, no.~12,
  pp. 1936--1948, 1992.

\bibitem{sharma06}
G.~Sharma, R.~R. Mazumdar, and N.~B. Shroff, ``{On the complexity of scheduling
  in wireless networks},'' in \emph{Proceedings of the annual international
  conference on Mobile computing and networking (MobiCom)}.\hskip 1em plus
  0.5em minus 0.4em\relax ACM New York, NY, USA, 2006, pp. 227--238.

\bibitem{jiang10}
L.~Jiang and J.~Walrand, ``{A distributed CSMA algorithm for throughput and
  utility maximization in wireless networks},'' \emph{IEEE/ACM Transactions on
  Networking}, vol.~18, no.~3, pp. 960--972, 2010.

\bibitem{ni09}
\BIBentryALTinterwordspacing
J.~Ni, B.~Tan, and R.~Srikant, ``{Q-CSMA: Queue-length based csma/ca algorithms
  for achieving maximum throughput and low delay in wireless networks},''
  \emph{Arxiv preprint arXiv:0901.2333}, 2009. [Online]. Available:
  \url{http://arxiv.org/PS_cache/arxiv/pdf/0901/0901.2333v4.pdf}
\BIBentrySTDinterwordspacing

\bibitem{rajagopalan09}
S.~Rajagopalan, D.~Shah, and J.~Shin, ``{Network adiabatic theorem: an
  efficient randomized protocol for contention resolution},'' in \emph{The ACM
  International Conference on Measurement and Modeling of Computer Systems
  (SIGMETRICS)}, 2009, pp. 133--144.

\bibitem{bui11}
L.~Bui, R.~Srikant, and A.~Stolyar, ``{A Novel Architecture for Reduction of
  Delay and Queueing Structure Complexity in the Back-Pressure Algorithm},''
  \emph{IEEE/ACM Transactions on Networking}, vol.~19, no.~6, pp. 1597--1609,
  December 2011.

\bibitem{stolyar11}
A.~Stolyar, ``{Large number of queues in tandem: Scaling properties under
  back-pressure algorithm},'' \emph{Queueing Systems}, vol.~67, no.~2, pp.
  111--126, 2011.

\bibitem{liu10c}
S.~Liu, E.~Ekici, and L.~Ying, ``{Scheduling in Multihop Wireless Networks
  without Back-pressure},'' in \emph{Proceedings of the Annual Conference on
  Communication, Control and Computing (Allerton)}, 2010.

\bibitem{ying08}
L.~Ying, R.~Srikant, and D.~Towsley, ``{Cluster-based back-pressure routing
  algorithm},'' in \emph{The IEEE International Conference on Computer
  Communications (INFOCOM)}, 2008, pp. 484--492.

\bibitem{ying09}
L.~Ying, S.~Shakkottai, and A.~Reddy, ``{On combining shortest-path and
  back-pressure routing over multihop wireless networks},'' in \emph{The IEEE
  International Conference on Computer Communications (INFOCOM)}, 2009, pp.
  1674--1682.

\bibitem{wu07}
X.~Wu, R.~Srikant, and J.~Perkins, ``{Scheduling Efficiency of Distributed
  Greedy Scheduling Algorithms in Wireless Networks},'' \emph{IEEE Transactions
  on Mobile Computing}, pp. 595--605, 2007.

\bibitem{rfc2453}
IETF, ``Rfc 2453,'' 1998.

\bibitem{rfc791}
------, ``Rfc 791,'' 1981.

\bibitem{internet11}
\BIBentryALTinterwordspacing
``Internet world stats,'' March 2011. [Online]. Available:
  \url{http://www.internetworldstats.com/stats.htm}
\BIBentrySTDinterwordspacing

\bibitem{andrews04}
M.~Andrews, K.~Kumaran, K.~Ramanan, A.~Stolyar, R.~Vijayakumar, and P.~Whiting,
  \emph{{Scheduling in a queuing system with asynchronously varying service
  rates}}.\hskip 1em plus 0.5em minus 0.4em\relax Cambridge Univ Press, 2004,
  vol.~18.

\bibitem{bramson08}
M.~Bramson, ``{Stability of queueing networks},'' \emph{Probability Surveys},
  vol.~5, no.~1, pp. 169--345, 2008.

\bibitem{dai95}
J.~Dai, ``{On positive Harris recurrence of multiclass queueing networks: a
  unified approach via fluid limit models},'' \emph{The Annals of Applied
  Probability}, pp. 49--77, 1995.

\bibitem{rybko92}
A.~Rybko and A.~Stolyar, ``{Ergodicity of stochastic processes describing the
  operation of open queueing networks},'' \emph{Problems of Information
  Transmission}, vol.~28, pp. 199--220, 1992.

\bibitem{lu91}
S.~Lu and P.~Kumar, ``{Distributed scheduling based on due dates and buffer
  priorities},'' \emph{IEEE Transactions on Automatic Control}, vol.~36,
  no.~12, pp. 1406--1416, 1991.

\bibitem{chen00}
H.~Chen and H.~Zhang, ``{Stability of multiclass queueing networks under
  priority service disciplines},'' \emph{Operations Research}, pp. 26--37,
  2000.

\bibitem{chen02}
H.~Chen and H.~Ye, ``{Piecewise linear Lyapunov function for the stability of
  multiclass priority fluid networks},'' \emph{IEEE Transactions on Automatic
  Control}, vol.~47, no.~4, pp. 564--575, 2002.

\bibitem{chen96}
H.~Chen and D.~Yao, ``{Stable priority disciplines for multiclass networks},''
  in \emph{Stochastic Networks: Stability and Rare Events}, P.~Glasserman,
  K.~Sigman, and D.~Yao, Eds.\hskip 1em plus 0.5em minus 0.4em\relax
  Springer-Verlag, 1996, ch.~2, pp. 27--39.

\bibitem{eryilmaz05}
A.~Eryilmaz, R.~Srikant, and J.~Perkins, ``{Stable scheduling policies for
  fading wireless channels},'' \emph{IEEE/ACM Transactions on Networking},
  vol.~13, no.~2, pp. 411--424, 2005.

\bibitem{shah04}
D.~Shah, ``{Randomization and heavy traffic theory: new approaches to the
  design and analysis of switch algorithms},'' Ph.D. dissertation, Stanford
  University, 2004.

\bibitem{ji11}
B.~Ji, C.~Joo, and N.~B. Shroff, ``{Delay-Based Back-Pressure Scheduling in
  Multi-Hop Wireless Networks},'' in \emph{The IEEE International Conference on
  Computer Communications (INFOCOM)}, 2011, pp. 2579--2587.

\bibitem{malyshev79}
V.~Malyshev and M.~Menshikov, ``{Ergodicity, continuity and analyticity of
  countable Markov chains},'' \emph{Transactions of the Moscow Mathematical
  Society}, vol.~39, pp. 3--48, 1979.

\bibitem{dai00}
J.~Dai and B.~Prabhakar, ``{The throughput of data switches with and without
  speedup},'' in \emph{The IEEE International Conference on Computer
  Communications (INFOCOM)}, 2000, pp. 556--564.

\end{thebibliography}

\end{document}